\documentclass{llncs}

\usepackage{amsmath,amssymb,amsfonts,amsopn}
\usepackage{latexsym,textcomp,xspace,paralist,todonotes}
\usepackage{booktabs,tabularx,multirow}
\usepackage[pdfpagelabels,colorlinks,allcolors=blue]{hyperref}

\graphicspath{{figures/}}
\usepackage{color,xcolor,graphicx,colortbl}
\usepackage[compatibility=false,font=small]{caption}
\usepackage[font=small]{subcaption}

\newcommand{\myparagraph}[1]{\smallskip\noindent\textbf{#1}}
\newcommand{\rephrase}[3]{\smallskip\noindent\textbf{#1 #2}.~\emph{#3}\smallskip}

\let\doendproof\endproof
\renewcommand\endproof{~\hfill$\qed$\doendproof}

\newtheorem{chr}[theorem]{Characterization} 
\newtheorem{obs}[theorem]{Observation} 

\usepackage{ifthen}
\newcommand{\ver}{arxiv}
\newcommand{\arxapp}[2]{\ifthenelse{\equal{\ver}{arxiv}}{#2}{#1}}

\begin{document}
\title{Efficient~Generation~of~Different~Topological Representations of Graphs Beyond-Planarity\thanks{This project was supported by DFG grant KA812/18-1.}}
\titlerunning{Generation of Topological Representations of Graphs Beyond-Planarity}

\author{%
Patrizio~Angelini\orcidID{0000-0002-7602-1524},
Michael~A.~Bekos\orcidID{0000-0002-3414-7444},
Michael~Kaufmann\orcidID{0000-0001-9186-3538},
Thomas~Schneck\orcidID{0000-0003-4061-8844}}

\institute{
Institut f\"ur Informatik, Universit\"at T\"ubingen, T\"ubingen, Germany\\
\email{\{angelini,bekos,mk,schneck\}@informatik.uni-tuebingen.de}
}

\maketitle

\begin{abstract}
Beyond-planarity focuses on combinatorial properties of clas\-ses of non-planar graphs that allow for representations satisfying certain local geometric or topological constraints on their edge crossings.~Beside the study of a specific graph class for its maximum edge density, another parameter that is often considered in the literature is the size of the largest complete or complete bipartite graph belonging to it.

Overcoming the limitations of standard combinatorial arguments, we present a technique to systematically generate all non-isomorphic topological representations of complete and complete bipartite graphs, taking into account the constraints of the~specific class. As a proof of concept, we apply our technique to various beyond-planarity classes and achieve new tight bounds for the aforementioned parameter.
\end{abstract}


\section{Introduction}
\label{sec:introduction}

Beyond-planarity is an active research area concerned with combinatorial properties of non-planar graphs that lie in the ``neighborhood'' of planar graphs.~More concretely, these graphs allow for non-planar drawings in which certain geometric~or topological crossing configurations are forbidden. The most studied beyond-planarity classes, with early results dating back to 60's~\cite{avital-66,MR0187232}, are the~\emph{$k$-planar} graphs~\cite{PachT97}, which forbid an edge to be crossed more than $k$ times, and the \emph{$k$-quasiplanar} graphs~\cite{AgarwalAPPS97}, which forbid $k$ mutually crossing edges; see Figs.~\ref{fig:1-planar}-\ref{fig:quasi}. 

More recently, several other classes have been suggested (e.g.,~\cite{%
DBLP:journals/tcs/AngeliniBKKS18,DBLP:journals/tcs/BaeBCEE0HKMRT18%
}), 
also motivated by cognitive experiments~\cite{HuangHE08,DBLP:journals/siamjo/Mutzel01} indicating that the absence of certain types of crossings helps in improving the readability of a drawing; for a survey, refer to~\cite{DBLP:journals/corr/abs-1804-07257}. 
Some of the most studied are: 
\begin {inparaenum}[(i)]
\item \emph{fan-planar} graphs, in which no edge can be crossed by two independent edges or by two adjacent edges from different
directions~\cite{BekosCGHK14,%
BinucciGDMPST15,KaufmannU14}, 
\item \emph{fan-crossing free} graphs, in which no edge can be crossed by two adjacent edges~\cite{DBLP:journals/ipl/Brandenburg18,DBLP:journals/algorithmica/CheongHKK15}, 
\item \emph{gap-planar} graphs, in which each crossing is assigned to one of its two involved edges, such that each edge can be assigned at most one crossing~\cite{DBLP:journals/tcs/BaeBCEE0HKMRT18}, 
and
\item \emph{RAC graphs}, in which edge crossings occur only at right angles~\cite{DidimoEL11,Didimo2013,EadesL13}; see Figs.~\ref{fig:fan-planar}-\ref{fig:rac}.
\end{inparaenum}
Note that all the aforementioned graph classes are \textit{topological}, i.e., each edge is represented as a simple curve, with the only exception of the class of RAC graphs, which is a purely \textit{geometric} graph class, i.e., each edge must be represented as a straight-line segment. In this work, we refer to the aforementioned topological graph classes as \emph{beyond-planarity classes of topological graphs}.

\begin{figure}[t]
	\centering
	\subcaptionbox{\label{fig:1-planar}}{\includegraphics[width=0.15\textwidth,page=1]{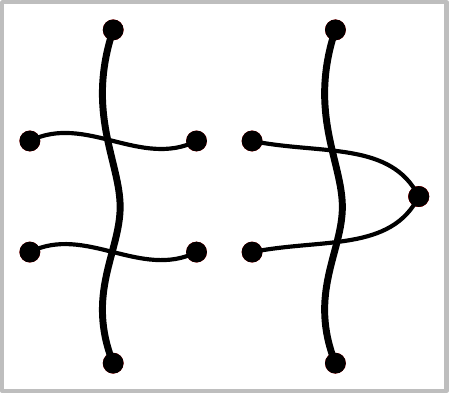}}
	\hfil
	\subcaptionbox{\label{fig:quasi}}{\includegraphics[width=0.15\textwidth,page=5]{beyond-planarity}}
	\hfil
	\subcaptionbox{\label{fig:fan-planar}}{\includegraphics[width=0.15\textwidth,page=4]{beyond-planarity}}
	\hfil
	\subcaptionbox{\label{fig:fan-crossing-free}}{\includegraphics[width=0.15\textwidth,page=2]{beyond-planarity}}
	\hfil
	\subcaptionbox{\label{fig:rac}}{\includegraphics[width=0.15\textwidth,page=6]{beyond-planarity}}
	\hfil
	\subcaptionbox{\label{fig:gap}}{\includegraphics[width=0.15\textwidth,page=3]{beyond-planarity}}
   \caption{Different forbidden crossing configurations in:
   (a)~1-planar,
   (b)~3-quasiplanar,
   (c)~fan-planar,
   (d)~fan-crossing~free, 
   (e)~gap-planar, and
   (f)~RAC graphs.}
\label{fig:beyond-planarity}
\end{figure}

A common characteristic of these graph classes is that their edge density is at most linear in the number of vertices, e.g., $1$-planar graphs with $n$ vertices~have at most $4n-8$ edges~\cite{PachT97}; see Table~\ref{table:density}. Another common measure to determine the extent of a specific class is the size of the largest complete or complete bipartite graph belonging to it~\cite{DBLP:journals/tcs/BaeBCEE0HKMRT18,franz-quasi-planar,DBLP:journals/dam/CzapH12,DidimoEL10}, which also provides a lower bound on their chromatic number~\cite{Hadwiger43} and has been studied in related fields (e.g.,~\cite{%
DBLP:journals/comgeo/ArleoBGEGLMMWW18,%
DBLP:journals/jgaa/BruckdorferCGKM17,%
DBLP:journals/algorithmica/EppsteinKKLLMMV18,%
DBLP:journals/gc/HartsfieldJR85
}). 

For $1$-planar graphs, Czap and Hud\'ak~\cite{DBLP:journals/dam/CzapH12} proved that the complete graph~$K_n$ is $1$-planar if and only if $n \leq 6$, and that the complete bipartite graph $K_{a,b}$, with $a \leq b$, is $1$-planar if and only if $a\leq 2$, or $a=3$ and $b \leq 6$, or $a=b=4$. An analogous characterization is known for the class of RAC graphs by Didimo et al.~\cite{DidimoEL10,DidimoEL11}, who proved that $K_n$ is a RAC graph if and only if $n \leq 5$, while $K_{a,b}$, with $a \leq b$, is a RAC graph if and only if $a \leq 2$, or $a=3$ and $b \leq 4$. For the classes of $3$-quasiplanar (also known as \emph{quasiplanar}), gap-planar, and fan-crossing free graphs, characterizations exist only for complete graphs, i.e.,~$K_n$~is quasiplanar if and only if $n \leq 10$~\cite{DBLP:journals/jct/AckermanT07,franz-quasi-planar}, gap-planar if and only if $n \leq 8$~\cite{DBLP:journals/tcs/BaeBCEE0HKMRT18}, and fan-crossing free if and only if $n \leq 6$~\cite{DBLP:journals/algorithmica/CheongHKK15,DBLP:journals/dam/CzapH12}; Table~\ref{table:density}~gives~more~details. 

To prove the ``if part'' of these characterizations, one has to provide a certificate drawing of the respective graph.
The proof for the ``only if part'' is generally more complex, as it requires arguments to show that no such drawing exists. 

One of the main techniques is provided by the linear edge density of the graph classes; e.g., $K_7$ is neither $1$-planar nor fan-crossing free, as it has more than $4n-8$ edges~\cite{DBLP:journals/algorithmica/CheongHKK15,PachT97}. However, this technique has a limited applicability; e.g., for $2$-planar and fan-planar graphs, which have at most $5n-10$ edges,~it~only ensures that $K_9$ is not a member of these classes. Proving that $K_8$ is also not a member requires a different approach. The limitations are even more evident~for complete bipartite graphs, as they are sparser than the complete~ones~(see~Section~\ref{sec:applications}).

\begin{table}[t!]
  \caption{%
  Known results and our findings. For each class, we present the largest complete and complete bipartite graphs that belong to this class (col. ``$\in$''), and the smallest ones that do not (col. ``$\notin$''). Color gray indicates weaker results that follow from other~entries.
  }
  \label{table:density}
  \medskip
  \resizebox{\columnwidth}{!}{
  \begin{tabular}{lc@{\hspace{.9em}}c@{\hspace{.01em}}c@{\hspace{.01em}}c@{\hspace{.01em}}ccc@{\hspace{.01em}}c@{\hspace{.01em}}c@{\hspace{.01em}}cc}
    \toprule
     & & \multicolumn{4}{c}{complete} & \multicolumn{4}{c}{complete bipartite}\\
    \cmidrule(r{8pt}){3-6} \cmidrule(r{8pt}){7-10}
    Class & Density & $\in$ & Ref. & $\notin$ & Ref. & $\in$ & Ref. & $\notin$ & Ref. \\
    \midrule
    1-planar & $4n-8$ &$K_{6}$ & \cite[Fig.1]{DBLP:journals/dam/CzapH12} & $K_{7}$ & \cite[Thm.1]{PachT97} & $K_{3,6}$ & \cite[Fig.2]{DBLP:journals/dam/CzapH12} & $K_{3,7}$ & \cite[Lem.4.2]{DBLP:journals/dam/CzapH12} \\
     & & & & & & $K_{4,4}$ & \cite[Fig.3]{DBLP:journals/dam/CzapH12} & $K_{4,5}$ & \cite[Lem.4.3]{DBLP:journals/dam/CzapH12} \\
     \midrule
    2-planar & $5n-10$ & $K_{7}$ & \cite[Fig.7]{BinucciGDMPST15} & $K_{8}$ & Char.\ref{th:complete:kplanar}  & $K_{3,10}$ & \cite[Lem.1]{DBLP:journals/tcs/AngeliniBKKS18} & $K_{3,11}$ & \cite[Lem.1]{DBLP:journals/tcs/AngeliniBKKS18} \\
     & & & & & & $K_{4,6}$ & Char.\ref{th:bipartite:2planar} &  $K_{4,7}$ & Char.\ref{th:bipartite:2planar} \\
     & & & & & & \textcolor{gray}{$K_{4,5}$}  &   &  $K_{5,5}$ & Char.\ref{th:bipartite:2planar} \cite{Kehribar18} \\[0.5ex]
     \midrule
    3-planar & $\frac{11}{2}n-11$ & $K_{8}$ & Char.\ref{th:complete:kplanar} & $K_{9}$ & Char.\ref{th:complete:kplanar}  & $K_{3,14}$ & \cite[Lem.1]{DBLP:journals/tcs/AngeliniBKKS18} & $K_{3,15}$ & \cite[Lem.1]{DBLP:journals/tcs/AngeliniBKKS18} \\
     & & & & & &  $K_{4,9}$ & Char.\ref{th:bipartite:3planar} &  $K_{4,10}$ & Char.\ref{th:bipartite:3planar} \\
     & & & & & &  $K_{5,6}$ & Char.\ref{th:bipartite:3planar} &  $K_{5,7}$ & Char.\ref{th:bipartite:3planar}\\
     & & & & & & \textcolor{gray}{$K_{5,6}$} &  &  $K_{6,6}$ & Char.\ref{th:bipartite:3planar} \\
     \midrule
    4-planar & $6n-12$ & $K_{9}$ & Char.\ref{th:complete:kplanar} & $K_{10}$ & Char.\ref{th:complete:kplanar}  & $K_{3,18}$ & \cite[Lem.1]{DBLP:journals/tcs/AngeliniBKKS18} & $K_{3,19}$ & \cite[Lem.1]{DBLP:journals/tcs/AngeliniBKKS18} \\
     & & & & & &  $K_{4,11}$ & Obs.\ref{th:bipartite:4planar} &  \textcolor{gray}{$K_{4,19}$} &  \\
     & & & & & &  $K_{5,8}$ & Obs.\ref{th:bipartite:4planar} &  \textcolor{gray}{$K_{5,19}$} &  \\
     & & & & & &  $K_{6,6}$ & Obs.\ref{th:bipartite:4planar} &  \textcolor{gray}{$K_{6,19}$} &  \\
     \midrule
    fan-planar & $5n-10$ & $K_{7}$ & \cite[Fig.7]{BinucciGDMPST15} & $K_{8}$ & Char.\ref{th:complete:fanplanar}  & $K_{4,n}$ & \cite[Fig.3]{KaufmannU14} & $K_{5,5}$ & Char.\ref{th:bipartite:fanplanar} \\[0.5ex]
     \midrule
    fan-crossing & $4n-8$ & $K_{6}$ & \cite[Fig.1]{DBLP:journals/dam/CzapH12} & $K_{7}$ & \cite[Thm.1]{DBLP:journals/algorithmica/CheongHKK15}  & \textcolor{gray}{$K_{3,6}$} &  & $K_{3,7}$ & Char.\ref{th:bipartite:fcf} \\
    free & & & & & &  $K_{4,6}$ & Char.\ref{th:bipartite:fcf} &  \textcolor{gray}{$K_{4,7}$} &  \\
     & & & & & & \textcolor{gray}{$K_{4,5}$} &  &  $K_{5,5}$ & Char.\ref{th:bipartite:fcf} \\
     \midrule
    gap-planar & $5n-10$ & $K_{8}$ & \cite[Fig.7]{DBLP:journals/tcs/BaeBCEE0HKMRT18} & $K_{9}$ & \cite[Thm.23]{DBLP:journals/tcs/BaeBCEE0HKMRT18}  & $K_{3,12}$ & \cite[Fig.7]{DBLP:journals/tcs/BaeBCEE0HKMRT18} & $K_{3,14}$ & \cite[Thm.1]{DBLP:conf/gd/BachmaierRS18} \\
     & & & & & &  $K_{4,8}$ & \cite[Fig.9]{DBLP:journals/tcs/BaeBCEE0HKMRT18} & $K_{4,9}$ & Obs.\ref{th:bipartite:gapplanar} \\
     & & & & & &  $K_{5,6}$ & \cite[Fig.9]{DBLP:journals/tcs/BaeBCEE0HKMRT18} & $K_{5,7}$ & \cite{DBLP:journals/tcs/BaeBCEE0HKMRT18} \\
     & & & & & &  \textcolor{gray}{$K_{5,6}$} &  &  $K_{6,6}$ & \cite[Thm.1]{DBLP:conf/gd/BachmaierRS18}  \\
		\midrule
    RAC & $4n-10$ & $K_{5}$ & \cite[Fig.5]{EadesL13} & $K_{6}$ & \cite[Thm.1]{DidimoEL11}  & $K_{3,4}$ & \cite[Fig.4]{DidimoEL10} & $K_{3,5}$ & \cite[Thm.2]{DidimoEL10} \\
     & & & & & & \textcolor{gray}{$K_{3,4}$} & & $K_{4,4}$ & \cite[Thm.2]{DidimoEL10} \\
     \midrule
    quasiplanar & $\frac{13}{2}n-20$ & $K_{10}$ & \cite[Fig.1]{franz-quasi-planar} & $K_{11}$ & \cite[Thm.5]{DBLP:journals/jct/AckermanT07}  & $K_{4,n}$ & \cite[Fig.3]{KaufmannU14} & -- &  \\
     & & & & & &  $K_{5,18}$ & Obs.\ref{th:bipartite:quasiplanar} &  ? &  \\
     & & & & & &  $K_{6,10}$ & Obs.\ref{th:bipartite:quasiplanar} &  ? &  \\
     & & & & & &  $K_{7,7}$ & Obs.\ref{th:bipartite:quasiplanar} &  $K_{7,52}$ & \cite[Thm.5]{DBLP:journals/jct/AckermanT07}  \\
    \bottomrule
  \end{tabular}
  }
\end{table}

Another technique consists of showing that the minimum number of crossings required by \emph{any} drawing of a certain graph (as derived by, e.g., the Crossing Lemma~\cite{%
DBLP:journals/corr/Ackerman15,DBLP:books/daglib/0019107,%
PachRTT06%
} or closed formulas~\cite{%
Zarankiewicz54%
}) exceeds the maximum number of crossings allowed in the considered graph class. However, this technique only applies to classes that impose such restrictions, e.g., gap- and $1$-planar graphs~\cite{DBLP:conf/gd/BachmaierRS18,DBLP:journals/dam/CzapH12}.

This difficulty in finding combinatorial arguments to prove that certain complete (bipartite) graphs do not belong to specific classes 
often results in the need of a large case analysis on the different topological representations of the graph. Beside the proofs in~\cite{DidimoEL10,Kehribar18}, we give in \arxapp{\cite{arxiv}}{Appendix~\ref{app:bipartite:fcf}} another example~of~a combinatorial proof that, based on a tedious case analysis, yields a characterization of the complete bipartite fan-crossing free graphs. The range of the cases~in~these proofs justifies the need of a tailored approach to systematically explore them. 

\smallskip\noindent\textbf{Our contribution.} We suggest a technique to engineer the analysis of all topological representations of a graph that satisfy certain beyond-planarity constraints. Our technique does not extend to classes of geometric graphs, and is tailored for complete and complete bipartite graphs, as we exploit their symmetry to reduce the search space, by discarding equivalent topological representations.

In Section~\ref{sec:enumeration}, we present an algorithm to generate all possible representations of such graphs under different topological constraints on the crossing configurations. Our algorithm builds on two key ingredients, which allow to drastically reduce the search space. First, the representations are constructed by adding a vertex at a time, directly taking into account the topological constraints, thus avoiding constructing unnecessary representations. Second, at each intermediate step, the produced drawings are efficiently tested for equivalence (up to a relabeling of the vertices), which usually allows to discard a large set of them. Using this algorithm, we derived characterizations for several classes, as described in Section~\ref{sec:applications}; Table~\ref{table:density} positions our results with respect to the state of the art.  We give preliminary definitions in Section~\ref{sec:preliminaries} and discuss future directions in Section~\ref{sec:conclusions}.

\section{Preliminaries}
\label{sec:preliminaries}

We assume familiarity with standard definitions on planar graphs and drawings \arxapp{(see, e.g.,~\cite{arxiv})}{(see Appendix~\ref{app:preliminaries})}. We assume \emph{simple} drawings, in which there are no self-crossing edges, two edges cross at most once, and adjacent edges do not cross; note this assumption is not without loss of generality~\cite{DBLP:journals/jct/AckermanT07}. Given a planarization $\Gamma$ of a graph $G$, a \emph{half-pathway for a vertex} $u$ in $\Gamma$ is a path in the dual of $\Gamma$ from a face incident to $u$ to some face in $\Gamma$, called its \emph{destination}; see Fig.~\ref{subfig:prohibitedEdgesA}. The \emph{length} of a half-pathway is the number of edges in this path. A half-pathway for $u$ is \emph{valid} with respect to a beyond-planarity class~$\mathcal{C}$ of topological graphs, if $\Gamma$ can be augmented such that
\begin {inparaenum}[(i)]
\item a vertex $v$ is placed in its destination,
\item edge $(u,v)$ is drawn as a curve from $u$ to $v$ that crosses only the edges that are dual to the edges in this half-pathway, in the same order, and 
\item the drawing of $(u,v)$ violates neither the simplicity of the resulting drawing nor the crossing restrictions of class $\mathcal{C}$.
\end{inparaenum}
Accordingly, a \emph{pathway for an edge} $(u,v)$ is a half-pathway for vertex $u$ in $\Gamma$, whose destination is a face incident to vertex $v$. A \emph{valid pathway} is defined analogously, with the only exception that $v$ is already part~of~$\Gamma$.

Another ingredient of our algorithm is an equivalence-relationship between~different drawings of a graph $G$, i.e, drawings $D_1$ and $D_2$ of $G$ are \emph{isomorphic}~\cite{DBLP:journals/dcg/Kyncl11} if there exists a homeomorphism of the sphere transforming $D_1$ into $D_2$\arxapp{}{; see Fig.~\ref{fig:k5} in Appendix~\ref{app:preliminaries}}. Namely, $D_1$ and $D_2$ are isomorphic if $D_1$~can~be~transformed into $D_2$ by relabeling vertices, edges, and faces of $D_1$, and by moving~vertices and edges of $D_1$, so that at no time of this process new~crossings are~introduced, existing crossings are eliminated, or the order of the crossings along an edge is modified. We define a \emph{valid} bijective mapping between vertices, crossings, edges, and faces of the planarizations $\Gamma_1$ and $\Gamma_2$ of $D_1$ and $D_2$ such that:
%
\begin{inparaenum}[\bf ({P.}1)]
\item \label{p:vm2} if an edge $(v_1,w_1)$ is mapped to an edge~$(v_2,w_2)$~in $\Gamma_1$ and $\Gamma_2$, respectively, and $v_1$ is mapped to $v_2$, then $w_1$ is mapped~to~$w_2$;
\item \label{p:vm3} if a face $f_1$ is mapped to a face $f_2$ in $\Gamma_1$ and $\Gamma_2$, respectively, and an edge $e_1$ incident to $f_1$ is mapped to an edge $e_2$ incident to $f_2$, then the predecessor (successor) of $e_1$ is mapped to the predecessor (successor) of $e_2$ when walking along the boundaries of $f_1$ and $f_2$ in clockwise direction. Also, the face incident to the other side of $e_1$ is mapped to the face incident to the other side of~$e_2$.
\end{inparaenum}
%
Clearly, Properties~P.\ref{p:vm2} and~P.\ref{p:vm3} are sufficient for $D_1$ and $D_2$ to be isomorphic. We believe they are also necessary, but this is beyond the scope of this work. Note that Property~P.\ref{p:vm3} guarantees that two vertices are mapped to each other only if they have the same degree.

\begin{figure}[t]
	\centering
	\subcaptionbox{\label{subfig:prohibitedEdgesA}}{
		\includegraphics[page=4, scale=0.9]{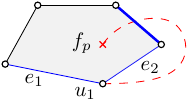}}
	\hfil
	\subcaptionbox{\label{subfig:prohibitedEdgesB}}{
		\includegraphics[page=5, scale=0.9]{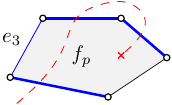}}
	\hfil
	\subcaptionbox{\label{subfig:prohibitedEdgesC}}{
		\includegraphics[page=6, scale=0.9]{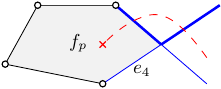}}
	\caption{%
		The prohibited edges (blue solid) for a half-pathway (red dashed) that ends in a face $f_p$. The thick blue edges are prohibited, because they are crossed by the half-pathway. 
		In~(a) edges $e_1$ and $e_2$ are prohibited, since they are incident to $u_1$.
		In~(b) edge $e_3$ is prohibited, since, in order to cross this edge, the half-pathway would make a self-crossing.
		In~(c) edge $e_4$ is prohibited since it is part of a crossed edge.}
	\label{fig:prohibitedEdges}
\end{figure}

Several works~\cite{Abrego15,Gronau90,Rafla88} that generate simple drawings of complete graphs adopt a weaker definition of isomorphism; two drawings $D_1$ and~$D_2$ are \emph{weakly isomorphic}~\cite{DBLP:journals/dcg/Kyncl11}, if there exists an incidence preserving bijection between their vertices and edges, such that two edges cross in $D_1$  if and only if they~do~in $D_2$. Weakly isomorphic drawings that are non-isomorphic differ~in~the~order in which their edges cross~\cite{DBLP:conf/wg/Gioan05}. Two simple drawings of a complete graph with the same cyclic order of the edges around each vertex (called \emph{rotation system}) are weakly isomorphic, and vice versa~\cite{DBLP:conf/wg/Gioan05,DBLP:journals/combinatorica/PachT06}; hence, generating all simple drawings of a complete graph reduces to finding all rotation systems that determine simple drawings~\cite{DBLP:journals/dcg/Kyncl13}. However, this property holds only for complete graphs~\cite{Abrego15}, while for the complete bipartite graphs, which are more difficult to handle, only partial results exist in this direction~\cite{DBLP:journals/jocg/CardinalF18}. Thus, we decided not to follow this approach. 

\section{Generation Procedure}
\label{sec:enumeration}

Let $\mathcal{C}$ be a beyond-planarity class of topological graphs and let $G$ be a graph with $n \geq 3$ vertices. Assuming~that $G$ is either complete or complete bipartite, we~describe in this section an algorithm to generate all non-isomorphic simple drawings of $G$ that are certificates that $G$ belongs to $\mathcal{C}$ (if any). We stress that, if $G$ is neither complete nor complete bipartite, then it is a more involved task to recognize isomorphic drawings~\cite{gj-cigtnpc-79}, and thus to eliminate them, which is a key point in the efficiency of our approach (we provide more details in Section~\ref{sec:applications}).

Our algorithm aims at computing a set $\mathcal{S}$ containing all non-isomorphic simple drawings of $G$. In the base of the recursion, graph $G$ is a cycle of length $3$ or $4$, depending on whether $G$ is the complete graph $K_3$ or the complete bipartite graph $K_{2,2}$. In the former case, set $\mathcal{S}$ only contains a planar drawing of $K_3$, while in the latter case, $\mathcal{S}$ contains a planar drawing and one with a crossing between two non-adjacent edges. This is because, in both cases, any other drawing is either isomorphic to one of these, or non-simple.

In the recursive step, we consider a vertex $v$ of $G$ and assume that we have recursively computed the set $\mathcal{S}$ for $G\setminus\{v\}$. We may assume w.l.o.g.~that $\mathcal{S}\neq\emptyset$, as otherwise $G$ would not belong to $\mathcal{C}$. Then, we consider each drawing of $\mathcal{S}$~and our goal is to report all non-isomorphic simple drawings of $G$ that have it as a subdrawing. In other words, we aim at reporting all non-isomorphic simple drawings that can be derived by all different placements of vertex $v$ and the routing of its incident edges in the drawings of $\mathcal{S}$. To this end, let $\Gamma$ be the planarization of one of the drawings in $\mathcal{S}$, and let $u_1,\ldots,u_k$ be the neighbors of $v$ in $G$, where $k=\deg(v)$. If $G$ is a complete graph, then $k=n-1$;~otherwise, $G$ is a complete bipartite graph $K_{a,b}$ with $a+b=n$, and $k=a$ or $k=b$ holds.  

We start by computing all possible valid half-pathways for $u_1$ in $\Gamma$ with~respect to $\mathcal{C}$, which corresponds to constructing all possible drawings of edge $(v,u_1)$ that respect simplicity and the restrictions of class $\mathcal{C}$. To compute these half-pathways, we again use recursion. For each half-pathway, we maintain a list of so-called \emph{prohibited} edges, which are not allowed to be crossed~when inserting edge $(u_1,v)$, as otherwise either the simplicity or the crossing restrictions of class $\mathcal{C}$ would be violated; see Fig.~\ref{fig:prohibitedEdges}\arxapp{}{, and Fig.~\ref{fig:insertionExample} in Appendix~\ref{app:example}}. This list is initialized with all edges incident to $u_1$ and is updated at every recursive~step. 

In the base of this inner recursion, we determine all valid half-pathways for $u_1$ of length zero; this means that, for each face $f$ incident to $u_1$, we create a half-pathway that starts at $f$ and has its destination also at $f$, which corresponds to placing $v$ in $f$ and drawing edge $(v,u_1)$ crossing-free. Assume now that we have computed all valid half-pathways of some length $i\ge 0$ in $\Gamma$. We show how to compute all valid half-pathways for $u_1$ of length $i+1$ (if any). Consider a half-pathway $p$ of length $i$. Let $f_p$ be its destination. Every non-prohibited edge $e$ of $f_p$ implies a new half-pathway of length $i+1$, composed of $p$ followed by the edge that is dual to $e$ in $\Gamma$. Note that this process will eventually terminate, since the length of a half-pathway is bounded by the number of edges of $\Gamma$.

For each valid half-pathway $p$ computed by the procedure above, we obtain~a new drawing by inserting $(u_1,v)$ into $\Gamma$ following $p$ and by inserting~$v$ into the destination of $p$. It remains to insert the remaining edges incident to $v$, i.e., $(v,u_2),\ldots,(v,u_k)$, into each of these drawings -- again in all possible ways. For this, we proceed mostly as above with one difference. Instead of half-pathways, we search for valid pathways for each edge $(v,u_i)$, $2\le i \le k$, i.e.,~we~only consider pathways that start in a face incident to $v$ and end in a face incident~to~$u_i$. 

If we find an edge $(v,u_i)$ for which no valid pathway exists, we declare that $\Gamma$ cannot be extended to a simple drawing of $G$ that respects the crossing restrictions of $\mathcal{C}$. Otherwise, the computed drawings of $G$ are added~to~$\mathcal{S}$, once all the drawings of $G\setminus\{v\}$ have been removed from it. To maintain our initial invariant, however, once a new drawing is to be added to $\mathcal{S}$, it will be first checked for isomorphism against all previously added drawings. If there is an isomorphic one, then the current drawing is discarded; otherwise, it is added~to~$\mathcal{S}$. 

We stress that we test isomorphism using Properties~P.\ref{p:vm2} and~P.\ref{p:vm3} of a valid bijection. Since these properties are sufficient but we do not know whether they are also necessary, set $\mathcal{S}$ might contain some isomorphic drawings. However, our experiments indicate that the vast majority of them~will~be~discarded.

\myparagraph{Testing for isomorphism.}
We describe a procedure to test whether the planarizations $\Gamma_1$ and $\Gamma_2$ of two drawings of $G$ comply with Properties~P.\ref{p:vm2} and~P.\ref{p:vm3} of a valid bijection. We start by selecting two edges $e_1=(v_1,w_1)$ and $e_2=(v_2,w_2)$ in $\Gamma_1$ and $\Gamma_2$, respectively, whose end-vertices have compatible types (i.e., $v_1$ and $v_2$ are both real vertices or both crossings, and the same holds for $w_1$ and $w_2$). We bijectively map $e_1$ to $e_2$, $v_1$ to $v_2$, and $w_1$ to $w_2$, which complies with Property~P.\ref{p:vm2}. We call this a \emph{base mapping} and try to extend it to a valid bijection.

We map to each other the face $f_1$ of $\Gamma_1$ that is ``left'' of $e_1$ (when walking along $e_1$ from $v_1$ to $w_1$) and the face $f_2$ of $\Gamma_2$ that is ``left'' of $e_2$ (when walking along $e_2$ from $v_2$ to $w_2$). If the degrees of $f_1$ and $f_2$ are different, then the base mapping cannot be extended. Otherwise, both $f_1$ and $f_2$ have degree $\delta$, and we walk simultaneously along their boundaries, starting at $e_1$ and $e_2$ respectively; in view of Property P.\ref{p:vm3}, for each $i=1,\ldots,\delta$, we bijectively map the $i$-th vertex (either real or crossing) of $f_1$ to the $i$-th vertex of $f_2$, and the $i$-th edge of $f_1$ to the $i$-th edge of $f_2$. If a crossing is mapped to a real vertex, or if the degrees of two mapped vertices are different, then the base mapping cannot be extended.

If the vertices and edges of $f_1$ and $f_2$ have been mapped successfully, we proceed by considering the two maximal connected subdrawings $\Gamma'_1$ and $\Gamma'_2$ of $\Gamma_1$ and $\Gamma_2$, respectively, such that each edge of $\Gamma'_1$ and $\Gamma'_2$ has at least one face incident to it that is already mapped. Consider an edge $e'_1$ of $\Gamma'_1$ that is incident to only one mapped face $f'_1$ (such an edge exists, as long as the base mapping has not been completely extended). Let $e'_2$ be the edge of $\Gamma'_2$ mapped to $e'_1$; note that $e'_2$ must be incident to a face $f'_2$ that is mapped to $f'_1$ and to a face that is not mapped yet. We map to each other the faces incident to $e'_1$ end $e'_2$ that are not mapped yet, and we proceed by applying the procedure described above (i.e., we walk along the boundaries of $f'_1$ and $f'_2$ simultaneously, while ensuring that the mapping remains valid). If this procedure can be performed successfully, then we have computed two subdrawings $\Gamma''_1$ and $\Gamma''_2$, such that $\Gamma'_1\subseteq\Gamma''_1$, $\Gamma'_2\subseteq\Gamma''_2$, and each edge of them has at least one face incident to it that is already mapped. Hence, we can recursively apply the aforementioned procedure to $\Gamma''_1$ and $\Gamma''_2$. 

Drawings $\Gamma_1$ and $\Gamma_2$ are isomorphic, if the base mapping can be eventually extended. If not, then we have to consider another base mapping and check whether this can be extended. Note that the case where $e_1$ is bijectively mapped to $e_2$, $v_1$ to $w_2$, and $w_1$ to $v_2$ defines a different base mapping than the one we were currently considering. If none of the base mappings can be extended, then we consider $\Gamma_1$ and $\Gamma_2$ as non-isomorphic. To reduce the number of base mappings that we have to consider, we first count the number of edges of $\Gamma_1$ and $\Gamma_2$ whose endpoints are both real vertices, both crossings, and those consisting of one real vertex and one crossing. These numbers have to be the same in $\Gamma_1$ and $\Gamma_2$. Since it is enough to consider base mappings only restricted to one of the three types of edges, we choose the type with the smallest positive number of occurrences. 
We summarize the above discussion in the following theorem.

\begin{theorem}\label{thm:algorithm}
Let $G$ be a complete (or a complete bipartite) graph and let $\mathcal{C}$ be a beyond-planarity class of topological graphs. Then, $G$ belongs to $\mathcal{C}$ if and only if, under the restrictions of class $\mathcal{C}$, our algorithm returns a valid drawing of $G$.
\end{theorem}
	
\section{Proof of Concept - Applications}
\label{sec:applications}

In this section we use the algorithm described in Section~\ref{sec:enumeration} to test whether certain complete or complete bipartite graphs belong to specific beyond-planarity graph classes. We give corresponding characterizations and discuss how our findings are positioned within the literature. Our lower bound examples are drawings that certify membership to  particular beyond-planarity graph classes, computed by an implementation (\url{https://github.com/beyond-planarity/complete-graphs}) of our algorithm; for typesetting reasons we redrew them. Our upper bounds are the smallest corresponding instances reported as negative by our algorithm.

\myparagraph{The class of k-planar graphs.}
\label{subsec:application:kPlanar}
%
We start our discussion with the case of complete graphs. As already mentioned in the introduction, the complete graph $K_n$ is $1$-planar if and only if $n \leq 6$~\cite{DBLP:journals/dam/CzapH12}.

For the case of complete $2$-planar graphs, the fact that a $2$-planar graph with $n$ vertices has at most $5n-10$ edges~\cite{PachT97} implies that $K_9$ is not a member of this class. Fig.~7 in~\cite{BinucciGDMPST15}, on the other hand, shows that $K_7$ is $2$-planar. We close this gap by showing, with our implementation, that even $K_8$ is not $2$-planar. 

For the cases of complete $3$-, $4$-, and $5$-planar graphs, the application of a similar density argument as above proves that $K_{10}$, $K_{11}$, and $K_{19}$ are not $3$-, $4$-, and $5$-planar, respectively~\cite{DBLP:journals/corr/Ackerman15,PachRTT06}. With our implementation, we could show that even $K_{9}$ is not $3$-planar, while $K_{10}$ is neither $4$- nor $5$-planar. On the other hand, our algorithm was able to construct $3$- and $4$-planar drawings of $K_{8}$ and $K_{9}$, respectively; see Figs.~\ref{fig:drawing:3planar:complete:k8} and~\ref{fig:drawing:4planar:complete:k9}. Note that a $6$-planar drawing of $K_{10}$ can be easily derived from the $4$-planar drawing of $K_9$ in Fig.~\ref{fig:drawing:4planar:complete:k9} by adding one extra vertex inside the red colored triangle. We have the following characterization. 

\begin{chr} \label{th:complete:kplanar}
For $k \in \{1,2,3,4\}$, the complete graph $K_{n}$ is $k$-planar if and only if $n \leq 5+k$. Also, $K_{n}$ is $5$-planar if and only if $n \leq 9$.
\end{chr}

Note that the $3$-planarity of $K_8$ implies that the chromatic number of \mbox{$3$-planar} graphs is lower bounded by~$8$. Analogous implications can be derived for the classes of $4$-, $5$-, and $6$-planar graphs. Another observation that came out from our experiments is that, up to isomorphism, $K_6$ has~a unique $1$-planar drawing, $K_7$ has only two $2$-planar drawings, and $K_8$ has only three $3$-planar drawings, while the number of non-isomorphic $4$-planar~drawings of $K_9$ is significantly larger, namely $35$. For more details, refer to Table~\ref{table:comparison}, and \arxapp{to~\cite{arxiv}}{to Table~\ref{table:results} in Appendix~\ref{app:numbers}}.

\begin{figure}[t]
\centering
	\subcaptionbox{\label{fig:drawing:3planar:complete:k8}}{	
	\includegraphics[page=1,scale=0.16]{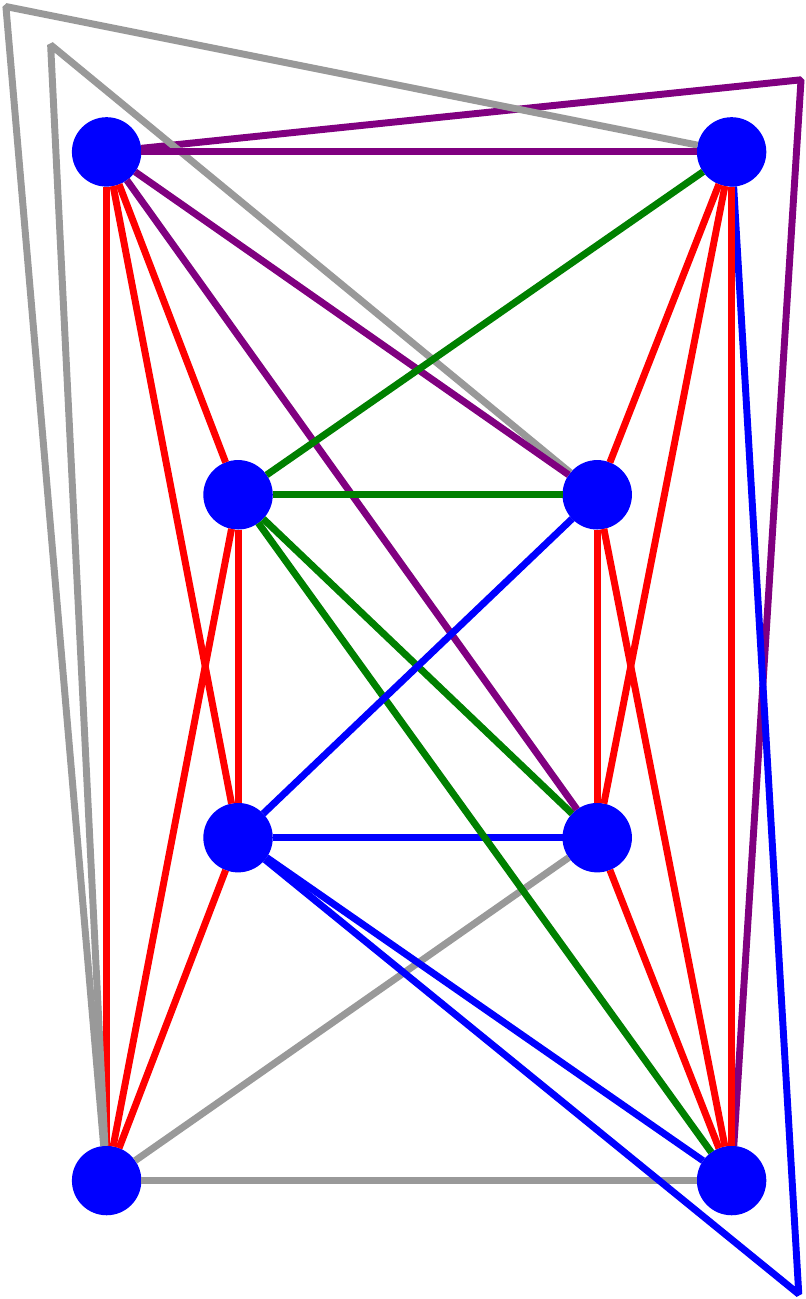}}
	\hfil
	\subcaptionbox{\label{fig:drawing:4planar:complete:k9}}{
	\includegraphics[page=1,scale=0.21]{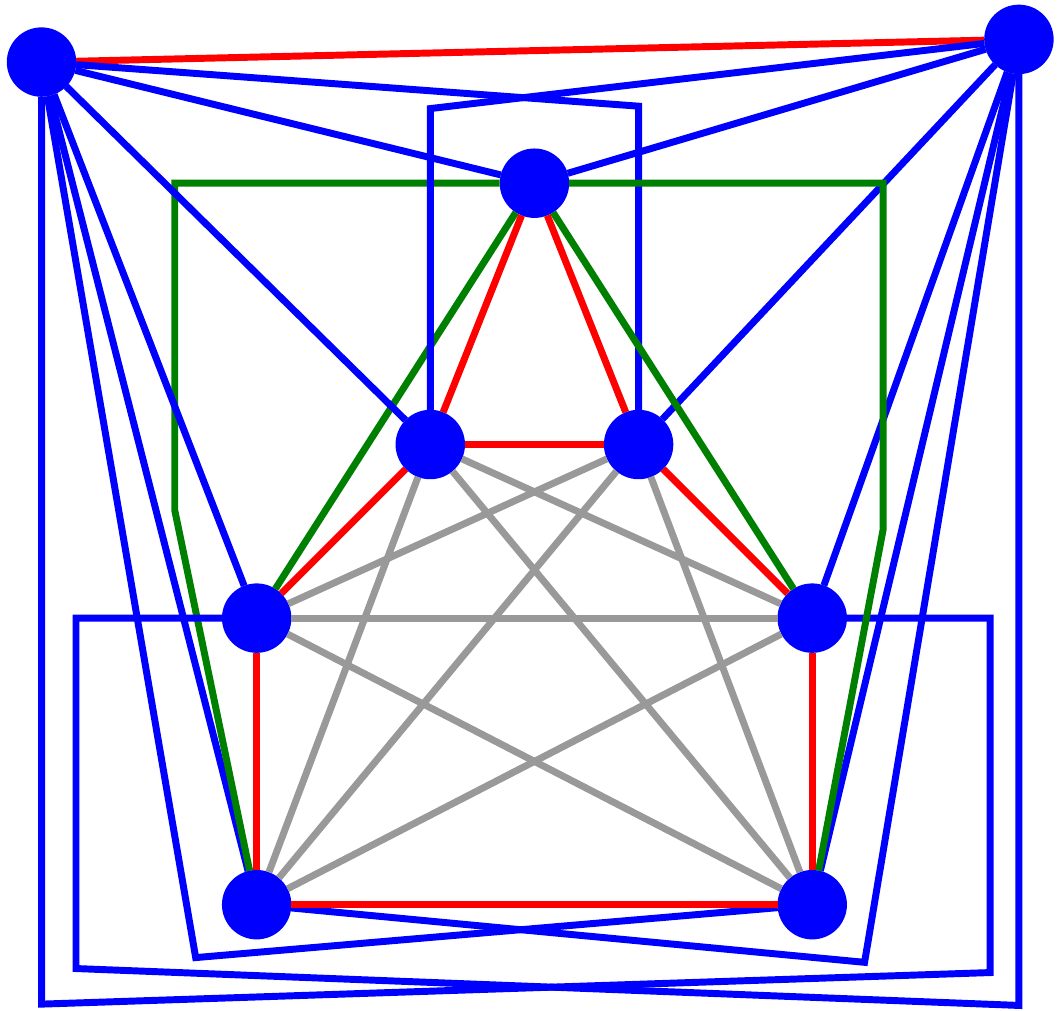}}
	\hfil
	\subcaptionbox{\label{fig:drawing:2planar}}{
	\includegraphics[scale=0.21]{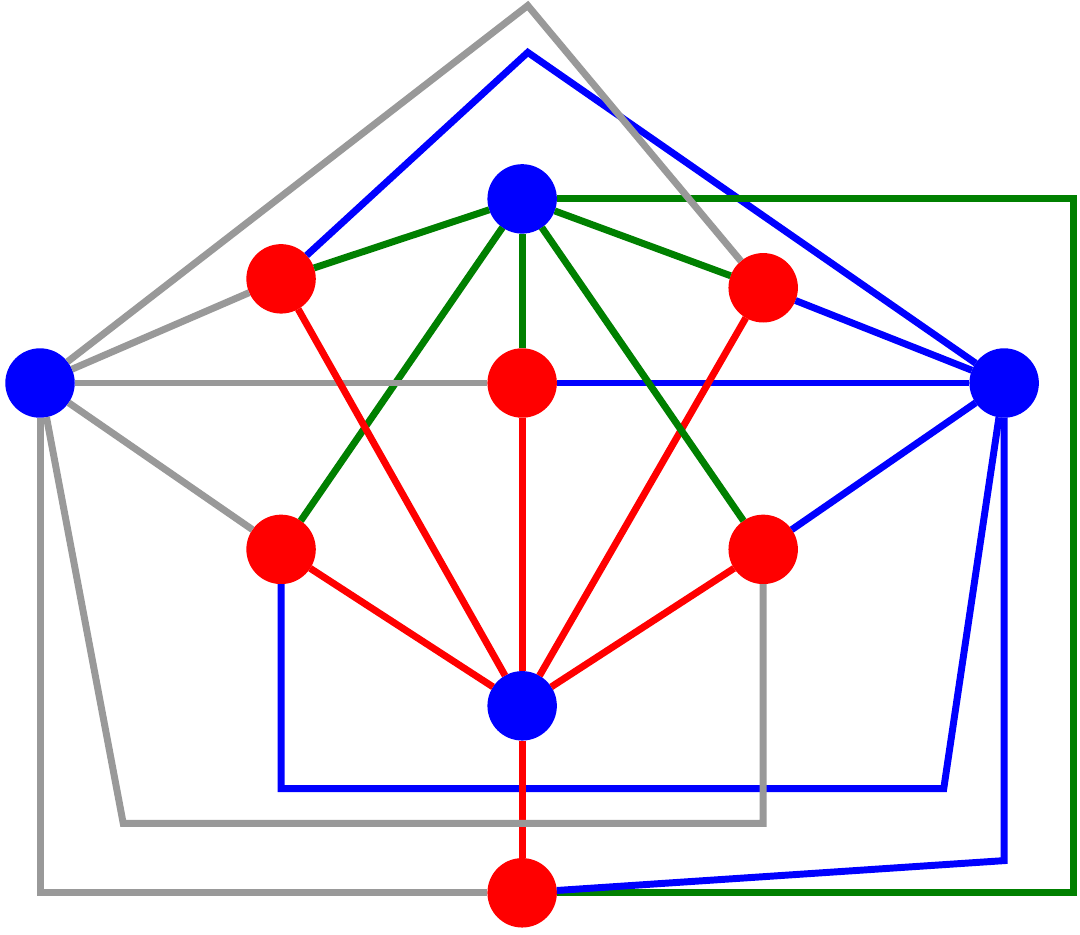}}
	\hfil
	\subcaptionbox{\label{fig:drawing:3planar:bipartite:k49}} {		
	\includegraphics[page=1,scale=0.21]{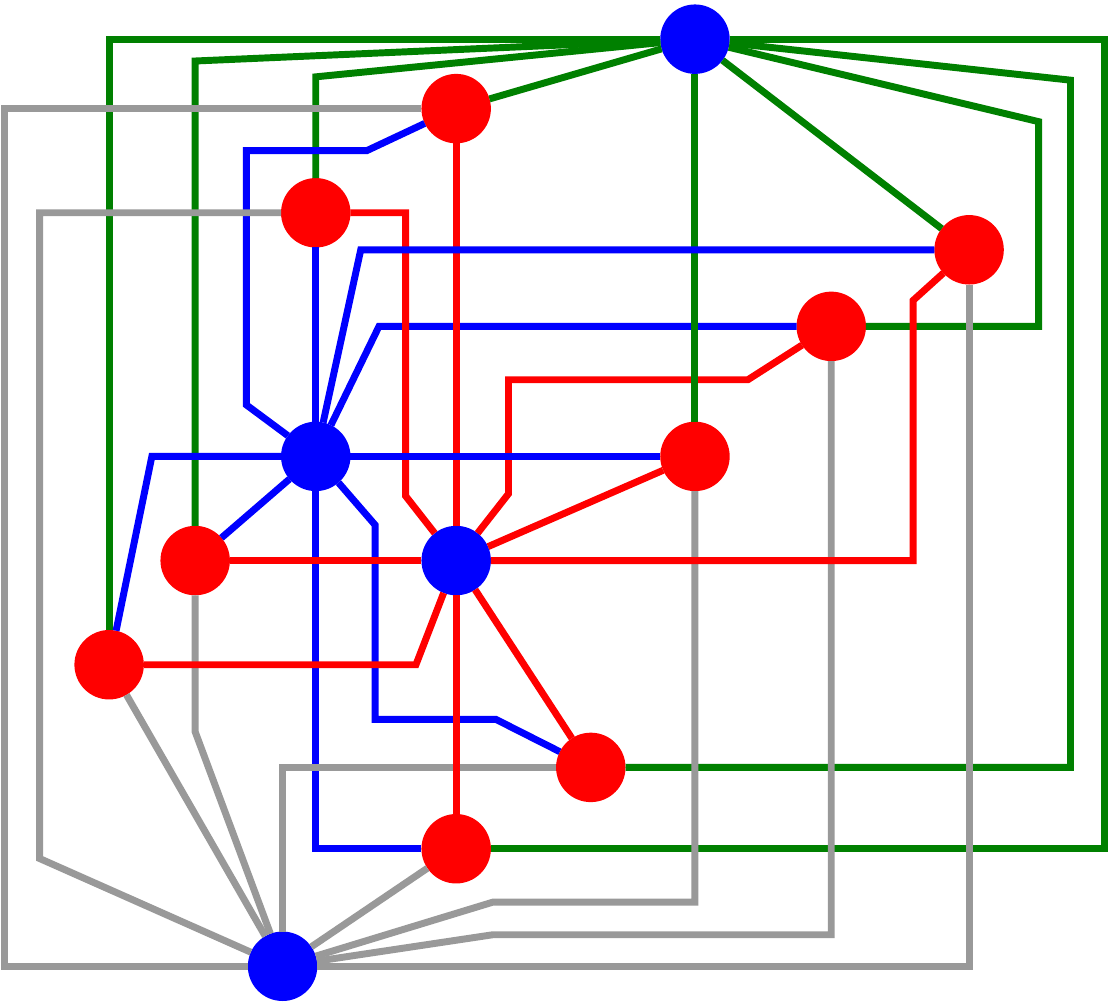}}
	\hfil
	\subcaptionbox{\label{fig:drawing:3planar:bipartite:k56}} {		
	\includegraphics[page=1,scale=0.21]{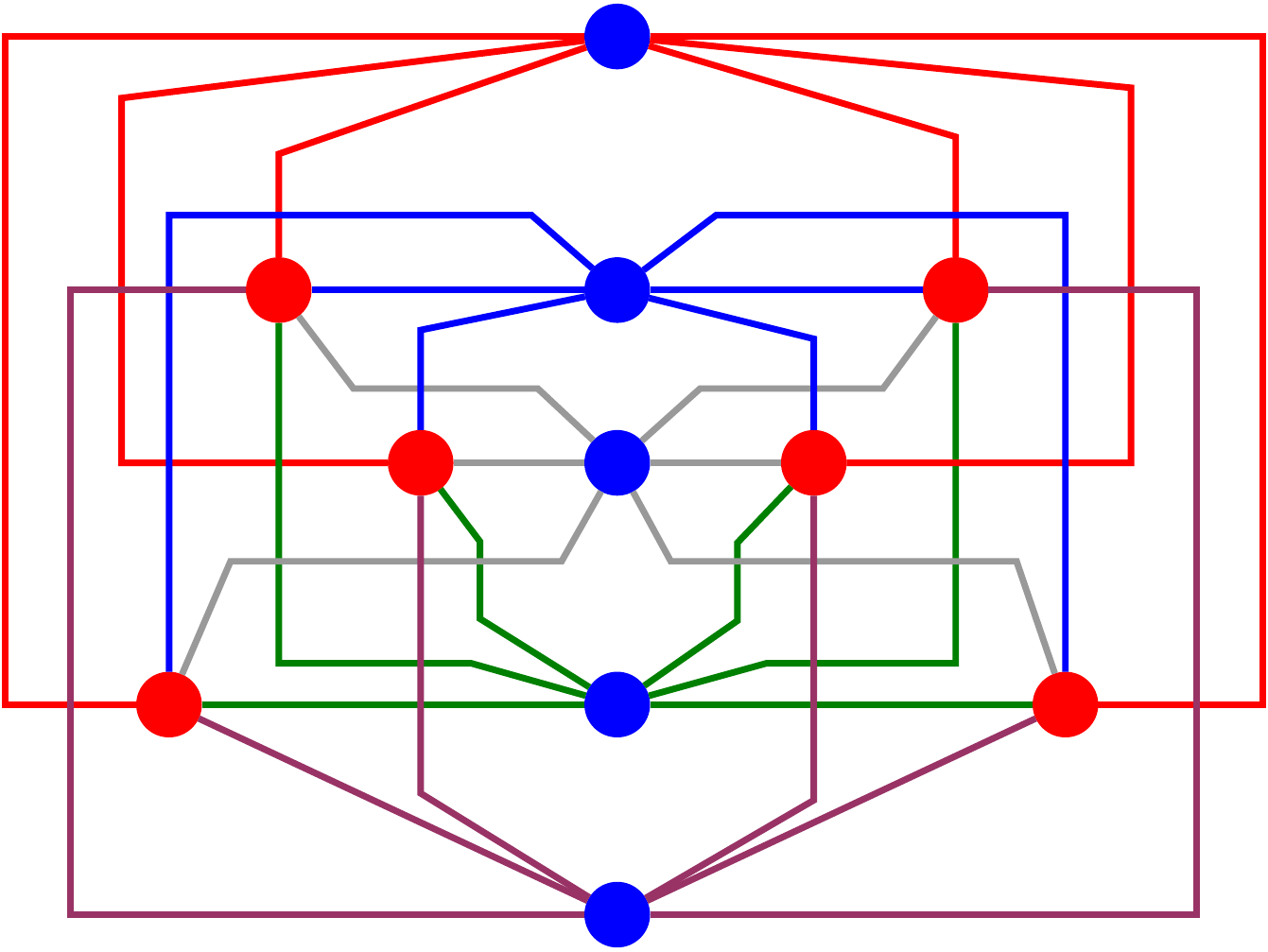}}
	\caption{%
	Illustration of
	(a)~a $3$-planar drawing of~$K_8$,  
	(b)~a $4$-planar drawing of $K_9$,
	(c)~a drawing of $K_{4,6}$ that is both 2-planar and fan-crossing free, 
	(d)~a $3$-planar drawing of $K_{4,9}$, and
	(e)~a $3$-planar drawing of $K_{5,6}$.}
	\label{fig:drawing:kplanar:complete}
\end{figure}

Consider now a complete bipartite graph $K_{a,b}$ with $a \leq b$. Note that $a \leq 2$ implies that $K_{a,b}$ is planar; thus, it trivially belongs to all beyond-planarity graph classes. Also, recall that $K_{a,b}$ is $1$-planar if and only if $a \leq 2$, or $a=3$ and $b \leq 6$, or $a=b=4$~\cite{DBLP:journals/dam/CzapH12}. Further, a recent combinatorial result states that $K_{3,b}$ is $k$-planar if and only if $b \leq 4k+2$~\cite{DBLP:journals/tcs/AngeliniBKKS18}. So, in the following we assume $a \geq 4$.

For complete bipartite $2$-planar graphs, the fact that a bipartite $2$-planar graph with $n$ vertices has at most $3.5n-7$ edges~\cite{DBLP:conf/isaac/AngeliniB0PU18} implies that neither $K_{4,15}$ nor $K_{5,8}$ is $2$-planar. With our implementation, we could show that $K_{4,7}$ and $K_{5,5}$ are not $2$-planar, while $K_{4,6}$ is (see Fig.~\ref{fig:drawing:2planar}), yielding the following characterization.
\begin{chr} \label{th:bipartite:2planar}
The complete bipartite graph $K_{a,b}$ (with $a \le b$) is 2-planar if and only if
\begin{inparaenum}[(i)]
	\item $a \le 2$, or
	\item $a = 3$ and $b \le 10$, or
	\item $a = 4$ and $b \le 6$.
\end{inparaenum}
\end{chr}

As opposed to the corresponding $2$-planar case, there exists no upper bound on the edge density of $3$-planar graphs tailored for the bipartite setting. The upper bound of $5.5n-11$ edges~\cite{PachRTT06} for general $3$-planar graphs with $n$ vertices does not provide any negative instance for $a \leq 5$, and only proves that $K_{6,b}$, with $b \ge 45$, is not $3$-planar. With our implementation, we could provide significant improvements, by showing that $K_{4,10}$, $K_{5,7}$, and $K_{6,6}$ are not $3$-planar, while $K_{4,9}$ and $K_{5,6}$ are (see Figs.~\ref{fig:drawing:3planar:bipartite:k49} and~\ref{fig:drawing:3planar:bipartite:k56}), which yields the following characterization.

\begin{chr} \label{th:bipartite:3planar}
The complete bipartite graph $K_{a,b}$ (with $a \le b$) is 3-planar if and only if
\begin{inparaenum}[(i)]
	\item $a \le 2$, or
	\item $a = 3$ and $b \le 14$, or
	\item $a = 4$ and $b \le 9$, or
	\item $a = 5$ and $b \le 6$.
\end{inparaenum}
\end{chr}

For complete bipartite $4$-planar graphs, we were unable to derive a characterization, but only some partial results, because the search space becomes drastically larger and, as a consequence, our generation technique could not terminate. To give an intuition, note that $K_{4,4}$ has 81817 non-isomorphic $4$-planar drawings, which makes the computation of the corresponding non-isomorphic drawings of $K_{4,5}$ infeasible in reasonable time;  \arxapp{for more details refer to~\cite{arxiv}}{see also Appendix~\ref{app:numbers}}. 

However, we were at least able to report some positive certificate drawings by slightly refining our generation technique. Instead of computing \emph{all} possible non-isomorphic simple drawings of graph $K_{a-1,b}$ or $K_{a,b-1}$, in order to compute the corresponding ones for $K_{a,b}$, we only computed few \emph{samples}, hoping that we will eventually find a positive certificate drawing. With this so-called \emph{DFS-like} approach, we managed to derive $4$-planar drawings for $K_{4,11}$, $K_{5,8}$, and $K_{6,6}$; see Fig.~\ref{fig:drawing:4planar}. We summarize these findings in the following observation.

\begin{figure}[t]
\centering
	\subcaptionbox{\label{fig:drawing:4planar:bipartite:k411}} {		
	\includegraphics[page=1,scale=0.2]{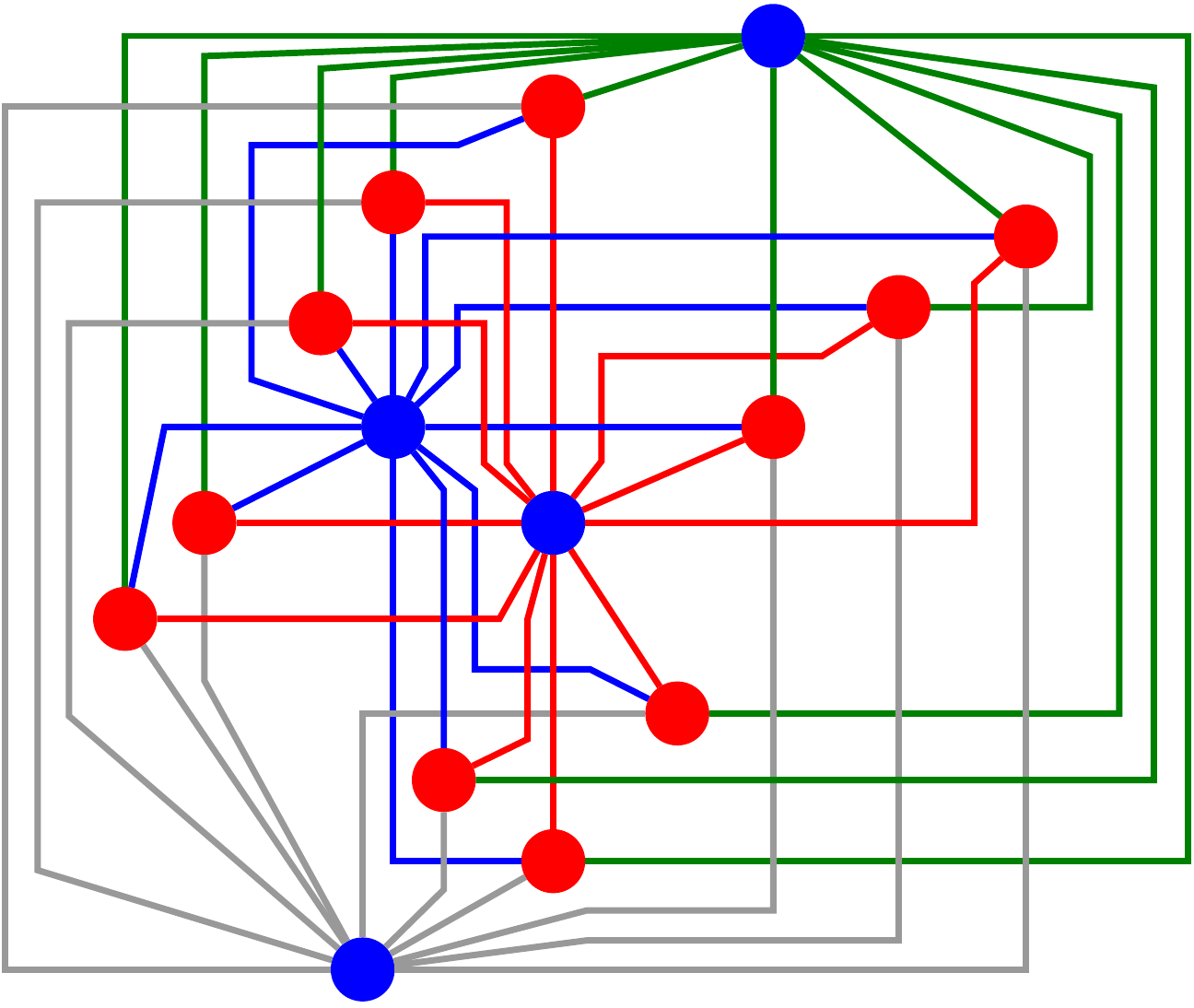}}
	\hfil
	\subcaptionbox{\label{fig:drawing:4planar:bipartite:k58}} {
	\includegraphics[page=1,scale=0.19]{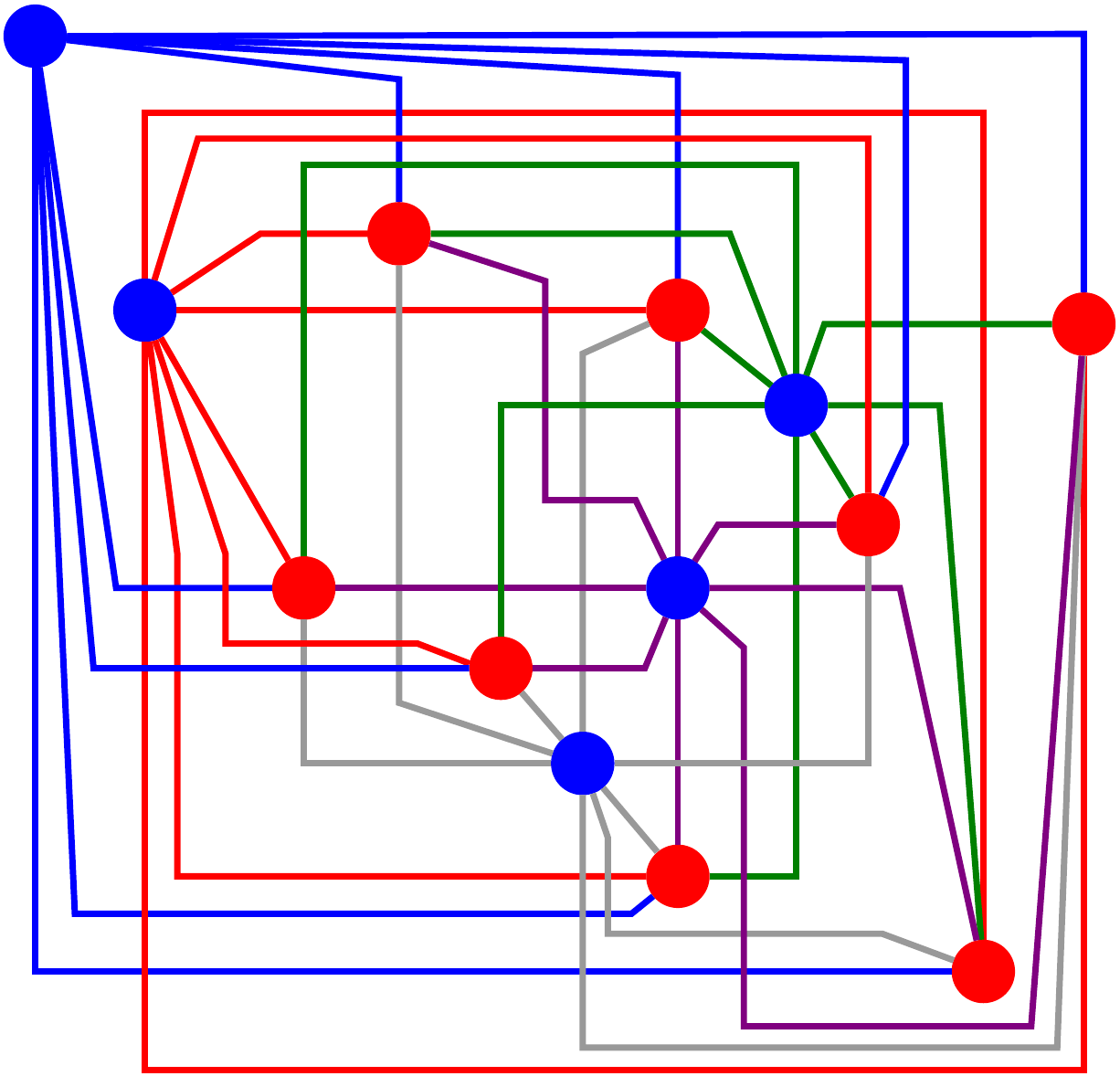}}
	\hfil
	\subcaptionbox{\label{fig:drawing:4planar:bipartite:k66}} {
	\includegraphics[page=1,scale=0.22]{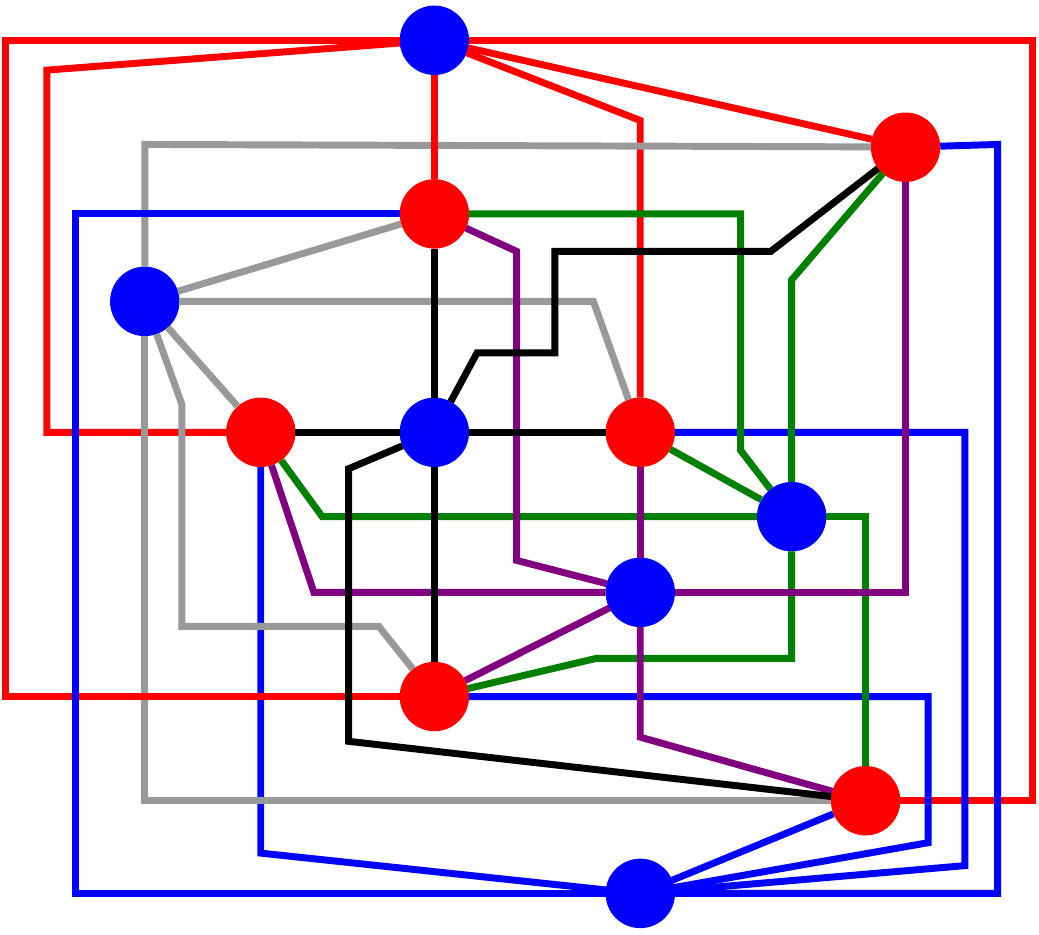}}
	\caption{%
	Illustration of $4$-planar drawings of 
	(a)~$K_{4,11}$, 
	(b)~$K_{5,8}$ and 
	(c)~$K_{6,6}$.}
	\label{fig:drawing:4planar}
\end{figure}

\begin{obs} \label{th:bipartite:4planar}
The complete bipartite graph $K_{a,b}$ (with $a \le b$) is 4-planar if
\begin{inparaenum}[(i)]
	\item $a \le 2$, or
	\item $a = 3$ and $b \le 18$, or
	\item $a = 4$ and $b \le 11$, or
	\item $a = 5$ and $b \le 8$, or
	\item $a = 6$ and $b = 6$.
\end{inparaenum}
Further, $K_{a,b}$ is not 4-planar if $a\ge3$ and $b\ge19$.
\end{obs}

\myparagraph{The class of fan-planar graphs.}
\label{subsec:application:fanPlanar}
%
We start our discussion with complete graphs. The fact that a  fan-planar graph with $n$ vertices has at most $5n-10$ edges~\cite{KaufmannU14} implies that $K_9$ is not fan-planar, while Fig.7 in~\cite{BinucciGDMPST15} shows that $K_7$ is. With our implementation, we showed that $K_8$ is not fan-planar, even relaxing~the requirement that an edge crossed by two or more adjacent edges must~be~crossed from the same direction; see, e.g.,~\cite{DBLP:journals/jgaa/Brandenburg18}. This yields the following~characterization.

\begin{chr} \label{th:complete:fanplanar}
The complete graph $K_{n}$ is fan-planar if and only if $n \le 7$.
\end{chr}

Consider now a complete bipartite graph $K_{a,b}$ with $a \leq b$. For $a \leq 4$, $K_{a,b}$ is fan-planar for any value of $b$~\cite{KaufmannU14}. On the other hand, the fact that a bipartite fan-planar graph has at most $4n-12$ edges~\cite{DBLP:conf/isaac/AngeliniB0PU18} implies that $K_{5,9}$ is not fan-planar. Using our implementation, we could show that even $K_{5,5}$ is not fan-planar (again by relaxing the requirement of having the crossings from the same direction). These two results together imply the following characterization. 

\begin{chr} \label{th:bipartite:fanplanar}
The complete bipartite graph $K_{a,b}$ (with $a \le b$) is fan-planar if and only if $a \le 4$.
\end{chr}

\myparagraph{The class of fan-crossing free graphs.}
\label{subsec:application:fcf}
%
A characterization for the case of complete graphs can be derived by combining two known results. First, $K_6$ is fan-crossing free, as it is $1$-planar\arxapp{. We additionally show in~\cite{arxiv}}{; see Table~\ref{table:results} in Appendix~\ref{app:numbers} additionally shows} that, up to isomorphism, $K_6$ has a unique fan-crossing free drawing. Second, the fact that a fan-crossing free graph with $n$ vertices has at most $4n-8$ edges~\cite{DBLP:journals/algorithmica/CheongHKK15} implies that $K_7$ is not fan-crossing free. Hence, we have the following characterization.

\begin{chr}[Cheong et al.~\cite{DBLP:journals/algorithmica/CheongHKK15}, Czap et al.~\cite{DBLP:journals/dam/CzapH12}] \label{th:complete:fcf}
The complete\\graph $K_{n}$ is fan-crossing free if and only if $n \le 6$.
\end{chr}

As already stated, for the complete bipartite fan-crossing free graphs, we provide  \arxapp{in~\cite{arxiv}}{in Appendix~\ref{app:bipartite:fcf}} a combinatorial proof of their characterization. The same result was also obtained by our implementation\arxapp{.}{; see Table~\ref{table:results} in Appendix~\ref{app:numbers}.}

\begin{chr} \label{th:bipartite:fcf}
The complete bipartite graph $K_{a,b}$ (with $a \le b$) is fan-crossing free if and only if
\begin{inparaenum}[(i)]
	\item $a \le 2$, or
	\item $a \le 4$ and $b \le 6$.
\end{inparaenum}
\end{chr}

\myparagraph{The class of gap-planar graphs.}
\label{subsec:application:gapPlanar}
%
A characterization of the complete gap-planar graphs has already been provided~\cite{DBLP:journals/tcs/BaeBCEE0HKMRT18} as follows.
\begin{chr}[Bae et al.~\cite{DBLP:journals/tcs/BaeBCEE0HKMRT18}] \label{th:complete:gapplanar}
The complete graph $K_{n}$ is gap-planar if and only if $n \le 8$.
\end{chr}

For the case of complete bipartite graphs, Bae et al.~\cite{DBLP:journals/tcs/BaeBCEE0HKMRT18} proved that $K_{3,12}$, $K_{4,8}$, and $K_{5,6}$ are gap-planar, while $K_{3,15}$, $K_{4,11}$, and $K_{5,7}$ are not. These negative results were derived using the technique discussed in Section~\ref{sec:introduction} that compares the crossing number of these graphs with their number of edges, which is an upper bound to the number of crossings allowed in a gap-planar drawing. By refining this technique, Bachmaier et al.~\cite{DBLP:conf/gd/BachmaierRS18} proved that even $K_{3,14}$, $K_{4,10}$, and $K_{6,6}$ are not gap-planar. Hence, towards a characterization the cases that are left open are $K_{3,13}$ and $K_{4,9}$. Here, we address one of these two open cases by showing that $K_{4,9}$ is not gap-planar, thus yielding the following observation.

\begin{obs} \label{th:bipartite:gapplanar}
The complete bipartite graph $K_{a,b}$ (with $a \le b$) is gap-planar if
\begin{inparaenum}[(i)]
	\item\label{item:gapplanar:bipartite2} $a \le 2$, or
	\item\label{item:gapplanar:bipartite3} $a = 3$ and $b \le 12$, or
	\item\label{item:gapplanar:bipartite4} $a = 4$ and $b \le 8$, or
	\item\label{item:gapplanar:bipartite5} $a = 5$ and $b \le 6$.
\end{inparaenum}
Further, $K_{a,b}$ is not gap-planar if 
\begin{inparaenum}[(i)]
	\item\label{item:nongapplanar:bipartite3} $a = 3$ and $b \ge 14$, or
	\item\label{item:nongapplanar:bipartite4} $a = 4$ and $b \ge 9$, or
	\item\label{item:nongapplanar:bipartite5} $a = 5$ and $b \ge 7$, or 
	\item\label{item:nongapplanar:bipartite6} $a \ge 6$ and $b \ge 6$.
\end{inparaenum}
\end{obs}

\myparagraph{The class of quasiplanar graphs.}
\label{subsec:application:quasiPlanar}
%
A characterization for the complete quasiplanar graphs can be also derived by combining two known results. Namely, the fact that a quasiplanar graph with $n$ vertices has at most $6.5n-20$ edges~\cite{DBLP:journals/jct/AckermanT07} implies that $K_{11}$ is not quasiplanar, while  $K_{10}$ is in fact quasiplanar~\cite{franz-quasi-planar}.

\begin{chr}[Ackerman et al.~\cite{DBLP:journals/jct/AckermanT07}, Brandenburg~\cite{franz-quasi-planar}] \label{th:complete:quasiplanar}
The complete graph $K_{n}$ is quasiplanar if and only if $n \le 10$.
\end{chr}

Consider now a complete bipartite graph $K_{a,b}$ with $a \leq b$. First, we observe that for $a \leq 4$, graph $K_{a,b}$ is quasiplanar for any value of $b$, since it is even fan-planar~\cite{KaufmannU14}. On the other hand, the fact that a quasiplanar graph with $n$ vertices has at most $6.5n-20$ edges~\cite{DBLP:journals/jct/AckermanT07} does not provide any negative answer for $a \leq 6$, while for $a=7$ it only implies that $K_{7,52}$ is not quasiplanar. We stress that we were not able to find any improvement on the latter result. The reason is the same as the one that we described for the class of complete bipartite $4$-planar graphs. To give an intuition, we note that $K_{4,4}$ has in total 46711 non-isomorphic quasiplanar drawings, which makes the computation of the corresponding non-isomorphic drawings of $K_{4,5}$ infeasible in reasonable time; \arxapp{refer to~\cite{arxiv} for details}{see also Appendix~\ref{app:numbers}}. Notably, using the DFS-like variant of our algorithm, we were able to derive at least positive certificate drawings for $K_{5,18}$, $K_{6,10}$, and $K_{7,7}$\arxapp{, which are given in~\cite{arxiv}}{; see Figs.~\ref{fig:drawing:quasiplanar5}, \ref{fig:drawing:quasiplanar6}, and~\ref{fig:drawing:quasiplanar7} in Appendix~\ref{app:drawings}}. We summarize these findings in the following observation.

\begin{obs} \label{th:bipartite:quasiplanar}
The complete bipartite graph $K_{a,b}$ (with $a \le b$) is quasiplanar if
\begin{inparaenum}[(i)]
	\item\label{item:quasiplanar:bipartite4} $a \le 4$, or
	\item\label{item:quasiplanar:bipartite5} $a = 5$ and $b \le 18$, or
	\item\label{item:quasiplanar:bipartite6} $a = 6$ and $b \le 10$, or
	\item\label{item:quasiplanar:bipartite7} $a = 7$ and $b \le 7$.
\end{inparaenum}
Further, $K_{a,b}$ is not quasiplanar if $a\ge7$ and $b\ge52$.
\end{obs}

\begin{table}[t!]
  \caption{A comparison of the number of drawings reported by our algorithm with the elimination of isomorphic drawings (col.~``Non-Iso'') and without it (col.~``All'') for the classes of $1$- and $2$-planar graphs; the corresponding execution times (in sec.) to compute these drawings are reported next to them.}
  \label{table:comparison}
  \centering
  \medskip
  \resizebox{\columnwidth}{!}{
  \begin{tabular}{lc@{\hspace{.9em}}c@{\hspace{.9em}}r@{\hspace{.9em}}r@{\hspace{.9em}}c@{\hspace{.9em}}c@{\hspace{.9em}}c@{\hspace{.9em}}r@{\hspace{.9em}}r@{\hspace{.9em}}r@{\hspace{.9em}}}
    \toprule
     & 
     \multicolumn{5}{c}{complete} & \multicolumn{5}{c}{complete bipartite}\\
    \cmidrule(r{8pt}){2-6} \cmidrule(r{8pt}){7-11}
    Class    &   Graph & Non-Iso. & Time & All & Time & Graph & Non-Iso. & Time & All & Time \\
    \midrule
    1-planar & $K_{4}$ &   2 & 0.043 &   8 & 0.043 & $K_{2,3}$ &     3 &  0.061 &    34 &  0.061 \\
             & $K_{5}$ &   1 & 0.043 &  30 & 0.206 & $K_{3,3}$ &     2 &  0.049 &    84 &  0.539 \\
             & $K_{6}$ &   1 & 0.020 & 120 & 0.737 & $K_{3,4}$ &     3 &  0.065 &   960 &  5.642 \\
             & $K_{7}$ &   0 & 0.006 &   0 & 0.448 & $K_{4,4}$ &     2 &  0.044 &  1584 & 10.871 \\
             &         &     &       &     &       & $K_{4,5}$ &     0 &  0.010 &     0 &  7.198 \\
    \midrule
             &  total: &   4 & 0.112 & 158 & 1.434 &    total: &    10 &  0.229 &  2662 &  24.311\\[0.1ex]
    \midrule \midrule
    2-planar & $K_{4}$ &    2 &  0.028 &     8 &  0.028  & $K_{2,3}$ &   6 & 0.090 &      76 & 0.090 \\
             & $K_{5}$ &    4 &  0.105 &   294 &  2.661  & $K_{3,3}$ &  19 & 0.254 &    2352 & 10.571 \\
             & $K_{6}$ &    6 &  0.233 &  2664 &  3.292  & $K_{3,4}$ &  71 & 1.458 &   52248 & 244.964 \\
             & $K_{7}$ &    2 &  0.119 &  8400 & 55.323  & $K_{4,4}$ &  38 & 1.152 &  168624 & 1128.457 \\
             & $K_{8}$ &    0 &  0.029 &     0 & 51.321  & $K_{4,5}$ &  37 & 1.826 & 1200384 & 8135.843 \\
             &         &      &        &       &         & $K_{5,5}$ &   0 & 0.357 &       0 & 12639.293 \\
    \midrule
             &  total: &   14 &  0.514 & 11366 & 112.625 &    total: & 171 & 5.137 & 1423684 &  22159.218 \\[0.1ex]
    \bottomrule
  \end{tabular}
  }
\end{table}

\section{Conclusions and Open Problems}
\label{sec:conclusions}

%
We conclude this work by noting that our results also have some theoretical implications. In particular, $K_{5,5}$ was conjectured in~\cite{DBLP:conf/isaac/AngeliniB0PU18} not to be fan-planar; Characterization~\ref{th:bipartite:fanplanar} settles in the positive this conjecture. By Characterization~\ref{th:bipartite:fanplanar} and Observation~\ref{th:bipartite:gapplanar}, we deduce that $K_{5,5}$ is a certificate that there exist graphs which are gap-planar but not fan-planar. Since $K_{4,9}$ is fan-planar but not gap-planar, the two classes are incomparable, which answers a related question posed in~\cite{DBLP:journals/tcs/BaeBCEE0HKMRT18} about the relationship between 1-gap-planar graphs and fan-planar graphs.

We stress that the elimination of isomorphic drawings is a key step in our algorithm, as shown in Table~\ref{table:comparison}. For example, to test whether $K_{5,5}$ is $2$-planar without the elimination of intermediate isomorphic drawings, one would need to investigate 1423684 drawings, while in the presence of this step only 171. This significantly reduced the required time to roughly 5 seconds, including the time to perform all isomorphism tests and eliminations. We provide further insights in \arxapp{\cite{arxiv}}{Appendix~\ref{app:numbers}}, where we broaden our description to the other classes\arxapp{}{ (see Table~\ref{table:results})}. 

Our work leaves two open problems. Is it possible to extend our approach to graphs that are neither complete nor complete bipartite, e.g., to $k$-trees or to $k$-degenerate graphs (for small values of $k$)? A major difficulty is that, in the absence of symmetry, discarding isomorphic drawings becomes more complex. A general observation from our proof of concept is that our approach was of limited applicability on the classes of complete bipartite $k$-planar graphs, for $k>3$, and complete bipartite quasiplanar graphs, for which we could report partial results. So, is it possible to broaden these results by deriving improved upper bounds on the edge densities of these classes tailored for the bipartite setting (see, e.g.,~\cite{DBLP:conf/isaac/AngeliniB0PU18}).

\bibliography{abbrv,paper}
\bibliographystyle{splncs04}

\arxapp{}{\clearpage
\appendix

\section*{\Large Appendix}

\section{Preliminary Notions and Definitions}\label{app:preliminaries}

In this paper, we consider graphs containing neither multi-edges nor self-loops. Let $G=(V,E)$ be a graph. A \emph{drawing} of $G$ is a topological representation of $G$ in the plane $\mathbb R^2$ such that each vertex $v \in V$ is mapped to a distinct point $p_v$ of the plane, and each edge $(u,v)\in E$ is drawn as a simple Jordan curve connecting $p_u$ and $p_v$ without passing through any other vertex. Unless otherwise specified, we consider \emph{simple} drawings, in which any two edges intersect in at most one point, which is either a common endpoint or a proper crossing. Hence, two edges are not allowed to cross twice (or more times) and no two edges incident to the same vertex are allowed to cross. 

A drawing without edge crossings is called \emph{planar}. Accordingly, a graph that admits a planar drawing is called \emph{planar}. The \emph{planarization} of a (non-planar) drawing is the planar drawing obtained by replacing each of its crossings with a dummy vertex. The dummy vertices are referred to as \emph{crossing vertices}, while the remaining ones (that is, the ones of the original drawing) as \emph{real vertices}.  A planar drawing divides the plane into connected regions, called \emph{faces}; the unbounded one is called \emph{outer face}. The \emph{degree} of a face is defined as the number of edges on its boundary, counted with multiplicity. The \emph{dual} of a planar drawing $\Gamma$ has a node for each face of $\Gamma$ and an arc between two nodes if the corresponding faces of $\Gamma$ share an edge. Finally, two drawings $\Gamma_1$ and $\Gamma_2$ of a graph $G$ are \emph{isomorphic}~\cite{DBLP:journals/dcg/Kyncl11}, if there exists a homeomorphism of the sphere that transforms $\Gamma_1$ into $\Gamma_2$ (see Fig.~\ref{fig:k5} for an illustration), and \emph{weakly isomorphic}~\cite{DBLP:journals/dcg/Kyncl11} if there exists an incidence preserving bijection between their vertices and their edges, such that two edges of $\Gamma_1$ cross if and only if the corresponding edges of~$\Gamma_2$~cross.

\begin{figure}[h]
	\centering
	\subcaptionbox{\label{fig:k5-1}}{\includegraphics[scale=0.5,page=1]{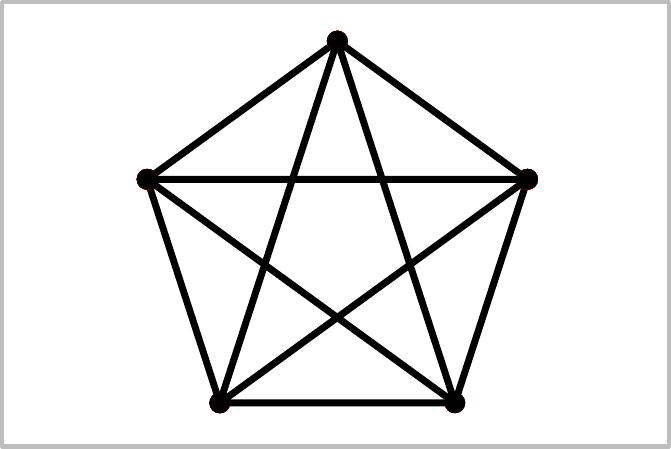}}
	~
	\subcaptionbox{\label{fig:k5-2}}{\includegraphics[scale=0.5,page=2]{k5}}
	~
	\subcaptionbox{\label{fig:k5-3}}{\includegraphics[scale=0.5,page=3]{k5}}
	\caption{Different drawings of $K_5$: 
	The drawing of (a) is isomorphic neither to the one of (b) nor to the one of (c), 
	while the drawings of (b) and~(c) are in fact isomorphic; 
	the colors of the vertices and the gray labels show the vertex and facial correspondences.}
\label{fig:k5}
\end{figure}

\clearpage

\section{A Running Example}\label{app:example}

In Fig.~\ref{fig:insertionExample}, we give an example for the insertion of a node $v$ into a crossing-free 4-cycle, such that $v$ is connected to two vertices $u_1$ and $u_2$. The dashed red edge is the newly inserted edge; the blue edges are prohibited; the turquoise edges are the edges that are marked as prohibited while computing the half-pathway of the red edge. Figs~\ref{subfig:exampleA}--\ref{subfig:exampleJ} illustrate all possible ways for drawing edge $(v,u_1)$. Figs~\ref{subfig:exampleK}--\ref{subfig:exampleO} illustrate all possible ways for inserting edge $(v,u_2)$ into the drawing of Fig.~\ref{subfig:exampleA}. Note that among the drawings that contain the edge $(v,u_2)$ the drawings of Figs.~\ref{subfig:exampleL} and~\ref{subfig:exampleN} are isomorphic, and the same holds for the drawings of Figs.~\ref{subfig:exampleM} and~\ref{subfig:exampleO}. Also, all obtained drawings are legal for the topological graph classes defined in the introduction, except for the class of $1$-planar graphs.

\begin{figure}[h!]
\centering
	\subcaptionbox{\label{subfig:exampleA}} {
	\includegraphics[page=1,scale=0.87]{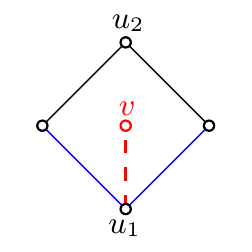}}
	\hfil
	\subcaptionbox{\label{subfig:exampleB}} {
	\includegraphics[page=2,scale=0.87]{figures/example}}
	\hfil
	\subcaptionbox{\label{subfig:exampleC}} {
	\includegraphics[page=3,scale=0.87]{figures/example}}
	\hfil
	\subcaptionbox{\label{subfig:exampleD}} {
	\includegraphics[page=4,scale=0.87]{figures/example}}
	\hfil
	\subcaptionbox{\label{subfig:exampleE}} {
	\includegraphics[page=5,scale=0.87]{figures/example}}
	\\
	\subcaptionbox{\label{subfig:exampleF}} {
	\includegraphics[page=6,scale=0.87]{figures/example}}
	\hfil
	\subcaptionbox{\label{subfig:exampleG}} {
	\includegraphics[page=7,scale=0.87]{figures/example}}
	\hfil
	\subcaptionbox{\label{subfig:exampleH}} {
	\includegraphics[page=8,scale=0.87]{figures/example}}
	\hfil
	\subcaptionbox{\label{subfig:exampleI}} {
	\includegraphics[page=9,scale=0.87]{figures/example}}
	\hfil
	\subcaptionbox{\label{subfig:exampleJ}} {
	\includegraphics[page=10,scale=0.87]{figures/example}}
	\\
	\subcaptionbox{\label{subfig:exampleK}} {
	\includegraphics[page=11,scale=0.87]{figures/example}}
	\hfil
	\subcaptionbox{\label{subfig:exampleL}} {
	\includegraphics[page=12,scale=0.87]{figures/example}}
	\hfil
	\subcaptionbox{\label{subfig:exampleM}} {
	\includegraphics[page=13,scale=0.87]{figures/example}}
	\hfil
	\subcaptionbox{\label{subfig:exampleN}} {
	\includegraphics[page=14,scale=0.87]{figures/example}}
	\hfil
	\subcaptionbox{\label{subfig:exampleO}} {
	\includegraphics[page=15,scale=0.87]{figures/example}}
	\caption{Illustration of a running example.}
	\label{fig:insertionExample}
\end{figure}

\section{Omitted Drawings from Section~\ref*{sec:applications}}\label{app:drawings}

In this section, we provide drawings certifying that certain complete and complete bipartite graphs belongs to specific beyond-planarity graph classes, which were omitted from Section~\ref{sec:applications} due to space constraints; refer to Figs.~\ref{fig:drawing:quasiplanar5} and~\ref{fig:drawing:quasiplanar67}.


\begin{figure}[h!]
\centering
	\includegraphics[page=1,scale=0.25]{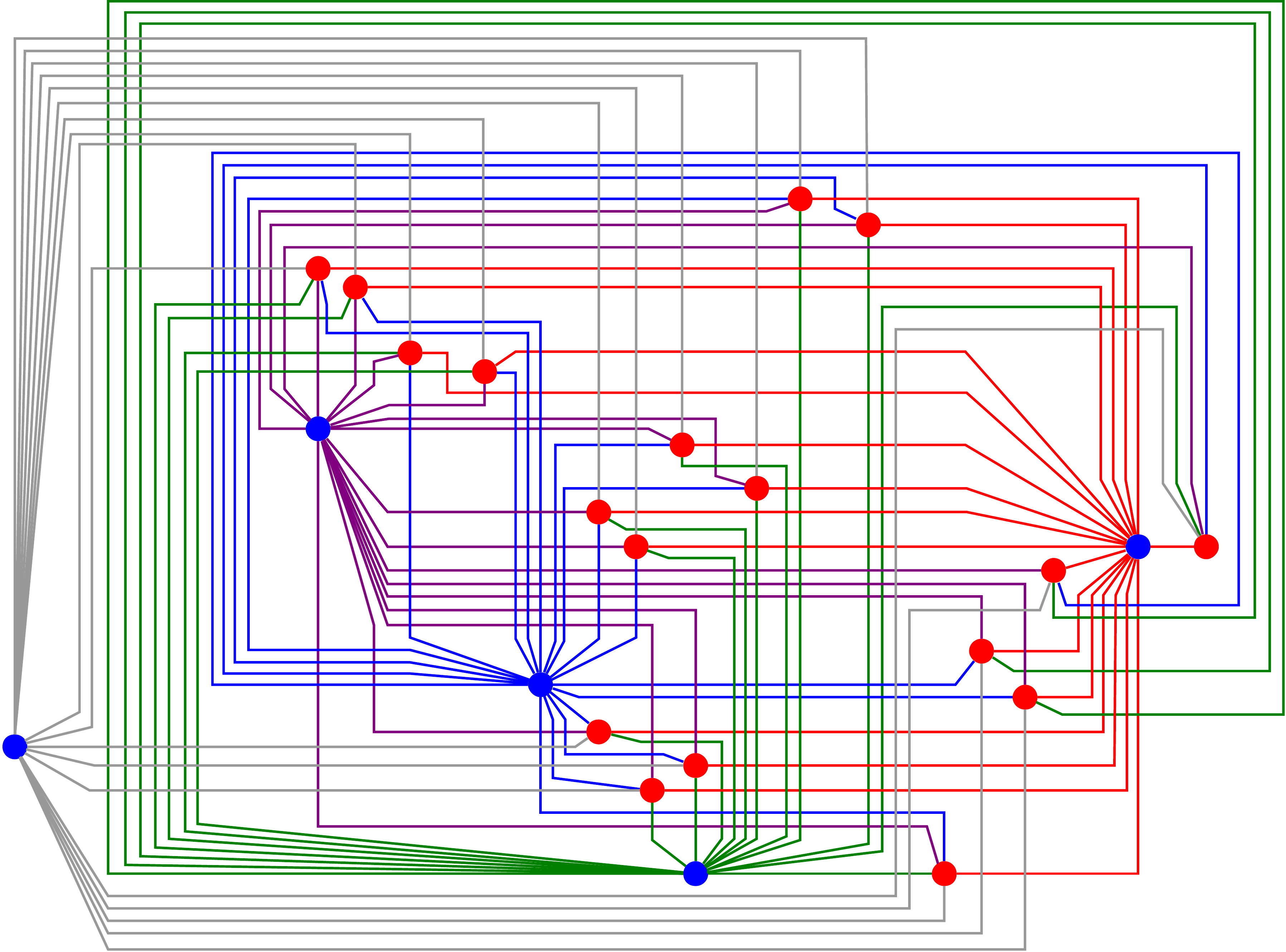}
	\caption{A quasiplanar drawing of $K_{5,18}$.}
	\label{fig:drawing:quasiplanar5}
\end{figure}

\begin{figure}[h!]
	\centering	
	\subcaptionbox{\label{fig:drawing:quasiplanar6}} {		
	\includegraphics[page=1,scale=0.22]{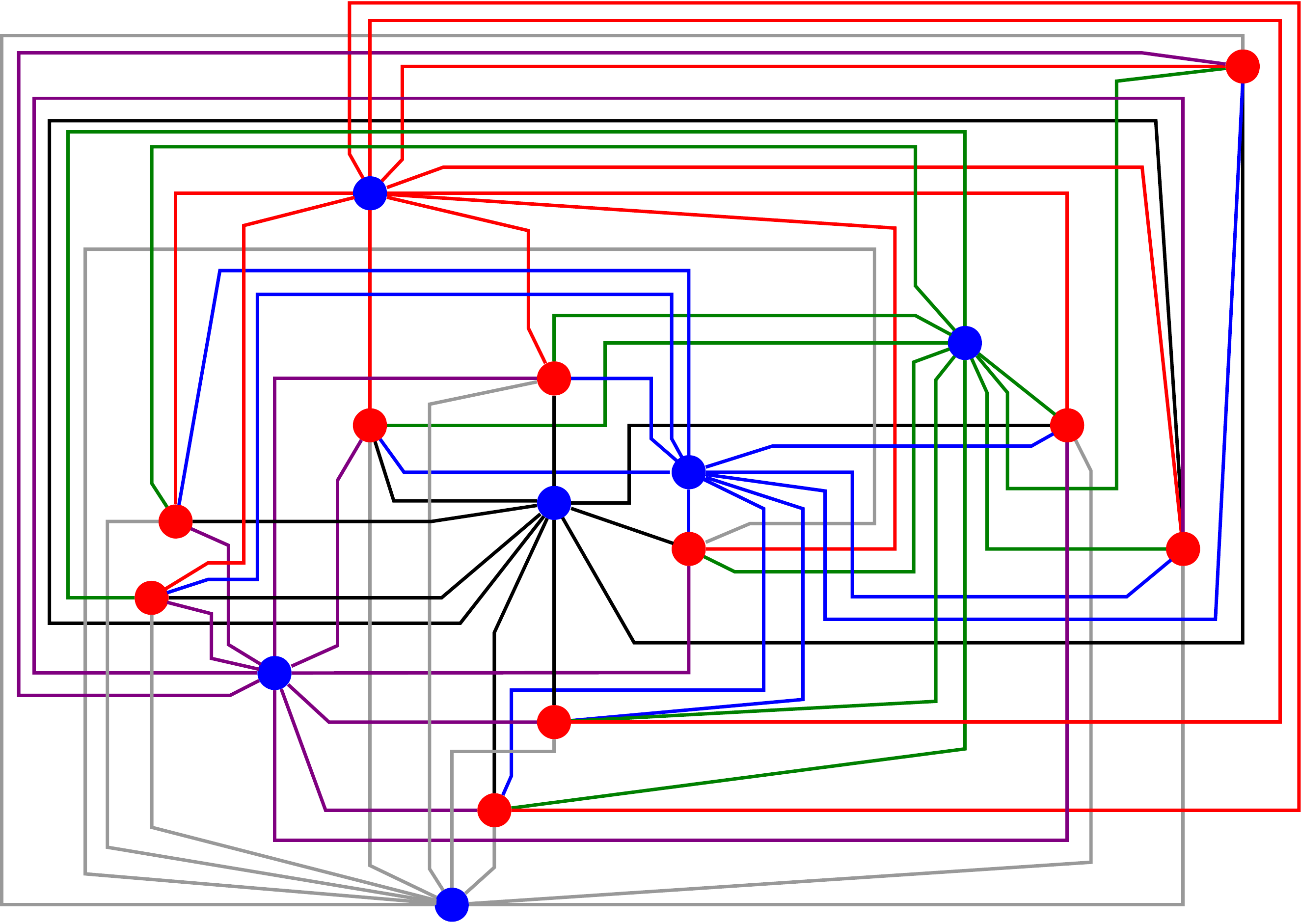}}
	\hfill
	\subcaptionbox{\label{fig:drawing:quasiplanar7}} {	
	\includegraphics[page=1,scale=0.22]{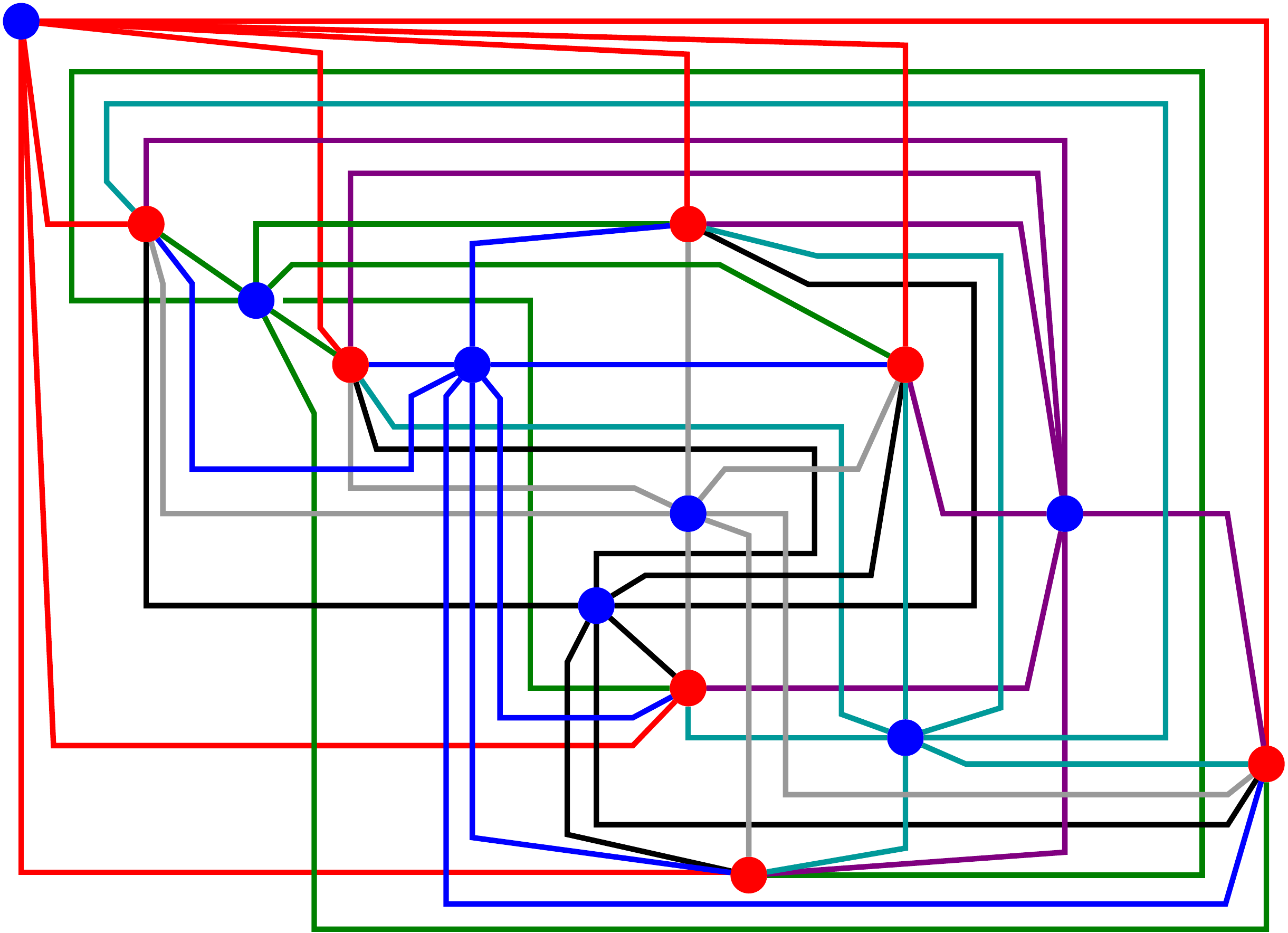}}	
	\caption{Illustration of quasiplanar drawings of (a)~$K_{6,10}$ and (b)~$K_{7,7}$.}
	\label{fig:drawing:quasiplanar67}
\end{figure}

\clearpage
\section{Further insights from our implementation}\label{app:numbers}

\begin{table}[p]
  \caption{A summary of the required time (in sec.) and of the number of general and non-isomorphic drawings for different complete and complete bipartite graphs.}
  \label{table:results}
  \subcaption*{Part A: Results concerning the classes of $k$-planar graphs; $k\in\{1,2,3,4\}$.}
  \medskip
  \centering
  \resizebox{\columnwidth}{!}{
  \begin{tabular}{lc@{\hspace{.9em}}r@{\hspace{.9em}}r@{\hspace{.9em}}r@{\hspace{.9em}}c@{\hspace{.9em}}r@{\hspace{.9em}}r@{\hspace{.9em}}r@{\hspace{.9em}}}
    \toprule
     & \multicolumn{4}{c}{complete} & \multicolumn{4}{c}{complete bipartite}\\
    \cmidrule(r{8pt}){2-5} \cmidrule(r{8pt}){6-9}
    Class &  Graph & General & Non-Iso. & Time & Graph & General & Non-Iso. & Time \\
    \midrule
    1-planar & $K_{4}$ &  8 & 2 & 0.043 & $K_{2,3}$ & 34 & 3 & 0.061 \\
             & $K_{5}$ & 13 & 1 & 0.043 & $K_{3,3}$ & 14 & 2 & 0.049 \\
             & $K_{6}$ &  4 & 1 & 0.020 & $K_{3,4}$ & 16 & 3 & 0.065 \\
             & $K_{7}$ &  0 & 0 & 0.006 & $K_{4,4}$ &  5 & 2 & 0.044 \\
             &         &    &   &       & $K_{4,5}$ &  0 & 0 & 0.010 \\
    \midrule
             & total:  & 25 & 4 & 0.112 &    total: & 69 & 10& 0.229 \\[0.1ex]
    \midrule \midrule
    2-planar & $K_{4}$ & 8  & 2 & 0.028 & $K_{2,3}$ &  76 &  6 & 0.090 \\
             & $K_{5}$ & 89 & 4 & 0.105 & $K_{3,3}$ & 243 & 19 & 0.254 \\
             & $K_{6}$ & 56 & 6 & 0.233 & $K_{3,4}$ & 526 & 71 & 1.458 \\
             & $K_{7}$ & 38 & 2 & 0.119 & $K_{4,4}$ & 310 & 38 & 1.152 \\
             & $K_{8}$ & 0  & 0 & 0.029 & $K_{4,5}$ & 318 & 37 & 1.826 \\
             &         &    &   &       & $K_{5,5}$ &   0 &  0 & 0.357 \\
    \midrule
             & total:  &191 &14 & 0.514 &    total: &1473 &171& 5.137 \\[0.1ex]
    \midrule \midrule
    3-planar & $K_{4}$ &   8 &  2 & 0.042 & $K_{2,3}$ &    76 &    6 &   0.234 \\
             & $K_{5}$ & 109 &  5 & 0.195 & $K_{3,3}$ &   678 &   69 &   1.802 \\
             & $K_{6}$ & 548 & 39 & 0.953 & $K_{3,4}$ &  7141 & 1188 &  16.969 \\
             & $K_{7}$ & 648 & 39 & 3.459 & $K_{4,4}$ & 24058 & 2704 &  97.801 \\
             & $K_{8}$ &  20 &  3 & 1.153 & $K_{4,5}$ & 44822 & 7653 & 310.194 \\
             & $K_{9}$ &   0 &  0 & 0.065 & \cellcolor{gray!25}$K_{5,5}$ & 20043 & 1899 & 199.908 \\
             &         &     &    &       & $K_{5,6}$ &  2516 &  438 &  47.396 \\
             &         &     &    &       & $K_{6,6}$ &     0 &    0 &   4.822 \\
    \midrule
             & total:  &1333 & 88 & 5.867 &    total: & 99334 &13957 & 679.126 \\[0.1ex]
    \midrule \midrule
    4-planar & $K_{4}$  &     8 &    2 &  0.040 & $K_{2,3}$ &     76 &     6 &     0.108 \\
             & $K_{5}$  &   109 &    5 &  0.222 & $K_{3,3}$ &    968 &   102 &     2.146 \\
             & $K_{6}$  &  1374 &   95 &  4.080 & $K_{3,4}$ &  32454 &  6194 &   163.000 \\
             & $K_{7}$  & 14728 & 1266 & 79.842 & $K_{4,4}$ & 681196 & 81817 & 34096.183 \\
             & $K_{8}$  &  7922 &  833 & 84.725 & $K_{4,5}$ &      ? &     ? &         ? \\
             & $K_{9}$  &   353 &   35 & 33.672 &           &        &       &           \\
             & $K_{10}$ &     0 &    0 &  1.175 &           &        &       &           \\
    \midrule
             & total:   & 24494 & 2236 &203.756 &    total: &      ? &     ? &         ? \\[0.1ex]
    \midrule \midrule
    5-planar & $K_{4}$  &       8 &       2 &      0.059 &  &  &  &  \\
             & $K_{5}$  &     109 &       5 &      0.259 &  &  &  &  \\
             & $K_{6}$  &    1752 &     119 &      4.716 &  &  &  &  \\
             & $K_{7}$  &   83710 &    8318 &   1396.781 &  &  &  &  \\
             & $K_{8}$  & 1190765 &  138750 & 262419.413 &  &  &  &  \\
             & $K_{9}$  &  285847 &   29939 &  32299.196 &  &  &  &  \\
             & $K_{10}$ &       0 &    0    &   2783.813 &  &  &  &  \\
		\midrule
             & total:   & 1562191 &  177133 & 298904.237 &  &  &  &  \\[0.1ex]
		\midrule
    \bottomrule
  \end{tabular}
  }
\end{table}

In this section, we present some insights from the computations that we made in order to check whether certain complete and complete bipartite graphs belong to specific graph classes; for a summary refer to Table~\ref{table:results}. Our algorithm was implemented in Java (\url{https://github.com/beyond-planarity/complete-graphs}) and was executed on a Windows machine with 2 cores at 2.9 GHz and~8~GB~RAM. 

As described in Section~\ref{sec:enumeration}, our algorithm constructs all possible drawings of a certain (complete or complete bipartite) graph by adding a single vertex to the non-isomorphic drawings of the subgraph of it without this vertex. Once a new drawing is obtained in this procedure, we compare it for isomorphism against the already computed ones (and possibly discard it). The total number of produced drawings is reported in the column ``General'', while the number of the non-isomorphic ones in the column ``Non-Iso.''. The reported times are in seconds and correspond to the total time needed for generation and filtering for isomorphism. The bottommost row of each section in the table corresponds to a negative instance, as no drawing satisfying the constraints of the respective graph class could be found. The class of complete bipartite $4$-planar graphs and the one of complete bipartite quasiplanar graphs form  exceptions, as for these classes we were not able to report all non-isomorphic drawings of $K_{4,5}$.

\begin{table}[p]
  \subcaption*{Part B: Results concerning the remaining graph classes considered in this paper.}
  \medskip
  \centering
  \resizebox{\columnwidth}{!}{
  \begin{tabular}{lc@{\hspace{.9em}}r@{\hspace{.9em}}r@{\hspace{.9em}}r@{\hspace{.9em}}c@{\hspace{.9em}}r@{\hspace{.9em}}r@{\hspace{.9em}}r@{\hspace{.9em}}}
    \toprule
     & \multicolumn{4}{c}{complete} & \multicolumn{4}{c}{complete bipartite}\\
    \cmidrule(r{8pt}){2-5} \cmidrule(r{8pt}){6-9}
    Class &  Graph & General & Non-Iso. & Time & Graph & General & Non-Iso. & Time \\
    \midrule
    fan-planar & $K_{4}$ &   8 &  2 & 0.034 & $K_{2,3}$ &  76 &  6 & 0.110 \\
               & $K_{5}$ &  89 &  5 & 0.133 & $K_{3,3}$ & 127 &  9 & 0.292 \\
               & $K_{6}$ & 147 & 39 & 0.226 & $K_{3,4}$ & 295 & 43 & 0.757 \\
               & $K_{7}$ &  75 & 39 & 0.405 & $K_{4,4}$ & 255 & 29 & 0.972 \\
               & $K_{8}$ &   0 &  0 & 0.196 & $K_{4,5}$ & 324 & 48 & 1.624 \\
               &         &     &    &       & $K_{5,5}$ &   0 &  0 & 0.637 \\
    \midrule
               & total:  & 319 & 22 & 0.994 &    total: &1077 &135& 4.392 \\[0.1ex]
    \midrule \midrule
    fan-crossing & $K_{4}$ &  8 & 2 & 0.049 & $K_{2,3}$ & 34 & 3 & 0.057 \\
    free         & $K_{5}$ & 13 & 1 & 0.054 & $K_{3,3}$ & 38 & 5 & 0.092 \\
                 & $K_{6}$ &  4 & 1 & 0.038 & $K_{3,4}$ & 28 & 5 & 0.098 \\
                 & $K_{7}$ &  0 & 0 & 0.009 & $K_{4,4}$ & 19 & 4 & 0.106 \\
                 &         &    &   &       & $K_{4,5}$ & 16 & 2 & 0.075 \\
                 &         &    &   &       & $K_{5,5}$ &  0 & 0 & 0.012 \\
    \midrule
                 & total:  & 25 & 4 & 0.150 &    total: &135 & 19& 0.440 \\[0.1ex]
    \midrule \midrule
    gap-planar  & $K_{4}$  &     14 &     2 &     0.135 & $K_{2,3}$ &     169 &     14 &      0.256 \\
                & $K_{5}$  &    243 &    10 &     0.366 & $K_{3,3}$ &    1425 &    266 &      4.359 \\
                & $K_{6}$  &    739 &   237 &     4.726 & $K_{3,4}$ &   16898 &   7466 &    170.396 \\
                & $K_{7}$  &   1124 &   665 &    13.943 & $K_{3,5}$ &  148527 &  56843 &  12032.226 \\
                & $K_{8}$  &      1 &     1 &    16.347 & $K_{4,5}$ &  199778 & 148367 &  28457.751 \\
                & $K_{9}$  &      0 &     0 &     0.019 & $K_{4,6}$ &  408476 & 246318 & 132622.664 \\
                &          &        &       &           & $K_{4,7}$ &  173271 & 101428 &  32958.628 \\
                &          &        &       &           & $K_{4,8}$ &    5981 &   4015 &   2708.278 \\
                &          &        &       &           & $K_{4,9}$ &       0 &      0 &     99.583 \\
    \midrule
                &   total: &   2121 &   915 &    35.536 &    total: &  954525 & 564717 & 209054.141 \\[0.1ex]
    \midrule \midrule
    quasiplanar & $K_{4}$  &      8 &     2 &     0.082 & $K_{2,3}$ &     76 &     6 &     0.187 \\
                & $K_{5}$  &    109 &     5 &     0.193 & $K_{3,3}$ &    604 &    53 &     0.859 \\
                & $K_{6}$  &    936 &    63 &     1.820 & $K_{3,4}$ &  11902 &  2248 &    34.073 \\
                & $K_{7}$  &  16505 &  1607 &    69.943 & $K_{4,4}$ & 386241 & 46711 & 11328.401 \\
                & $K_{8}$  & 173199 & 20980 &  4044.264 & $K_{4,5}$ &      ? &     ? &         ? \\
                & $K_{9}$  & 209248 & 23011 & 35163.772 &           &        &       &           \\
                & $K_{10}$ &     81 &     9 &  7593.865 &           &        &       &           \\
                & $K_{11}$ &      0 &     0 &     5.225 &           &        &       &           \\
		\midrule
                &   total: & 400086 & 45677 & 46879.164 &    total: &      ? &     ? &         ? \\[0.1ex]
    \midrule
    \bottomrule
  \end{tabular}
  }
\end{table}

As a typical example, we describe in the following one intermediate step in our computations; refer to the gray colored entry of Part~A of Table~\ref{table:results}. Our algorithm for reporting that $K_{6,6}$ is not a $3$-planar graph generated at some intermediate step all $3$-planar drawings of $K_{5,5}$, based on the non-isomorphic drawings of $K_{4,5}$. The algorithm reported in total 20043 drawings (including isomorphic ones), which were reduced to 1899 due to the elimination of isomorphic ones. These two steps together required 199.908 seconds. The obtained drawings were extended (by adding one additional vertex and its five incident edges) to 2516 drawings of $K_{5,6}$, which were reduced to 438 due to the filtering for isomorphism. None of these drawings could be extended to a $3$-planar drawing of $K_{6,6}$.

The class of complete bipartite $4$-planar graphs and the class of complete bipartite quasiplanar graphs show the limitations of our approach. We start our discussion with the former class. As already mentioned in Section~\ref{subsec:application:kPlanar}, for the class of complete bipartite $4$-planar graphs we were able to report only some partial results (and not a complete characterization). The reason is depicted in Part~A of Table~\ref{table:results}. Observe that, in order to determine the 81817 non-isomorphic drawings of $K_{4,4}$, our implementation needed to generate 681196 drawings starting from the 6194 non-isomorphic drawings of $K_{3,4}$. This growth in the number of non-isomorphic drawings and the time needed to generate them (i.e., 34096 sec.) form a clear indication of the reason why our implementation failed to report all corresponding drawings of $K_{4,5}$. Similar observations can be made for the class of quasiplanar graphs; see Part~B of Table~\ref{table:results}.

We conclude this section by making some additional observations. First, it is eye-catching from both parts of Table~\ref{table:results} that the number of general and non-isomorphic drawings of the complete graphs are significantly smaller than the corresponding ones for the complete bipartite graphs, which is explained by the fact that the former are very symmetric and denser. 

As it is naturally expected, we also observe that both the number of general drawings and the number of non-isomorphic drawings of a $k$-planar graph increases as $k$ increases (at least for values of $k$ in $\{1,2,3,4\}$). In particular, it seems that this increment becomes significantly larger from $3$- to $4$-planar graphs, both in the complete and in the complete bipartite settings. 

Comparing fan-planar and fan-crossing free graphs, which are in a sense complementary to each other, we observe significant differences in the number of general and non-isomorphic drawings. In particular, the number of non-isomorphic drawings of fan-crossing free graphs are always single digits. 

We finally observe that it is generally not a time-demanding task to conclude that a graph does not belong to a specific class, once all non-isomorphic drawings of its maximal realizable subgraph have been computed. In fact, the bottommost row of every section in Table~\ref{table:results} reports times in the order of few seconds at most.

\section{A Combinatorial Proof of Characterization~\ref*{th:bipartite:fcf}}
\label{app:bipartite:fcf}

In this appendix we give a combinatorial proof for Characterization~\ref{th:bipartite:fcf} of the complete bipartite fan-crossing free graphs.

\rephrase{Characterization}{\ref*{th:bipartite:fcf}}{The complete bipartite graph $K_{a,b}$ (with $a \le b$) is fan-crossing free if and only if
\begin{inparaenum}[(i)]
	\item $a \le 2$, or
	\item $a \le 4$ and $b \le 6$.
\end{inparaenum}
}

The sufficiency of the two conditions is proved by
observing that graph $K_{2,b}$, for any $b$, is planar, and that graph $K_{4,6}$ is fan-crossing free, as shown in Fig.~\ref{fig:drawing:2planar}.

To prove the necessity of the two conditions, we have to show that neither $K_{5,5}$ nor $K_{3,7}$
is fan-crossing free. We will discuss the two cases separately. However, we first observe some
properties that are common for the two cases. 
Recall that we restrict our analysis to simple drawings, in which 
there are no self-crossing edges, two distinct edges cross at most once, and adjacent edges do not
cross. For a complete bipartite graph $K_{a,b}$, we will denote the two partite sets as 
$V_1 = \{u_1,\dots,u_a\}$ and $V_2=\{w_1,\dots,w_b\}$.


The following lemma is important to bound the number of possible configurations 
we have to consider later in the two cases.

\begin{lemma}\label{lemma:fanCrossingFree:Planar22Subgraph}
	Let $\Gamma_{3,5}$ be a fan-crossing free drawing of $K_{3,5}$.
	There is a $K_{2,2}$-subgraph of $K_{3,5}$ whose edges do not 
	cross each other in $\Gamma_{3,5}$.
\end{lemma}

\begin{proof}
	Consider the $K_{2,2}$-subgraph $G$ induced by vertices $u_1, u_2, w_1, w_2$. If no two edges of $G$ cross each other in $\Gamma_{3,5}$,
	then the statement already follows. So, we may assume that two edges, say $(u_1,w_2)$ 
	and $(u_2,w_1)$, have a crossing. 
	
	We first consider the case in which another crossing exists between edges of $G$;
	since $\Gamma_{3,5}$ is simple and fan-crossing free, this crossing must be 
	between $(u_1,w_1)$ and $(u_2,w_2)$. However, it is possible to verify in 
	Fig~\ref{subfig:fanCrossingFree:noTwoCrossingsInKTwoTwoA}
	and Fig.~\ref{subfig:fanCrossingFree:noTwoCrossingsInKTwoTwoB} that it is
	not possible to realize these two crossings without creating a fan-crossing.
	
	\begin{figure}[b]
		\centering
		\subcaptionbox{\label{subfig:fanCrossingFree:noTwoCrossingsInKTwoTwoA}}{
			\includegraphics[page=1, width=0.3\textwidth]{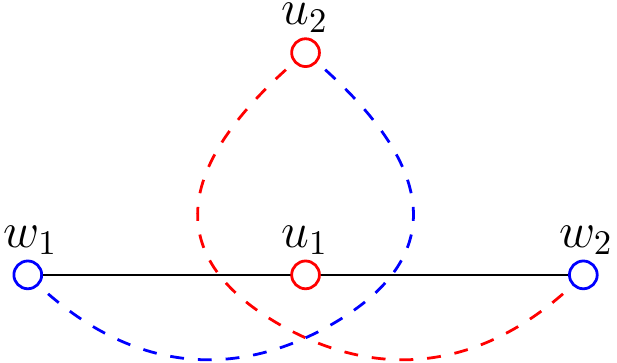}
		}
		\qquad
		\subcaptionbox{\label{subfig:fanCrossingFree:noTwoCrossingsInKTwoTwoB}}{
			\includegraphics[page=2, width=0.3\textwidth]{figures/FanCrFree-planarK22}
		}
		\caption{If $(u_2,w_1)$ crosses $(u_1,w_2)$ and $(u_1,w_1)$ crosses $(u_2,w_2)$, 
			then there exists also a fan-crossing.}
		\label{fig:fanCrossingFree:noTwoCrossingsInKTwoTwo}
	\end{figure}
	
	
	Thus, we may assume that there is no crossing between edges of $G$ other than
	the one involving $(u_1,w_2)$ and $(u_2,w_1)$, as in Fig.~\ref{subfig:fanCrossingFree:crossingFreeKTwoTwoA}, where the three regions
	$R_1$, $R_2$, and $R_3$, in which vertex $w_3$ may lie are shown.
	
	\begin{figure}[tb]
		\centering
		\hfill
		\subcaptionbox{\label{subfig:fanCrossingFree:crossingFreeKTwoTwoA}}{
			\includegraphics[page=3, scale=0.6]{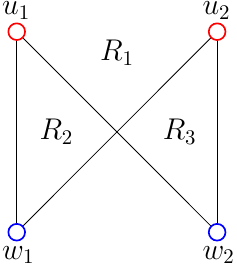}
		}
		\hfill
		\subcaptionbox{\label{subfig:fanCrossingFree:crossingFreeKTwoTwoB}}{
			\includegraphics[page=4, scale=0.6]{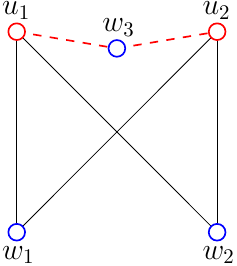}
		}
		\hfill
		\subcaptionbox{\label{subfig:fanCrossingFree:crossingFreeKTwoTwoC}}{
			\includegraphics[page=5, scale=0.6]{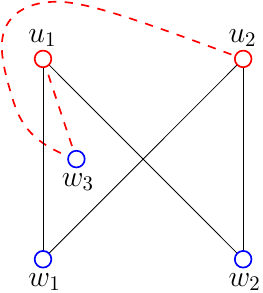}
		}
		\hfill~
		\caption{(a) The only drawing of $G$ in which
			$(u_1,w_2)$ and $(u_2,w_1)$ cross each other.
			(b) Vertex $w_3$ in region $R_1$. Edges $(u_1,w_3)$ and $(u_2,w_3)$
			(red edges)	are crossing-free.
			(c) Vertex $w_3$ in region $R_2$. Edge $(u_1,w_3)$ is crossing-free;
			edge $(u_2,w_3)$ crosses $(u_1,w_1)$.}
		\label{fig:fanCrossingFree:crossingFreeKTwoTwo}
	\end{figure}
	
	First we consider the case in which $w_3$ lies in $R_1$ (see Fig.~\ref{subfig:fanCrossingFree:crossingFreeKTwoTwoB}).
	The edge $(u_1,w_3)$ can neither cross $(u_1,w_1)$ nor $(u_1,w_2)$,
	since adjacent edges do not cross; it cannot cross $(u_2,w_1)$,
	as otherwise $(u_2,w_1)$ would cross a fan incident to $u_1$; 
	finally, $(u_1,w_3)$ cannot cross $(u_2,w_2)$ without crossing any 
	other edge. Hence, edge $(u_1,w_3)$ does not cross any edge of $G$,
	and the same holds for edge $(u_2,w_3)$. Further, $(u_1,w_3)$ and $(u_2,w_3)$ 
	do not cross each other, since they are adjacent. However, this implies that
	there exist even two $K_{2,2}$-subgraphs whose edges do not cross each other
	in $\Gamma_{3,5}$, namely the one induced by $u_1,u_2,w_3,w_1$, and the one 
	induced by $u_1,u_2,w_3,w_2$.
	
	We now consider the case in which $w_3$ lies in region $R_2$
	(see Fig.~\ref{subfig:fanCrossingFree:crossingFreeKTwoTwoC});
	note that the case in which $w_3$ lies in $R_3$ is symmetric.
	Similar to the previous case, edge $(u_1,w_3)$ cannot cross any
	edge of $G$. Further, the edge $(u_2,w_3)$ can neither cross one of the edges
	$(u_2,w_1)$ or $(u_2,w_2)$ (as adjacent edges do not cross),
	nor $(u_1,w_2)$ (this would be a fan-crossing). Hence, the only possible
	crossing of edge $(u_2,w_3)$ is with $(u_1,w_1)$, as in Fig.~\ref{subfig:fanCrossingFree:crossingFreeKTwoTwoC}. In this case,
	the graph induced by $u_1,u_2,w_2,w_3$ is a $K_{2,2}$-subgraph whose edges
	do not cross each other. 
	
	Since we considered all cases, the statement of the lemma follows.
\end{proof}

By Lemma~\ref{lemma:fanCrossingFree:Planar22Subgraph},
we assume in the following that in any fan-crossing free drawing of
either $K_{3,7}$ or $K_{5,5}$, the $K_{2,2}$-subgraph 
induced by vertices $u_1,u_2,w_1,w_2$ is such that no two of its
edges cross each other. We use the planar drawing of this
$K_{2,2}$-subgraph as a starting point to construct all possible fan-crossing free
drawings of graph $K_{2,5}$, by adding 
vertices $w_3$, $w_4$, and $w_5$ one at a time.

\smallskip

\begin{figure}[htpb]
	\centering
	\subcaptionbox{\label{subfig:fanCrossingFree:addThirdNode:CaseA}}{
		\includegraphics[page=1, width=0.24\textwidth]{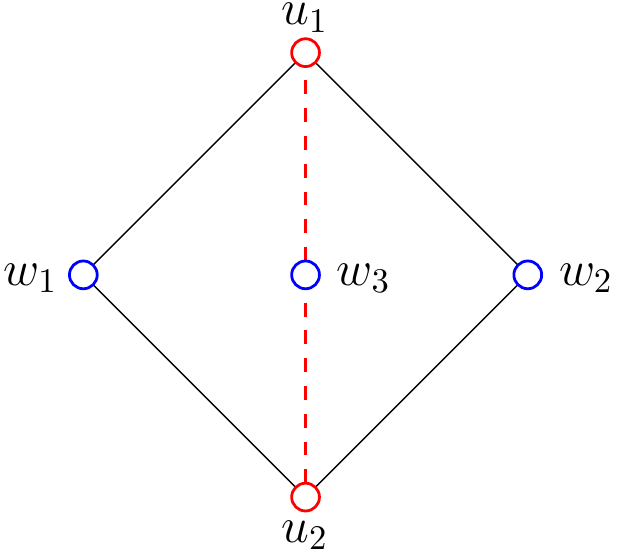}
	}
	\subcaptionbox{\label{subfig:fanCrossingFree:addThirdNode:CaseB}}{
		\includegraphics[page=2, width=0.22\textwidth]{figures/FanCrFree-addThirdNode}
	}
	\subcaptionbox{\label{subfig:fanCrossingFree:addThirdNode:CaseC}}{
		\includegraphics[page=3, width=0.22\textwidth]{figures/FanCrFree-addThirdNode}
	}
	\subcaptionbox{\label{subfig:fanCrossingFree:addThirdNode:CaseD}}{
		\includegraphics[page=4, width=0.22\textwidth]{figures/FanCrFree-addThirdNode}
	}\\
	\subcaptionbox{\label{subfig:fanCrossingFree:addThirdNode:CaseE}}{
		\includegraphics[page=5, width=0.22\textwidth]{figures/FanCrFree-addThirdNode}
	}
	\subcaptionbox{\label{subfig:fanCrossingFree:addThirdNode:CaseF}}{
		\includegraphics[page=6, width=0.22\textwidth]{figures/FanCrFree-addThirdNode}
	}
	\subcaptionbox{\label{subfig:fanCrossingFree:addThirdNode:CaseG}}{
		\includegraphics[page=7, width=0.22\textwidth]{figures/FanCrFree-addThirdNode}
	}
	\subcaptionbox{\label{subfig:fanCrossingFree:addThirdNode:CaseH}}{
		\includegraphics[page=8, width=0.22\textwidth]{figures/FanCrFree-addThirdNode}
	}\\
	\subcaptionbox{\label{subfig:fanCrossingFree:addThirdNode:CaseI}}{
		\includegraphics[page=9, width=0.22\textwidth]{figures/FanCrFree-addThirdNode}
	}
	\subcaptionbox{\label{subfig:fanCrossingFree:addThirdNode:CaseJ}}{
		\includegraphics[page=10, width=0.22\textwidth]{figures/FanCrFree-addThirdNode}
	}
	\subcaptionbox{\label{subfig:fanCrossingFree:addThirdNode:CaseK}}{
		\includegraphics[page=11, width=0.22\textwidth]{figures/FanCrFree-addThirdNode}
	}
	\subcaptionbox{\label{subfig:fanCrossingFree:addThirdNode:CaseL}}{
		\includegraphics[page=12, width=0.22\textwidth]{figures/FanCrFree-addThirdNode}
	}\\
	\subcaptionbox{\label{subfig:fanCrossingFree:addThirdNode:CaseM}}{
		\includegraphics[page=13, width=0.22\textwidth]{figures/FanCrFree-addThirdNode}
	}
	\subcaptionbox{\label{subfig:fanCrossingFree:addThirdNode:CaseN}}{
		\includegraphics[page=14, width=0.22\textwidth]{figures/FanCrFree-addThirdNode}
	}
	\subcaptionbox{\label{subfig:fanCrossingFree:addThirdNode:CaseO}}{
		\includegraphics[page=15, width=0.22\textwidth]{figures/FanCrFree-addThirdNode}
	}
	\subcaptionbox{\label{subfig:fanCrossingFree:addThirdNode:CaseP}}{
		\includegraphics[page=16, width=0.22\textwidth]{figures/FanCrFree-addThirdNode}
	}
	\caption{All the cases that preserve the fan-crossing free property,
		when adding a third node $w_3$ (red) to the plane $K_{2,2}$. The red edges
		indicate the newly added edges.}
	\label{fig:fanCrossingFree:addThirdNode}
\end{figure}

\noindent\textbf{Adding vertex $w_3$.} 
We consider vertex $w_3$ and edges $(u_1,w_3)$ and $(u_2,w_3)$.
By our previous observations for fan-crossing free drawings,
edge $(u_1,w_3)$ is not allowed to cross $(u_1,w_1)$, $(u_1,w_2)$
and $(u_2,w_3)$. However, a crossing of $(u_1,w_3)$ with at most one
of $(u_2,w_1)$ and $(u_2,w_2)$ is possible. For edge
$(u_2,w_3)$, we can observe that this edge is not allowed to cross
$(u_2,w_1)$, $(u_2,w_2)$, and $(u_1,w_3)$, but is allowed to cross
at most one of edges $(u_1,w_1)$ and $(u_1,w_2)$.
Fig.~\ref{fig:fanCrossingFree:addThirdNode} shows all the possible
drawings that can occur when we add the node $w_3$.

Note that many configurations in Fig.~\ref{fig:fanCrossingFree:addThirdNode}
are topologically equivalent, namely:
\begin{itemize}
	\item Fig.\ref{subfig:fanCrossingFree:addThirdNode:CaseA} and
	Fig.\ref{subfig:fanCrossingFree:addThirdNode:CaseJ},
	\item Fig.\ref{subfig:fanCrossingFree:addThirdNode:CaseB},
	Fig.\ref{subfig:fanCrossingFree:addThirdNode:CaseD},
	Fig.\ref{subfig:fanCrossingFree:addThirdNode:CaseF},
	Fig.\ref{subfig:fanCrossingFree:addThirdNode:CaseH},
	Fig.\ref{subfig:fanCrossingFree:addThirdNode:CaseK},
	Fig.\ref{subfig:fanCrossingFree:addThirdNode:CaseM},
	Fig.\ref{subfig:fanCrossingFree:addThirdNode:CaseN} and
	Fig.\ref{subfig:fanCrossingFree:addThirdNode:CaseP},
	\item Fig.\ref{subfig:fanCrossingFree:addThirdNode:CaseC},
	Fig.\ref{subfig:fanCrossingFree:addThirdNode:CaseE},
	Fig.\ref{subfig:fanCrossingFree:addThirdNode:CaseG},
	Fig.\ref{subfig:fanCrossingFree:addThirdNode:CaseI},
	Fig.\ref{subfig:fanCrossingFree:addThirdNode:CaseL} and
	Fig.\ref{subfig:fanCrossingFree:addThirdNode:CaseO}.
\end{itemize}
So there are basically three different configurations that we
have to consider, namely those from
Figs.~\ref{subfig:fanCrossingFree:addThirdNode:CaseA},
\ref{subfig:fanCrossingFree:addThirdNode:CaseB} and
\ref{subfig:fanCrossingFree:addThirdNode:CaseC}.

\medskip

\noindent\textbf{Adding vertex $w_4$.} 
In the next step we use these three configurations to create all drawings with the
additional node $w_4$ and the edges $(w_4,u_1)$ and $(w_4,u_2)$.

First, we consider Fig.~\ref{subfig:fanCrossingFree:addThirdNode:CaseA}.
In this drawing, the edge $(u_1,w_4)$ is not allowed to cross 
edges $(u_1,w_1)$, $(u_1,w_2)$, $(u_1,w_3)$ and $(u_2,w_4)$, but it may
cross at most one of the edges $(u_2,w_1)$, $(u_2,w_2)$ or $(u_2,w_3)$. 
Similar, the edge $(u_2,w_4)$ cannot cross $(u_2,w_1)$, $(u_2,w_2)$, $(u_2,w_3)$ 
and $(u_1,w_4)$, but it may cross at most one of the edges $(u_1,w_1)$, 
$(u_1,w_2)$ or $(u_1,w_3)$. From these conditions we obtain the three 
topologically different drawings in Fig.~\ref{fig:fanCrossingFree:addFourthNodeA}.

\begin{figure}[tb]
	\centering
	\subcaptionbox{\label{subfig:fanCrossingFree:addFourthNode:CaseAA}}{
		\includegraphics[page=1, width=0.22\textwidth]{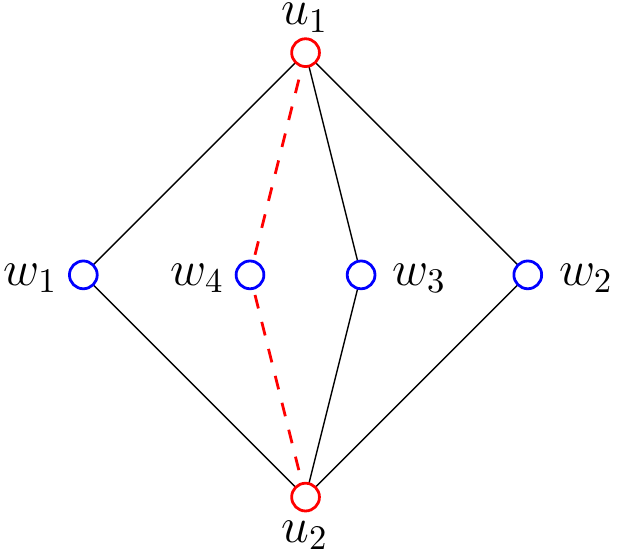}
	}
	\subcaptionbox{\label{subfig:fanCrossingFree:addFourthNode:CaseAB}}{
		\includegraphics[page=2, width=0.22\textwidth]{figures/FanCrFree-addFourthNode}
	}
	\subcaptionbox{\label{subfig:fanCrossingFree:addFourthNode:CaseAC}}{
		\includegraphics[page=3, width=0.22\textwidth]{figures/FanCrFree-addFourthNode}
	}
	\caption{All (topologically distinct) cases that preserve the fan-crossing free property,
		when adding a fourth node $w_4$ to the drawing from
		Fig.~\ref{subfig:fanCrossingFree:addThirdNode:CaseA}.}
	\label{fig:fanCrossingFree:addFourthNodeA}
\end{figure}

Now we consider Fig.~\ref{subfig:fanCrossingFree:addThirdNode:CaseB}.
In this drawing, the edge $(u_1,w_4)$ is not allowed to cross 
edges $(u_1,w_1)$, $(u_1,w_2)$, $(u_1,w_3)$, $(u_2,w_4)$ and $(u_2,w_1)$, but it
may cross at most one of the two edges $(u_2,w_2)$ or $(u_2,w_3)$. 
Further, edge $(u_2,w_4)$ is not allowed to cross $(u_2,w_1)$, $(u_2,w_2)$, 
$(u_2,w_3)$, $(u_1,w_4)$, and $(u_1,w_3)$, but it may cross at most one of the edges
$(u_1,w_1)$ or $(u_1,w_2)$. From these conditions we obtain the seven drawings in Fig.~\ref{subfig:fanCrossingFree:addFourthNode:CaseBA}--\ref{subfig:fanCrossingFree:addFourthNode:CaseBG}.

\begin{figure}[htpb]
	\centering
	\subcaptionbox{\label{subfig:fanCrossingFree:addFourthNode:CaseBA}}{
		\includegraphics[page=4, width=0.22\textwidth]{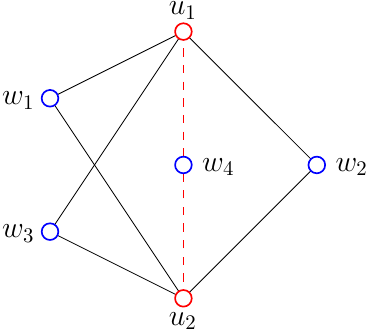}
	}
	\subcaptionbox{\label{subfig:fanCrossingFree:addFourthNode:CaseBB}}{
		\includegraphics[page=5, width=0.22\textwidth]{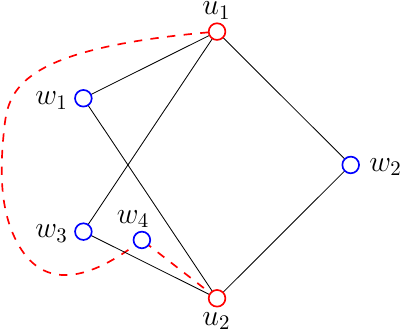}
	}
	\subcaptionbox{\label{subfig:fanCrossingFree:addFourthNode:CaseBC}}{
		\includegraphics[page=6, width=0.22\textwidth]{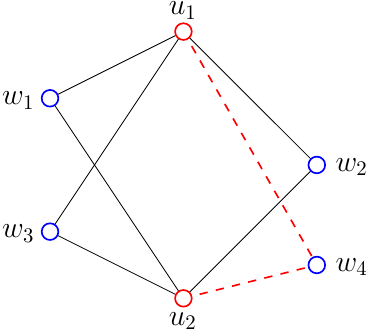}
	}
	\subcaptionbox{\label{subfig:fanCrossingFree:addFourthNode:CaseBD}}{
		\includegraphics[page=7, width=0.22\textwidth]{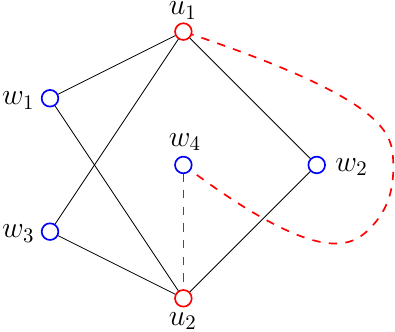}
	}\\
	\subcaptionbox{\label{subfig:fanCrossingFree:addFourthNode:CaseBE}}{
		\includegraphics[page=8, width=0.22\textwidth]{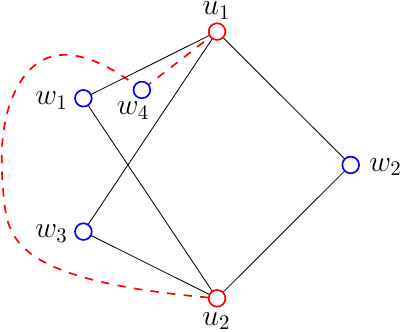}
	}
	\subcaptionbox{\label{subfig:fanCrossingFree:addFourthNode:CaseBF}}{
		\includegraphics[page=9, width=0.22\textwidth]{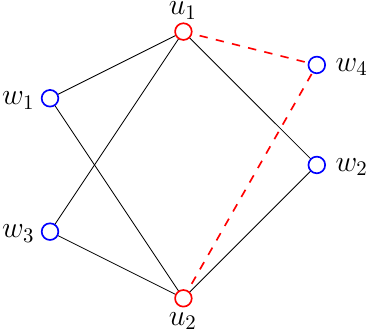}
	}
	\subcaptionbox{\label{subfig:fanCrossingFree:addFourthNode:CaseBG}}{
		\includegraphics[page=10, width=0.22\textwidth]{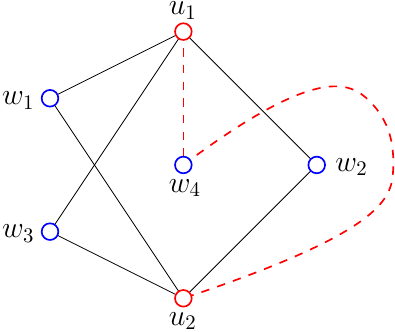}
	}
	\subcaptionbox{\label{fig:fanCrossingFree:addFourthNodeC}}{
		\includegraphics[page=11, width=0.22\textwidth]{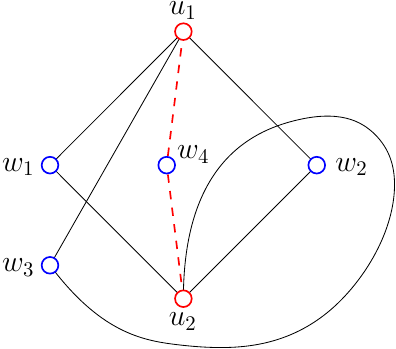}
	}
	\caption{(a)--(g) The cases that preserve the fan-crossing free property,
		when adding a fourth node $w_4$ to the drawing from
		Fig.~\ref{subfig:fanCrossingFree:addThirdNode:CaseB}.
		(h)~The only case that preserves the fan-crossing free property,
		when adding a fourth node $w_4$ to the drawing from
		Fig.~\ref{subfig:fanCrossingFree:addThirdNode:CaseC}.
	}
	\label{fig:fanCrossingFree:addFourthNodeB}
\end{figure}

In the last step we consider Fig.~\ref{subfig:fanCrossingFree:addThirdNode:CaseC}.
In this drawing, the edge $(u_1,w_4)$ is not allowed to cross edges $(u_1,w_1)$, $(u_1,w_2)$, $(u_1,w_3)$, $(u_2,w_4)$, $(u_2,w_1)$ and $(u_2,w_3)$, and so it may only cross $(u_2,w_2)$.
Further, the edge $(u_2,w_4)$ is not allowed to cross $(u_2,w_1)$, $(u_2,w_2)$, $(u_2,w_3)$,
$(u_1,w_4)$, $(u_1,w_2)$ and $(u_1,w_3)$, and so it may only cross $(u_1,w_1)$.
From these conditions we obtain only one single drawing, namely the one shown in
Fig.~\ref{fig:fanCrossingFree:addFourthNodeC}.

We observe that the configurations from the following figures are topologically equivalent:
\begin{itemize}
	\item Fig.\ref{subfig:fanCrossingFree:addFourthNode:CaseAB} and
	Fig.\ref{subfig:fanCrossingFree:addFourthNode:CaseBA},
	\item Fig.\ref{subfig:fanCrossingFree:addFourthNode:CaseAC},
	Fig.\ref{subfig:fanCrossingFree:addFourthNode:CaseBB},
	Fig.\ref{subfig:fanCrossingFree:addFourthNode:CaseBE} and
	Fig.\ref{fig:fanCrossingFree:addFourthNodeC},
	\item Fig.\ref{subfig:fanCrossingFree:addFourthNode:CaseBC} and
	Fig.\ref{subfig:fanCrossingFree:addFourthNode:CaseBF},
	\item Fig.\ref{subfig:fanCrossingFree:addFourthNode:CaseBD} and
	Fig.\ref{subfig:fanCrossingFree:addFourthNode:CaseBG}.
\end{itemize}
So there are basically five different configurations that we
have to consider, namely those from
Figs.~\ref{subfig:fanCrossingFree:addFourthNode:CaseAA},
\ref{subfig:fanCrossingFree:addFourthNode:CaseAB},
\ref{subfig:fanCrossingFree:addFourthNode:CaseAC},
\ref{subfig:fanCrossingFree:addFourthNode:CaseBC} and
\ref{subfig:fanCrossingFree:addFourthNode:CaseBD}.

\medskip

\noindent\textbf{Adding vertex $w_5$.} 
In the next step we use these five configurations to create all drawings with the
additional node $w_5$ and the edges $(w_5,u_1)$ and $(w_5,u_2)$.

First, we consider Fig.~\ref{subfig:fanCrossingFree:addFourthNode:CaseAA}.
In this drawing, the edge $(u_1,w_5)$ is not allowed to cross 
edges $(u_1,w_1)$, $(u_1,w_2)$, $(u_1,w_3)$, $(u_1,w_4)$ and $(u_2,w_5)$, but it
may cross at most one of the edges
$(u_2,w_1)$, $(u_2,w_2)$, $(u_2,w_3)$ or $(u_2,w_4)$. Similarly, the edge $(u_2,w_5)$
cannot cross $(u_2,w_1)$, $(u_2,w_2)$, $(u_2,w_3)$, $(u_2,w_4)$ and $(u_1,w_5)$, but
it may cross at most one of the edges $(u_1,w_1)$, $(u_1,w_2)$, $(u_1,w_3)$ or $(u_1,w_4)$.
From this, we obtain the three topologically different drawings in
Fig.~\ref{fig:fanCrossingFree:addFifthNodeAA}.

\begin{figure}[htpb]
	\centering
	\subcaptionbox{\label{subfig:fanCrossingFree:addFifthNode:CaseAAA}}{
		\includegraphics[page=1, width=0.22\textwidth]{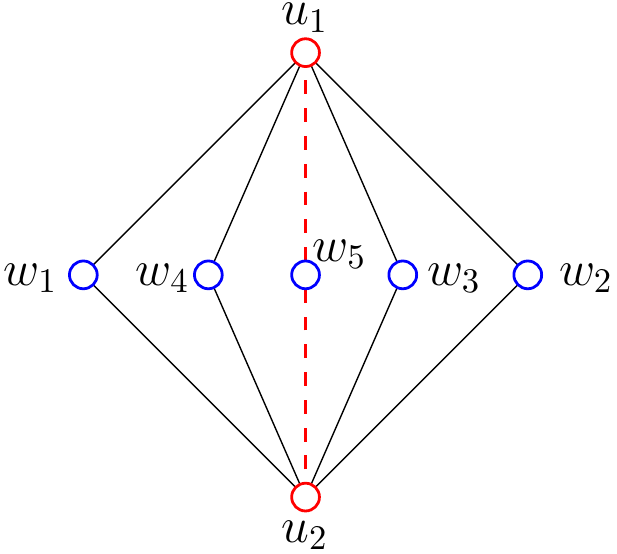}
	}
	\subcaptionbox{\label{subfig:fanCrossingFree:addFifthNode:CaseAAB}}{
		\includegraphics[page=2, width=0.22\textwidth]{figures/FanCrFree-addFifthNode}
	}
	\subcaptionbox{\label{subfig:fanCrossingFree:addFifthNode:CaseAAC}}{
		\includegraphics[page=3, width=0.22\textwidth]{figures/FanCrFree-addFifthNode}
	}
	\caption{All (topologically different) cases that preserve the fan-crossing free property,
		when adding a fifth node $w_5$ to the drawing from
		Fig.~\ref{subfig:fanCrossingFree:addFourthNode:CaseAA}.}
	\label{fig:fanCrossingFree:addFifthNodeAA}
\end{figure}

Second, we consider Fig.~\ref{subfig:fanCrossingFree:addFourthNode:CaseAB}.
In this drawing, the edge $(u_1,w_5)$ is not allowed to cross
edges $(u_1,w_1)$, $(u_1,w_2)$, $(u_1,w_3)$, $(u_1,w_4)$, $(u_2,w_5)$ and $(u_2,w_1)$, but
it may cross at most one of the edges
$(u_2,w_2)$, $(u_2,w_3)$ or $(u_2,w_4)$. Further, the edge $(u_2,w_5)$
cannot cross $(u_2,w_1)$, $(u_2,w_2)$, $(u_2,w_3)$, $(u_2,w_4)$, $(u_1,w_5)$ and $(u_1,w_4)$, but
it may cross at most one of the edges $(u_1,w_1)$, $(u_1,w_2)$ or $(u_1,w_3)$.
From this, we obtain the drawings in Fig.~\ref{fig:fanCrossingFree:addFifthNodeAB}.

\begin{figure}[tb]
	\centering
	\subcaptionbox{\label{subfig:fanCrossingFree:addFifthNode:CaseABA}}{
		\includegraphics[page=4, width=0.22\textwidth]{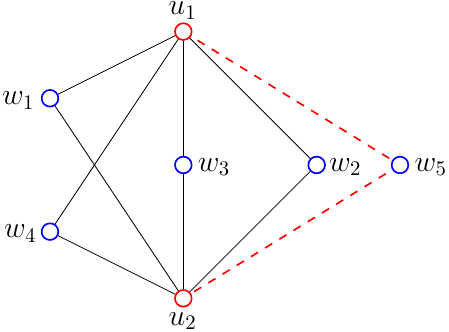}
	}
	\subcaptionbox{\label{subfig:fanCrossingFree:addFifthNode:CaseABB}}{
		\includegraphics[page=5, width=0.22\textwidth]{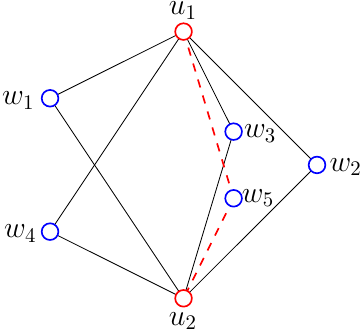}
	}
	\subcaptionbox{\label{subfig:fanCrossingFree:addFifthNode:CaseABC}}{
		\includegraphics[page=6, width=0.22\textwidth]{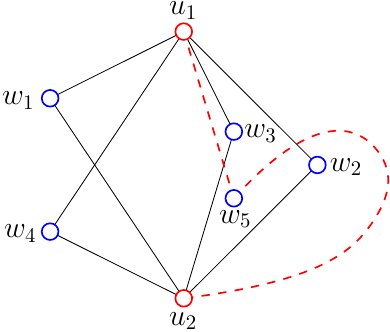}
	}
	\subcaptionbox{\label{subfig:fanCrossingFree:addFifthNode:CaseABD}}{
		\includegraphics[page=7, width=0.22\textwidth]{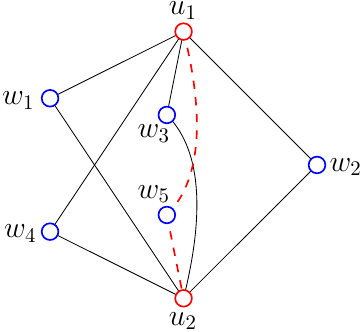}
	}\\
	\subcaptionbox{\label{subfig:fanCrossingFree:addFifthNode:CaseABE}}{
		\includegraphics[page=8, width=0.22\textwidth]{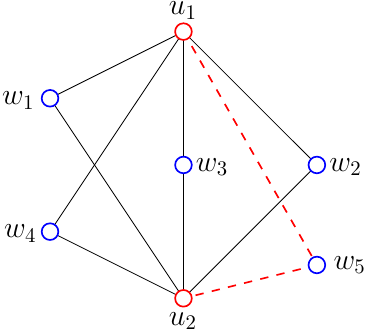}
	}
	\subcaptionbox{\label{subfig:fanCrossingFree:addFifthNode:CaseABF}}{
		\includegraphics[page=9, width=0.22\textwidth]{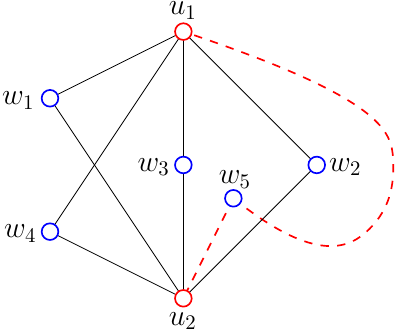}
	}
	\subcaptionbox{\label{subfig:fanCrossingFree:addFifthNode:CaseABG}}{
		\includegraphics[page=10, width=0.22\textwidth]{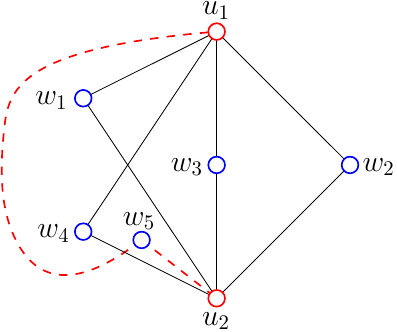}
	}
	\subcaptionbox{\label{subfig:fanCrossingFree:addFifthNode:CaseABH}}{
		\includegraphics[page=11, width=0.22\textwidth]{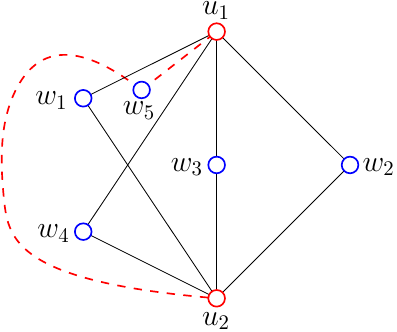}
	}\\
	\subcaptionbox{\label{subfig:fanCrossingFree:addFifthNode:CaseABI}}{
		\includegraphics[page=12, width=0.22\textwidth]{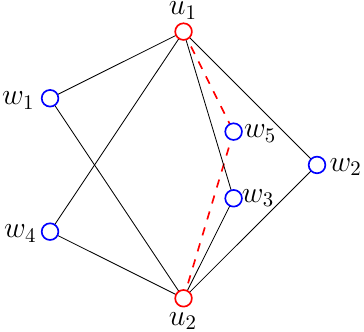}
	}
	\subcaptionbox{\label{subfig:fanCrossingFree:addFifthNode:CaseABJ}}{
		\includegraphics[page=13, width=0.22\textwidth]{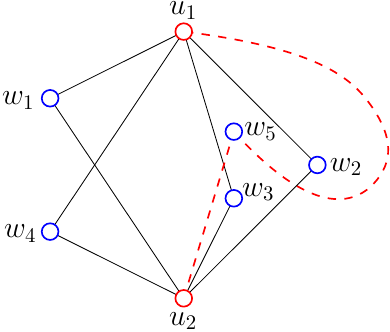}
	}
	\subcaptionbox{\label{subfig:fanCrossingFree:addFifthNode:CaseABK}}{
		\includegraphics[page=14, width=0.22\textwidth]{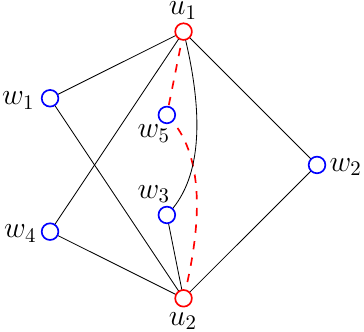}
	}
	\subcaptionbox{\label{subfig:fanCrossingFree:addFifthNode:CaseABL}}{
		\includegraphics[page=15, width=0.22\textwidth]{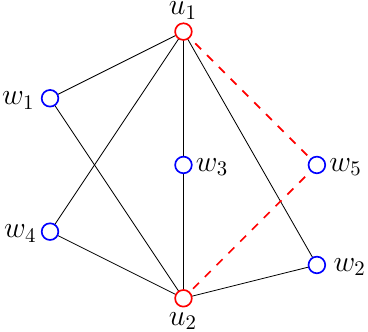}
	}
	\caption{All drawings that preserve the fan-crossing free property,
		when adding a fifth node $w_5$ to the drawing from
		Fig.~\ref{subfig:fanCrossingFree:addFourthNode:CaseAB}.}
	\label{fig:fanCrossingFree:addFifthNodeAB}
\end{figure}

Now we consider the drawing from Fig.~\ref{subfig:fanCrossingFree:addFourthNode:CaseAC}.
In this drawing, the edge $(u_1,w_5)$ is not allowed to cross 
edges $(u_1,w_1)$, $(u_1,w_2)$, $(u_1,w_3)$, $(u_1,w_4)$, $(u_2,w_5)$, $(u_2,w_1)$ and $(u_2,w_4)$, 
but it may cross at most one of the edges $(u_2,w_2)$ or $(u_2,w_3)$.
Further, the edge $(u_2,w_5)$ cannot cross $(u_2,w_1)$, $(u_2,w_2)$, $(u_2,w_3)$, $(u_2,w_4)$,
$(u_1,w_5)$, $(u_1,w_4)$ and $(u_1,w_2)$, but
it may cross at most one of the edges $(u_1,w_1)$ or $(u_1,w_3)$.
From this, we obtain the drawings in Fig.~\ref{fig:fanCrossingFree:addFifthNodeAC}.

\begin{figure}[htpb]
	\centering
	\subcaptionbox{\label{subfig:fanCrossingFree:addFifthNode:CaseACA}}{
		\includegraphics[page=16, width=0.22\textwidth]{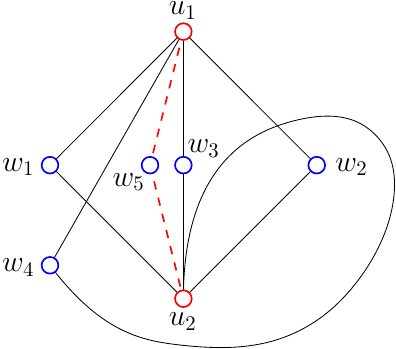}
	}
	\subcaptionbox{\label{subfig:fanCrossingFree:addFifthNode:CaseACB}}{
		\includegraphics[page=17, width=0.22\textwidth]{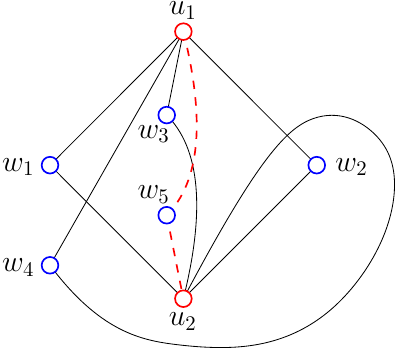}
	}
	\subcaptionbox{\label{subfig:fanCrossingFree:addFifthNode:CaseACC}}{
		\includegraphics[page=18, width=0.22\textwidth]{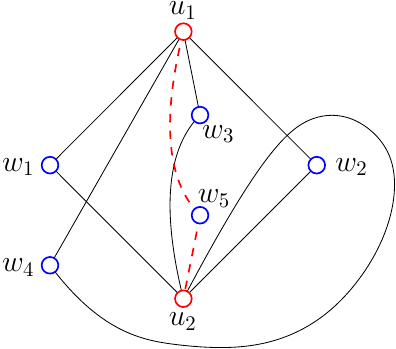}
	}\\
	\subcaptionbox{\label{subfig:fanCrossingFree:addFifthNode:CaseACD}}{
		\includegraphics[page=19, width=0.22\textwidth]{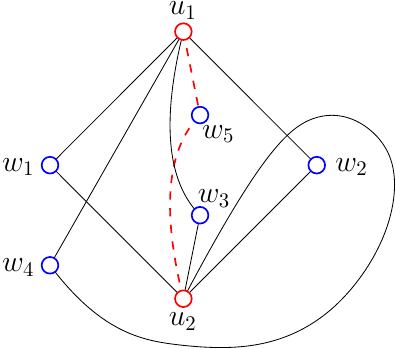}
	}
	\subcaptionbox{\label{subfig:fanCrossingFree:addFifthNode:CaseACE}}{
		\includegraphics[page=20, width=0.22\textwidth]{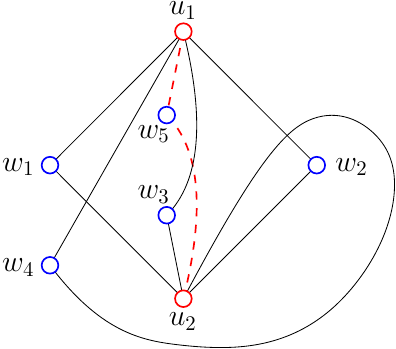}
	}
	\caption{All drawings that preserve the fan-crossing free property,
		when adding a fifth node $w_5$ to the drawing from
		Fig.~\ref{subfig:fanCrossingFree:addFourthNode:CaseAC}.}
	\label{fig:fanCrossingFree:addFifthNodeAC}
\end{figure}

Next, we consider the drawing from Fig.~\ref{subfig:fanCrossingFree:addFourthNode:CaseBC}.
In this drawing, the edge $(u_1,w_5)$ is not allowed to cross 
edges $(u_1,w_1)$, $(u_1,w_2)$, $(u_1,w_3)$, $(u_1,w_4)$, $(u_2,w_5)$, $(u_2,w_1)$ and $(u_2,w_2)$,
but it may cross at most one of the edges $(u_2,w_3)$ or $(u_2,w_4)$.
Further, the edge $(u_2,w_5)$ cannot cross $(u_2,w_1)$, $(u_2,w_2)$, $(u_2,w_3)$, $(u_2,w_4)$,
$(u_1,w_5)$, $(u_1,w_3)$ and $(u_1,w_4)$, but
it may cross at most one of the edges $(u_1,w_1)$ or $(u_1,w_2)$.
From this, we obtain the drawings in Fig.~\ref{fig:fanCrossingFree:addFifthNodeBC}.

\begin{figure}[htpb]
	\centering
	\subcaptionbox{\label{subfig:fanCrossingFree:addFifthNode:CaseBCA}}{
		\includegraphics[page=21, width=0.22\textwidth]{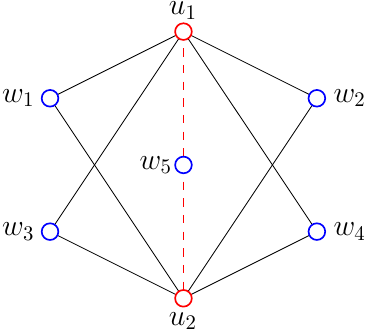}
	}
	\subcaptionbox{\label{subfig:fanCrossingFree:addFifthNode:CaseBCB}}{
		\includegraphics[page=22, width=0.22\textwidth]{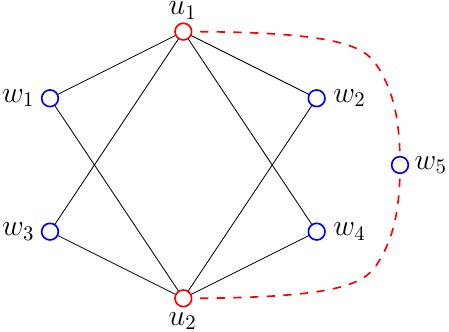}
	}
	\subcaptionbox{\label{subfig:fanCrossingFree:addFifthNode:CaseBCC}}{
		\includegraphics[page=23, width=0.22\textwidth]{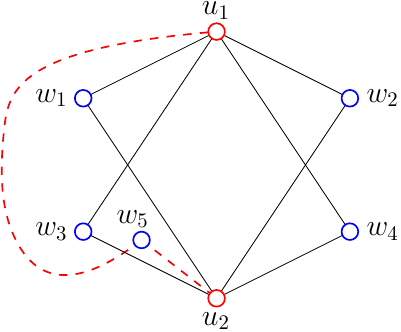}
	}\\
	\subcaptionbox{\label{subfig:fanCrossingFree:addFifthNode:CaseBCD}}{
		\includegraphics[page=24, width=0.22\textwidth]{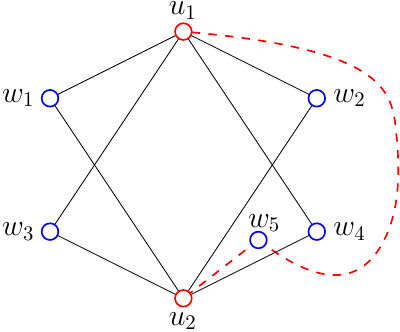}
	}
	\subcaptionbox{\label{subfig:fanCrossingFree:addFifthNode:CaseBCE}}{
		\includegraphics[page=25, width=0.22\textwidth]{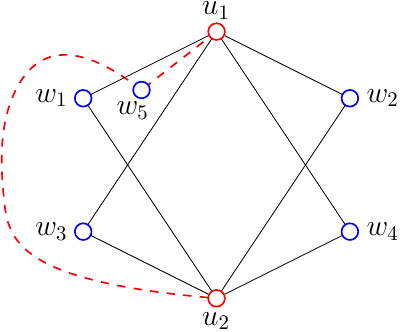}
	}
	\subcaptionbox{\label{subfig:fanCrossingFree:addFifthNode:CaseBCF}}{
		\includegraphics[page=26, width=0.22\textwidth]{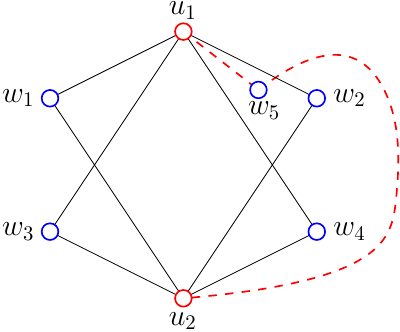}
	}
	\caption{All drawings that preserve the fan-crossing free property,
		when adding a fifth vertex $w_5$ to the drawing from
		Fig.~\ref{subfig:fanCrossingFree:addFourthNode:CaseBC}.}
	\label{fig:fanCrossingFree:addFifthNodeBC}
\end{figure}

Finally, we consider the drawing from Fig.~\ref{subfig:fanCrossingFree:addFourthNode:CaseBD}.
In this drawing, the edge $(u_1,w_5)$ is not allowed to cross 
edges $(u_1,w_1)$, $(u_1,w_2)$, $(u_1,w_3)$, $(u_1,w_4)$, $(u_2,w_5)$, $(u_2,w_1)$ and $(u_2,w_2)$,
but it may cross at most one of the edges $(u_2,w_3)$ or $(u_2,w_4)$.
Further, the edge $(u_2,w_5)$ cannot cross $(u_2,w_1)$, $(u_2,w_2)$, $(u_2,w_3)$, $(u_2,w_4)$,
$(u_1,w_5)$, $(u_1,w_3)$ and $(u_1,w_4)$, but
it may cross at most one of the edges $(u_1,w_1)$ or $(u_1,w_2)$.
From this, we obtain the drawings in Fig.~\ref{fig:fanCrossingFree:addFifthNodeBD}.

\begin{figure}[htpb]
	\centering
	\subcaptionbox{\label{subfig:fanCrossingFree:addFifthNode:CaseBDA}}{
		\includegraphics[page=27, width=0.22\textwidth]{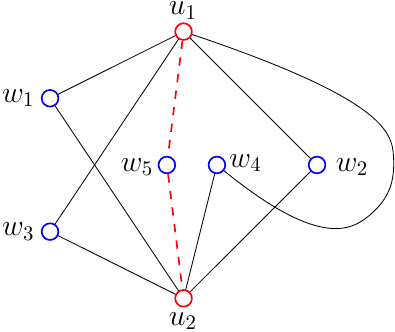}
	}
	\subcaptionbox{\label{subfig:fanCrossingFree:addFifthNode:CaseBDB}}{
		\includegraphics[page=28, width=0.22\textwidth]{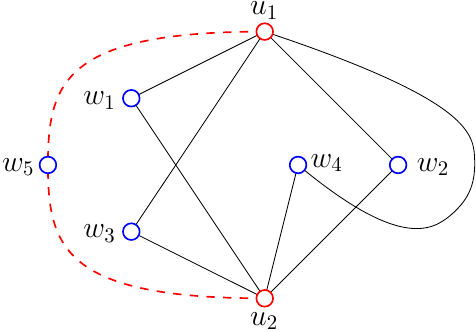}
	}
	\subcaptionbox{\label{subfig:fanCrossingFree:addFifthNode:CaseBDC}}{
		\includegraphics[page=29, width=0.22\textwidth]{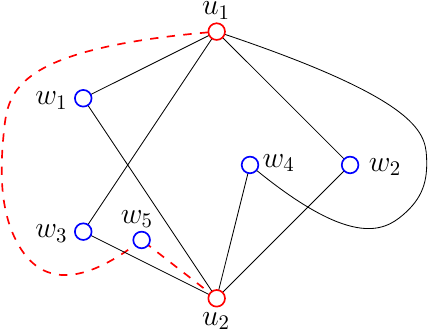}
	}\\
	\subcaptionbox{\label{subfig:fanCrossingFree:addFifthNode:CaseBDD}}{
		\includegraphics[page=30, width=0.22\textwidth]{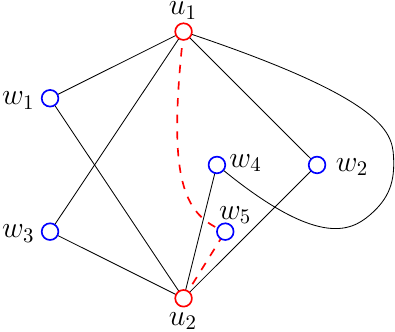}
	}
	\subcaptionbox{\label{subfig:fanCrossingFree:addFifthNode:CaseBDE}}{
		\includegraphics[page=31, width=0.22\textwidth]{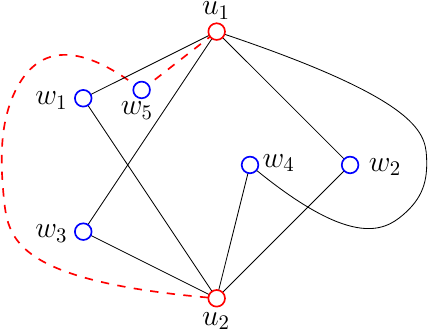}
	}
	\subcaptionbox{\label{subfig:fanCrossingFree:addFifthNode:CaseBDF}}{
		\includegraphics[page=32, width=0.22\textwidth]{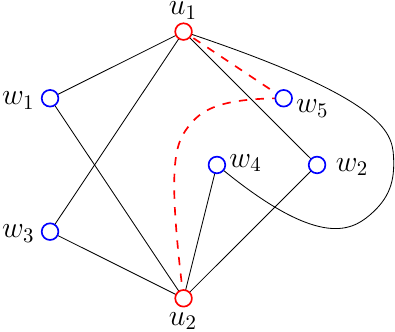}
	}
	\caption{All drawings that preserve the fan-crossing free property,
		when adding a fifth vertex $w_5$ to the drawing from
		Fig.~\ref{subfig:fanCrossingFree:addFourthNode:CaseBD}.}
	\label{fig:fanCrossingFree:addFifthNodeBD}
\end{figure}

We conclude again by observing that the configurations from the following
figures are topologically equivalent:
\begin{itemize}
	
	\item Fig.~\ref{subfig:fanCrossingFree:addFifthNode:CaseAAB} and
	Fig.~\ref{subfig:fanCrossingFree:addFifthNode:CaseABA},
	
	\item Fig.~\ref{subfig:fanCrossingFree:addFifthNode:CaseAAC},
	Fig.~\ref{subfig:fanCrossingFree:addFifthNode:CaseABG},
	Fig.~\ref{subfig:fanCrossingFree:addFifthNode:CaseABH} and
	Fig.~\ref{subfig:fanCrossingFree:addFifthNode:CaseACA},
	
	\item Fig.~\ref{subfig:fanCrossingFree:addFifthNode:CaseABB},
	Fig.~\ref{subfig:fanCrossingFree:addFifthNode:CaseABI} and
	Fig.~\ref{subfig:fanCrossingFree:addFifthNode:CaseBCB}
	
	\item Fig.~\ref{subfig:fanCrossingFree:addFifthNode:CaseABC},
	Fig.~\ref{subfig:fanCrossingFree:addFifthNode:CaseABJ},
	Fig.~\ref{subfig:fanCrossingFree:addFifthNode:CaseACC},
	Fig.~\ref{subfig:fanCrossingFree:addFifthNode:CaseACD},
	Fig.~\ref{subfig:fanCrossingFree:addFifthNode:CaseBCC},
	Fig.~\ref{subfig:fanCrossingFree:addFifthNode:CaseBCD},
	Fig.~\ref{subfig:fanCrossingFree:addFifthNode:CaseBCE},
	Fig.~\ref{subfig:fanCrossingFree:addFifthNode:CaseBCF},
	Fig.~\ref{subfig:fanCrossingFree:addFifthNode:CaseBDD} and
	Fig.~\ref{subfig:fanCrossingFree:addFifthNode:CaseBDF}.
	
	\item Fig.~\ref{subfig:fanCrossingFree:addFifthNode:CaseABC},
	Fig.~\ref{subfig:fanCrossingFree:addFifthNode:CaseABF},
	Fig.~\ref{subfig:fanCrossingFree:addFifthNode:CaseABK},
	Fig.~\ref{subfig:fanCrossingFree:addFifthNode:CaseBDA} and
	Fig.~\ref{subfig:fanCrossingFree:addFifthNode:CaseBDB},
	
	\item Fig.~\ref{subfig:fanCrossingFree:addFifthNode:CaseABE},
	Fig.~\ref{subfig:fanCrossingFree:addFifthNode:CaseABL} and
	Fig.~\ref{subfig:fanCrossingFree:addFifthNode:CaseBCA}
	
	\item Fig.~\ref{subfig:fanCrossingFree:addFifthNode:CaseACE},
	Fig.~\ref{subfig:fanCrossingFree:addFifthNode:CaseBDC},
	Fig.~\ref{subfig:fanCrossingFree:addFifthNode:CaseBDE} and
	Fig.~\ref{subfig:fanCrossingFree:addFifthNode:CaseBDF}
	
\end{itemize}
Note that Fig.~\ref{subfig:fanCrossingFree:addFifthNode:CaseBDF}
appears twice in this list. The reason for this is to make it easy
to recognize the similarity of these two sets of drawings. In fact, we observe
that there are seven different configurations that we
have to consider in the future, namely the configurations from
Figs.~\ref{subfig:fanCrossingFree:addFifthNode:CaseAAA},
\ref{subfig:fanCrossingFree:addFifthNode:CaseAAB},
\ref{subfig:fanCrossingFree:addFifthNode:CaseAAC},
\ref{subfig:fanCrossingFree:addFifthNode:CaseABB},
\ref{subfig:fanCrossingFree:addFifthNode:CaseABC},
\ref{subfig:fanCrossingFree:addFifthNode:CaseABD}, and
\ref{subfig:fanCrossingFree:addFifthNode:CaseABE}.

\begin{figure}[htpb]
	\centering
	\subcaptionbox{\label{subfig:fanCrossingFree:addFifthNode:allCasesA}}{
		\includegraphics[page=33, width=0.22\textwidth]{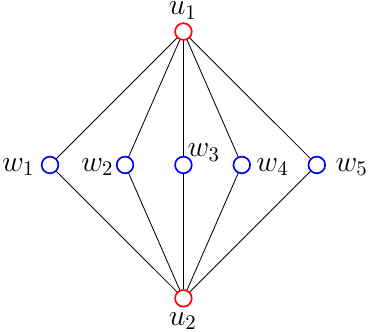}
	}
	\subcaptionbox{\label{subfig:fanCrossingFree:addFifthNode:allCasesB}}{
		\includegraphics[page=34, width=0.22\textwidth]{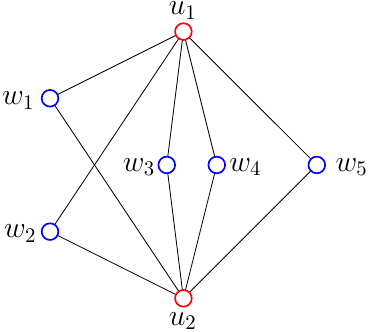}
	}
	\subcaptionbox{\label{subfig:fanCrossingFree:addFifthNode:allCasesC}}{
		\includegraphics[page=35, width=0.22\textwidth]{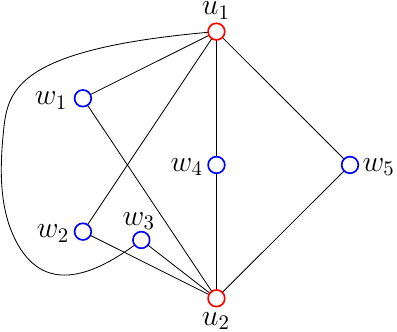}
	}
	\subcaptionbox{\label{subfig:fanCrossingFree:addFifthNode:allCasesD}}{
		\includegraphics[page=36, width=0.22\textwidth]{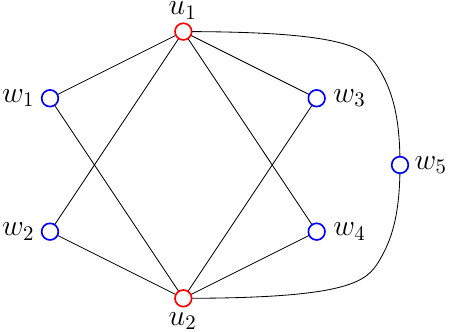}
	}\\
	\subcaptionbox{\label{subfig:fanCrossingFree:addFifthNode:allCasesE}}{
		\includegraphics[page=37, width=0.22\textwidth]{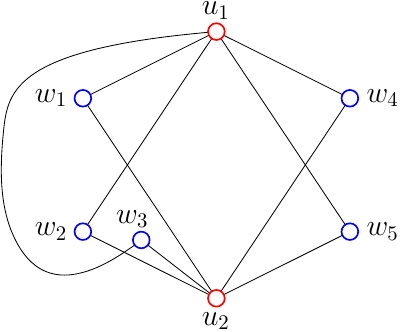}
	}
	\subcaptionbox{\label{subfig:fanCrossingFree:addFifthNode:allCasesF}}{
		\includegraphics[page=38, width=0.22\textwidth]{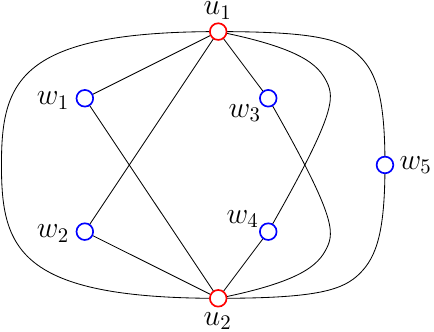}
	}
	\subcaptionbox{\label{subfig:fanCrossingFree:addFifthNode:allCasesG}}{
		\includegraphics[page=39, width=0.22\textwidth]{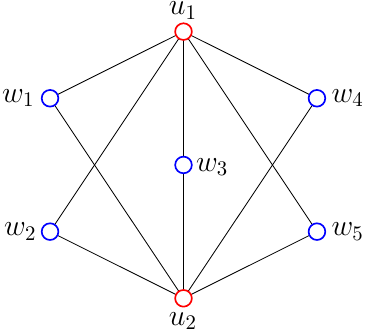}
	}
	\caption{All topologically different drawings of the subgraph $K_{2,5}$ that
		are fan-crossing free. Note that the nodes are relabelled in comparison to the
		Figures above.}
	\label{fig:fanCrossingFree:addFifthNode:AllCases}
\end{figure}

From the above discussion, it follows that any fan-crossing free drawing of $K_{2,5}$ is
topologically equivalent to one of the seven drawings described above, which are reported
again in Fig.~\ref{fig:fanCrossingFree:addFifthNode:AllCases} for the reader's convenience. 
We denote these drawings by $\Gamma_1,\ldots,\Gamma_7$.

Thus, in order to prove that neither $K_{3,7}$ nor $K_{5,5}$ is fan-crossing free,
we have to show that it is not possible to add the remaining vertices (one to 
$V_1$ and two to $V_2$ for $K_{3,7}$, and three to $V_2$ for $K_{5,5}$), without
violating the fan-crossing free property, to any of the seven drawings of $K_{2,5}$.
To make a systematic analysis of these cases, we will make use of the 
following lemmas.

\begin{lemma}\label{obs:1} 
	Consider the region $R$ bounded by a crossing-free edge $(u_1,w_i)$ and by 
	two crossing edges $(u_1,w_j)$, $(w_i,u_2)$, with $1 \leq i,j \leq 5$, 
	in a fan-crossing free drawing of $K_{a,b}$, with $a \geq 3$ and $b\geq 5$ (see Fig.~\ref{subfig:fanCrossingFree:observationA}). If at least four vertices of $V_2$ lie outside $R$, there is no vertex $u_h$, with $ 3 \leq h \leq a$, inside~$R$.
\end{lemma}
\begin{proof}
	If $R$ contained a vertex $u_h$, with $ 3 \leq h \leq a$, then at least one of the three 
	edges bounding $R$ would be crossed by at least two of the edges connecting $u_h$ to the 
	at least four vertices of $V_2$ outside $R$, hence creating a fan-crossing.
\end{proof}

\begin{figure}[htpb]
	\centering
	\subcaptionbox{\label{subfig:fanCrossingFree:observationA}}{
		\includegraphics[page=1, scale=0.6]{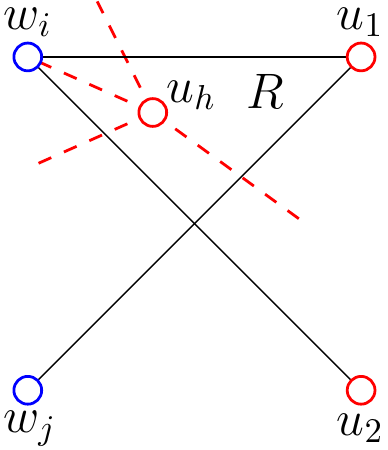}
	}\hfill
	\subcaptionbox{\label{subfig:fanCrossingFree:observationC}}{
		\includegraphics[page=3, scale=0.6]{figures/FanCrFree-observations}
	}\hfill
	\subcaptionbox{\label{subfig:fanCrossingFree:observationB}}{
		\includegraphics[page=2, scale=0.6]{figures/FanCrFree-observations}
	}
	\caption{Three observations for placing vertices from the set $\{u_3,u_4,u_5\}$.}
	\label{fig:fanCrossingFree:observation}
\end{figure}

\begin{lemma}\label{obs:3}  
	Consider the region $R_1$ bounded by two crossing-free edges $(u_1,w_k)$ and $(u_2,w_k)$, 
	and by two crossing edges $(u_1,w_j)$ and $(w_i,u_2)$, with $1 \leq i,j,k \leq 5$, 
	in a fan-crossing free drawing of $K_{a,b}$, with $a \geq 3$ and $b\geq 5$ (see 
	Fig.~\ref{subfig:fanCrossingFree:observationC}--\ref{subfig:fanCrossingFree:observationB}).
	Also, consider the two regions $R_2$ and $R_3$ bounded by $(u_1,w_i)$, $(u_1,w_j)$, 
	and $(u_2,w_i)$, and bounded by $(u_2,w_i)$, $(u_2,w_j)$, and $(u_1,w_j)$, respectively.
	
	Then, there is at most one vertex $u_h$, with $3 \leq h \leq a$, in $R_1$. 
	Also, if $u_h$ lies in $R$, then there exist at most two vertices of $V_2 \setminus \{w_i,w_j,w_k\}$
	outside $R_1$; one of these vertices lies in $R_2$ and the other in $R_3$.
\end{lemma}
\begin{proof}
	Suppose that a vertex $u_h$, with $3 \leq h \leq a$, lies in $R$. 
	We first claim that edges $(u_h,w_i)$ and $(u_h,w_j)$ cross edges $(u_1,w_k)$ 
	and $(u_2,w_k)$, respectively; see Fig.~\ref{subfig:fanCrossingFree:observationC}. 
	Namely, edge $(u_h,w_i)$ cannot cross $(u_1,w_j)$, 
	as this edge already crosses $(u_2,w_i)$. 
	Also, edge $(u_h,w_j)$ cannot cross $(u_2,w_i)$, since this edge already 
	crosses $(u_1,w_j)$. Finally, if $(u_h,w_i)$ crosses $(u_2,w_k)$, and $(u_h,w_j)$ 
	crosses $(u_1,w_k)$, then $(u_h,w_i)$ and $(u_h,w_j)$ cross each other, which is not allowed. 
	The claim follows. 
	
	We now argue that there exists no vertex $u_z$, with $3 \leq z \neq h \leq a$, in $R_1$.
	In fact, in this case, $(u_z,w_i)$ and $(u_z,w_j)$ would have to cross edges $(u_1,w_k)$
	and $(u_2,w_k)$, respectively, for the same reasons as above. Hence, both $(u_1,w_k)$ 
	and $(u_2,w_k)$ would cross fans incident to $w_i$ and $w_j$, respectively.
	
	We conclude the proof by considering the possible placement of a vertex 
	$w_x \in V_2 \setminus \{w_i,w_j,w_k\}$ outside $R_1$. Refer to Fig.~\ref{subfig:fanCrossingFree:observationB}.
	If $w_x$ lies neither in  $R_2$ nor in $R_3$, then the only possibility to connect it 
	to $u_h$ is to cross 
	both edges $(u_1,w_j)$ and $(u_2,w_i)$; in this case, however, it is not possible to place
	any other vertex of $V_2 \setminus \{w_i,w_j,w_k\}$ outside $R_1$, as this would require 
	an additional crossing between one of the edges bounding $R_1$ and an edge incident to $u_h$,
	which would create a fan-crossing. On the other hand, if $w_x$ lies in $R_2$, we can draw 
	$(u_h,w_x)$ by crossing only $(u_1,w_j)$, which still leaves the option to place an additional 
	vertex $w_y \in V_2 \setminus \{w_i,w_j,w_k,w_x\}$ in $R_3$, and draw $(u_h,w_y)$ by crossing only
	$(u_2,w_i)$. However, once this edge has been drawn, we cannot add any additional vertex of $V_2$
	outside $R_1$, and the statement follows.
\end{proof}



We show how to exploit these lemmas to complete the proof of Characterization~\ref{th:bipartite:fcf}.

\section{Graph $K_{5,5}$ is not fan-crossing free}\label{se:k55}

We consider the seven drawings $\Gamma_1,\dots,\Gamma_7$ of $K_{2,5}$ and show that none
of them can be extended to a fan-crossing free drawing of $K_{5,5}$.

\medskip

\noindent\textbf{The drawing $\Gamma_1$.} 
In drawing $\Gamma_1$ we have five topologically similar regions, denoted by 
$R_1,\ldots,R_5$ as in Fig.~\ref{fig:fanCrossingFree:drawingA}.
\begin{figure}[tb]
	\centering
	\includegraphics[page=1, width=0.5\textwidth]{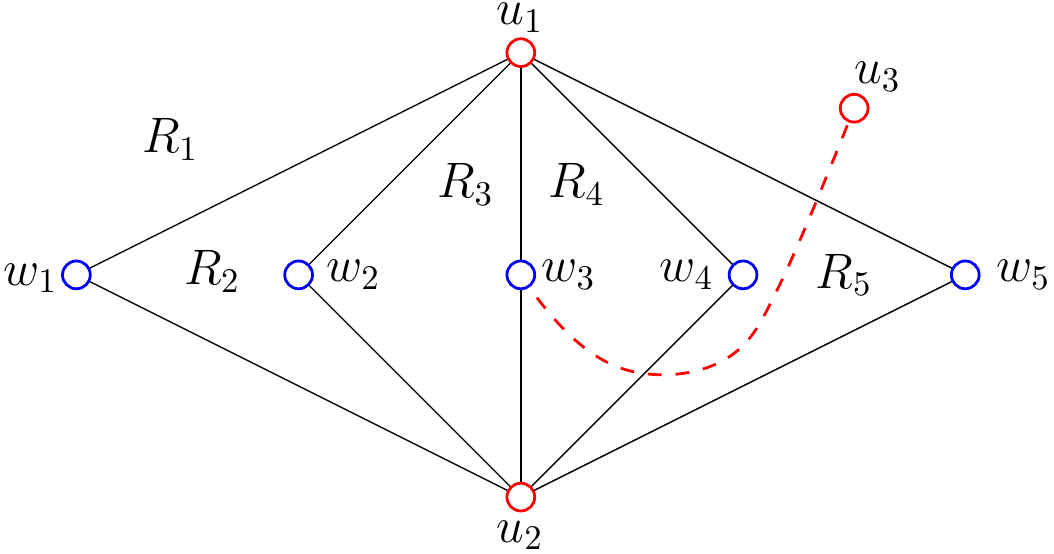}
	\caption{The regions in $\Gamma_1$. The dashed red line represents the edge $(u_3,w_3)$.}
	\label{fig:fanCrossingFree:drawingA}
\end{figure}

\begin{figure}[b!]
	\centering
	\hfill~
	\subcaptionbox{\label{fig:fanCrossingFree:drawingB}}{
		\includegraphics[page=2, width=0.46\textwidth]{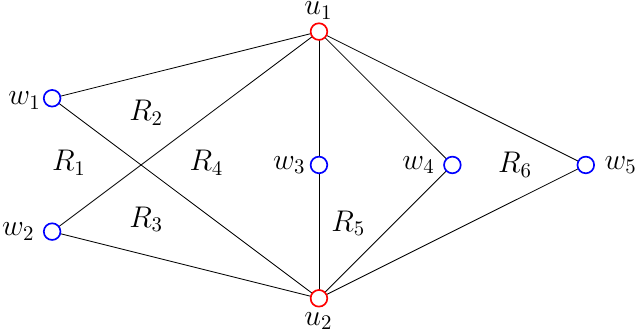}
	}
	\hfill
	\subcaptionbox{\label{fig:fanCrossingFree:drawingBB}}{
		\includegraphics[page=3, width=0.46\textwidth]{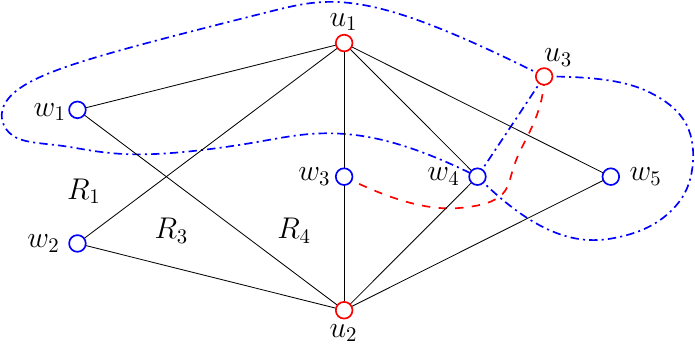}
	}
	\hfill~
	\caption{
		(a)~The regions in $\Gamma_2$.
		(b)~Node $u_3$ is placed in region $R_1$ and the edge $(u_3,w_3)$ (dashed red)
		crosses $(u_1,w_5)$	and $(u_2,w_4)$. The dashed dotted blue lines show the edge $(u_3,w_4)$
		that cannot be drawn fan-crossing free in this setting.}
	\label{fig:drawings-2}
\end{figure}
We assume w.l.o.g.~that $u_3$ is in the region $R_1$. The edge $(u_3,w_3)$ can only be realized
by crossing exactly one edge incident to $u_1$, say $(u_1,w_5)$, and exactly one edge incident 
to $u_2$, say $(u_2,w_4)$ (see dashed red edge in Fig.~\ref{fig:fanCrossingFree:drawingA}). 
Now the edge $(u_3,w_4)$ is not realizable. Namely, we are not allowed to draw this edge 
through the regions $R_2$, $R_3$ and $R_4$, since in this case it would cross at least two edges incident to $u_1$ or at least two edges incident to $u_2$ -- a fan-crossing. Also, it cannot cross 
the edge $(u_1,w_5$), as this would also create a fan-crossing. The only option left is to cross 
the edge $(u_2,w_5)$; however, this implies that $(u_3,w_4)$ crosses $(u_3,w_3)$, which is not allowed.
Hence, $\Gamma_1$ cannot be a subdrawing of a fan-crossing free drawing of $K_{5,5}$.

\medskip

\noindent\textbf{The drawing $\Gamma_2$.} 
In drawing $\Gamma_2$ we have six regions, denoted by $R_1,\ldots,R_6$ as in Fig.~\ref{fig:fanCrossingFree:drawingB}.
We observe some consequences of the previous lemmas:
\begin{inparaenum}[(i)]
	\item vertices $u_3,u_4,u_5$ can lie neither in $R_2,R_3$ (Lemma~\ref{obs:1}) nor in $R_4$ (Lemma~\ref{obs:3});
	\item vertices $u_3,u_4,u_5$ cannot lie in $R_5$, since the edge connecting one of them 
	to $w_1$ would cross both edges $(u_1,w_4)$ and $(u_1,w_5)$ (Lemma~\ref{obs:3});
	\item at most one of vertices $u_3,u_4,u_5$ can lie in $R_6$ (Lemma~\ref{obs:3}).
\end{inparaenum}

By the previous analysis, at least two vertices, say $u_3$ and $u_4$, lie in region $R_1$. 
At most one of the edges $(u_3,w_3)$ and $(u_4,w_3)$ can cross $(u_1,w_2)$ and $(u_2,w_1)$.
Thus, at least one of these edges has to cross either $(u_1,w_4)$ and $(u_2,w_5)$,
or $(u_1,w_5)$ and $(u_2,w_4)$, assume w.l.o.g.~the latter pair; see Fig.~\ref{fig:fanCrossingFree:drawingBB}. However, this implies that it is no longer possible
to draw the edge $(u_3,w_4)$ fan-crossing free.
Hence, $\Gamma_2$ cannot be a subdrawing of a fan-crossing free drawing of $K_{5,5}$.

\medskip

\noindent\textbf{The drawing $\Gamma_3$.} 
In drawing $\Gamma_3$ we have seven regions, denoted by $R_1,\ldots,R_7$ as in Fig.~\ref{fig:fanCrossingFree:drawingC}.
\begin{figure}[tb]
	\centering
	\includegraphics[page=5, width=0.5\textwidth]{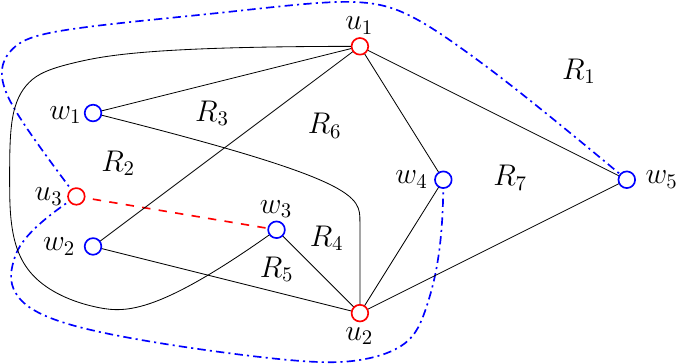}
	\caption{The regions in $\Gamma_3$. The dashed dotted blue lines indicate the two 
		edges $(u_3,w_4)$ and $(u_3,w_5)$ that are not drawable without fan-crossings
		when adding vertex $u_3$ in $R_2$.}
	\label{fig:fanCrossingFree:drawingC}
\end{figure}
We observe some consequences of the previous lemmas:
\begin{inparaenum}[(i)]
	\item vertices $u_3,u_4,u_5$ can lie neither in $R_3$ nor in $R_5$ (Lemma~\ref{obs:1});
	\item vertices $u_3,u_4,u_5$ cannot lie in $R_6$, since the edge connecting one of them 
	to $w_1$ would cross both edges $(u_1,w_4)$ and $(u_1,w_5)$ (Lemma~\ref{obs:3});
	\item vertices $u_3,u_4,u_5$ cannot lie in $R_7$, since the edges connecting one of them 
	to $w_1$ and $w_2$ would both cross $(u_1,w_3)$ (Lemma~\ref{obs:3}).
\end{inparaenum}
Since $u_3,u_4,u_5$ are all connected to $w_1$ and $w_2$, at least one of them, say $u_3$, must 
lie in region $R_2$. First we consider the edge $(u_3,w_3)$. This edge can neither
cross $(u_1,w_3)$, nor $(u_2,w_2)$, nor both $(u_1,w_1)$ and $(u_1,w_2)$, nor both
$(w_1,u_1)$ and $(w_1,u_2)$. So, the only option for edge $(u_3,w_3)$ is to cross $(u_1,w_2)$
(see Fig.~\ref{fig:fanCrossingFree:drawingC}).
However, since the edges $(u_3,w_4)$ and $(u_3,w_5)$ can cross neither $(u_1,w_2)$ nor $(u_3,w_3)$,
they both cross $(u_1,w_3)$, making a fan-crossing. 
Hence, $\Gamma_3$ cannot be a subdrawing of a fan-crossing free drawing of $K_{5,5}$.

\medskip

\noindent\textbf{The drawing $\Gamma_4$.} 
In drawing $\Gamma_4$ we have seven regions, denoted by $R_1,\ldots,R_7$, as in Fig.~\ref{fig:fanCrossingFree:drawingD}.
\begin{figure}[tb]
	\centering
	\includegraphics[page=6, width=0.5\textwidth]{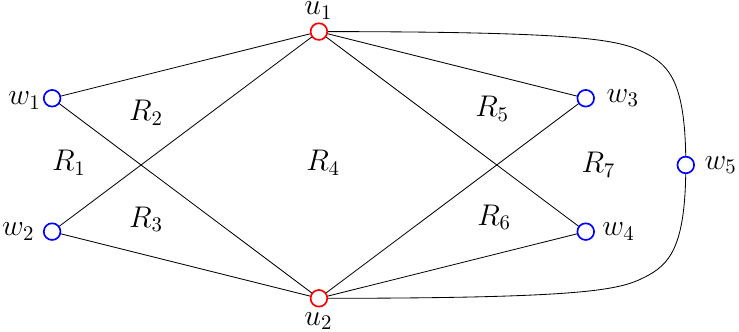}
	\caption{The regions in $\Gamma_4$.}
	\label{fig:fanCrossingFree:drawingD}
\end{figure}
By the lemmas above we conclude that vertices $u_3,u_4,u_5$ can lie neither in
$R_2,R_3,R_5,R_6$ (Lemma~\ref{obs:1}) nor in $R_4$ (Lemma~\ref{obs:3}).
Also, at most one of them can lie in $R_1$ and $R_7$ (Lemma~\ref{obs:3}).
So we cannot place all of the three nodes $u_3,u_4,u_5$ together with their edges
and obtain a fan-crossing free drawing. Hence, $\Gamma_4$ cannot be a subdrawing 
of a fan-crossing free drawing of $K_{5,5}$.

\medskip

\noindent\textbf{The drawing $\Gamma_5$.} 
In drawing $\Gamma_5$ we have eight regions,
denoted by $R_1,\ldots,R_8$ as in Fig.~\ref{fig:fanCrossingFree:drawingE}.
\begin{figure}[tb]
	\centering
	\includegraphics[page=7, width=0.5\textwidth]{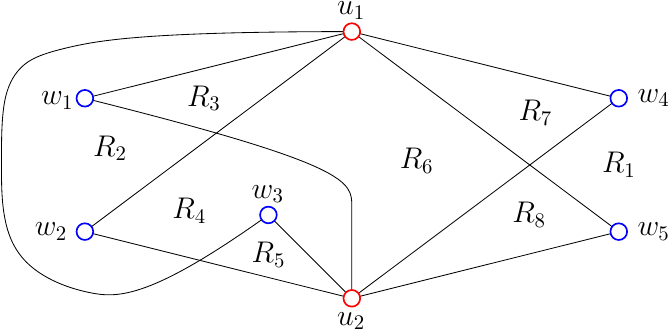}
	\caption{The regions in $\Gamma_5$.}
	\label{fig:fanCrossingFree:drawingE}
\end{figure}
The lemmas above imply that:
\begin{inparaenum}[(i)]
	\item vertices $u_3,u_4,u_5$ can lie neither in $R_3,R_5,R_7,R_8$ (Lemma~\ref{obs:1})
	nor in $R_6$ (Lemma~\ref{obs:1});
	\item vertices $u_3,u_4,u_5$ cannot lie in $R_4$, since the edges connecting one of them, 
	say $u_3$, to $w_4$ and $w_5$ would cross $(u_1,w_2)$ and $(u_2,w_2)$, respectively. 
	But then edge $(u_3,w_4)$ is not allowed to cross $(u_1,w_3)$ and $(u_2,w_2)$, and so this edge
	is not drawable at all;
	\item vertices $u_3,u_4,u_5$ cannot lie in $R_2$, since the edges connecting one of them, 
	say $u_3$, to $w_4$ and $w_5$ would cross $(u_1,w_3)$ and $(u_2,w_3)$, respectively. But then edge 
	$(u_3,w_5)$ is not allowed to cross $(u_1,w_3)$ and $(u_2,w_2)$, and so this edge
	is not drawable at all.
\end{inparaenum}
Thus, the only option left is that all the vertices $u_3,u_4$, and $u_5$ are in $R_1$.
However, edges $(u_3,w_3)$, $(u_4,w_3)$ and $(u_5,w_3)$ are neither allowed
to cross $(u_2,w_2)$, nor to cross both $(u_1,w_3)$ and $(u_1,w_2)$,
nor to cross both $(u_2,w_4)$ and $(u_2,w_1)$,
nor to cross both $(u_1,w_4)$ and $(u_1,w_5)$, and so they cannot be drawn.
Hence, $\Gamma_5$ cannot be a subdrawing of a fan-crossing free drawing of $K_{5,5}$.

\medskip

\noindent\textbf{The drawing $\Gamma_6$.}
In drawing $\Gamma_6$, we have seven regions,
denoted by $R_1,\ldots,R_7$ as in Fig.~\ref{fig:fanCrossingFree:drawingF}.
\begin{figure}[b]
	\centering
	\includegraphics[page=8, width=0.5\textwidth]{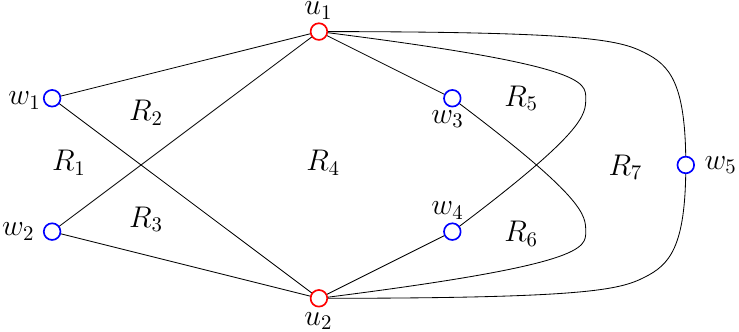}
	\caption{The regions in $\Gamma_6$.}
	\label{fig:fanCrossingFree:drawingF}
\end{figure}
The lemmas above imply that:
\begin{inparaenum}[(i)]
	\item vertices $u_3,u_4,u_5$ can lie neither in regions $R_2,R_3,R_5,R_6$ (Lemma~\ref{obs:1}) nor in region $R_7$ (Lemma~\ref{obs:3});
	\item vertices $u_3,u_4,u_5$ cannot lie in $R_4$, since the edge connecting one of them 
	to $w_1$ (to $w_2$) would cross both edges $(u_1,w_3)$ and $(u_1,w_4)$ (both edges 
	$(u_2,w_3)$ and $(u_2,w_4)$), hence creating a fan-crossing.
\end{inparaenum}
The only option left is that all vertices $u_3,u_4$ and $u_5$ are in $R_1$.
However, the edge connecting one of them to $w_3$ can neither cross $(u_1,w_4)$, 
nor $(u_2,w_3)$, nor the pair $(u_1,w_2)$ and $(u_1,w_1)$, nor the pair $(u_2,w_1)$ and $(u_2,w_2)$.
So the only option for this edge is to cross the edges $(u_1,w_2)$ and $(u_2,w_1)$. 
However, this cannot be done for all the three vertices.
Hence, $\Gamma_6$ cannot be a subdrawing of a fan-crossing free drawing of $K_{5,5}$.

\medskip

\noindent\textbf{The drawing $\Gamma_7$.}
In drawing $\Gamma_7$, we have seven regions,
denoted by $R_1,\ldots,R_7$ as in Fig.~\ref{fig:fanCrossingFree:drawingG}.
\begin{figure}[tb]
	\centering
	\includegraphics[page=9, width=0.5\textwidth]{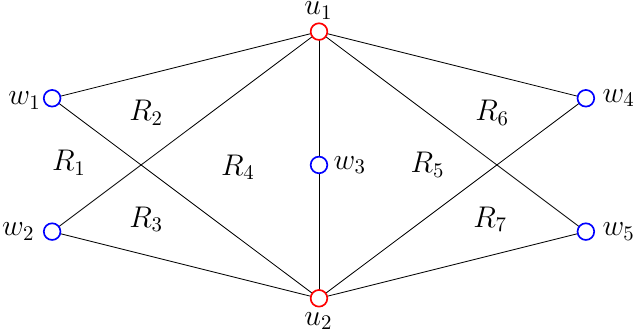}
	\caption{The regions in $\Gamma_7$.}
	\label{fig:fanCrossingFree:drawingG}
\end{figure}
By the lemmas above we get that vertices $u_3,u_4,u_5$ can lie neither in regions $R_2,R_3,R_6,R_7$ (Lemma~\ref{obs:1}) nor in regions $R_4,R_5$ (Lemma~\ref{obs:3}). 
Thus, the only option left is that all of them are in $R_1$.
However, the edge connecting one of them to $w_3$ can cross neither both edges 
$(u_1,w_2)$ and $(u_1,w_1)$, nor both edges $(u_2,w_1)$ and $(u_2,w_2)$, not both
edges $(u_1,w_5)$ and $(u_1,w_4)$, nor both edges $(u_2,w_4)$ and $(u_2,w_5)$.
Hence each of such edges must cross either both edges
$(u_1,w_2)$ and $(u_2,w_1)$, or both edges $(u_1,w_5)$ and $(u_2,w_4)$.
Since there are three such edges, there must be a fan-crossing.
Hence, $\Gamma_7$ cannot be a subdrawing of a fan-crossing free drawing of $K_{5,5}$.

So we can conclude that there is no fan-crossing free drawing of $K_{5,5}$.

\section{Graph $K_{3,7}$ is not fan-crossing free}

We now turn our attention to graph $K_{3,7}$. As in Section~\ref{se:k55}, we consider each of the seven drawings $\Gamma_1,\ldots,\Gamma_7$ of $K_{2,5}$ separately and try to add more vertices to them without violating the fan-crossing free property.

\medskip

\noindent\textbf{The drawing $\Gamma_1$.}
We already proved that in the presence of vertex $u_3$, drawing 
$\Gamma_1$ cannot be a subdrawing of a fan-crossing free drawing of $K_{3,5}$.

\medskip

\noindent\textbf{The drawing $\Gamma_2$.}
In drawing $\Gamma_2$, we have six regions $R_1,\ldots,R_6$ as in 
Fig.~\ref{fig:fanCrossingFree:drawingB}.
As in the corresponding case of Section~\ref{se:k55}, vertex $u_3$ can lie neither in $R_2,R_3$ (Lemma~\ref{obs:1}), nor in $R_4$ (Lemma~\ref{obs:3}), nor in $R_5$.

First, we assume that $u_3$ is in the region $R_1$. 
Then, similarly to the corresponding case of Section~\ref{se:k55} for $K_{5,5}$, 
it is not possible to draw $(u_3,w_3)$ through the regions $R_5$ and $R_6$. 
So, $(u_3,w_3)$ must cross $(u_1,w_2)$ and $(u_2,w_1)$; also, the edge $(u_3,w_4)$ must cross 
one of $(u_1,w_5)$ or $(u_2,w_5)$, say the former; see Fig.~\ref{fig:fanCrossingFree:general:drawingBB}, in which the regions $R'_1$ and $R'_6$ 
are the faces delimited by the black and red lines in the plane graph.

\begin{figure}[tb]
	\centering
	 \subcaptionbox{\label{fig:fanCrossingFree:general:drawingBB}}{
	 \includegraphics[page=1, width=0.49\textwidth]{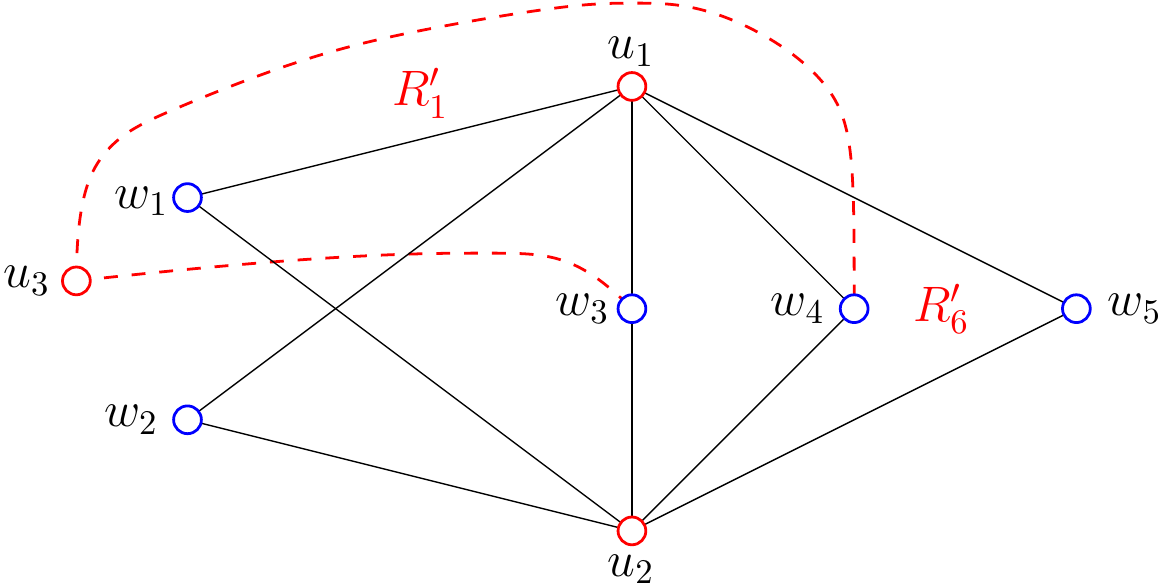}}
	 \hfil
	 \subcaptionbox{\label{fig:fanCrossingFree:general:drawingBC}}{
	 \includegraphics[page=2, width=0.49\textwidth]{figures/FanCrFree-contradictionK47}}
	\caption{
	(a)~$u_3$ is in $R_1$; 
	$(u_3,w_3)$ crosses $(u_1,w_2)$ and $(u_2,w_1)$; 
	$(u_3,w_4)$ crosses $(u_1,w_5)$, and 
	(b)~$w_6$ is in $R'_6$; 
	$(u_3,w_6)$ crosses $(u_2,w_5)$;
	$(u_1,w_6)$ crosses $(u_2,w_4)$.}
\end{figure}

Now, consider vertex $w_6\in V_2$. This vertex cannot be in region $R_2$. Indeed, if $w_6$ were in
$R_2$, then both edges $(u_2,w_6)$ and $(u_3,w_6)$ must cross $(u_1,w_1)$, since they are allowed 
to cross neither $(u_2,w_1)$ nor $(u_1,w_2)$; this yields a fan-crossing. Also, vertex $w_6$ cannot 
be in $R_3$, due to a similar observation. If $w_6$ were in $R_3$, then both edges $(u_1,w_6)$ and 
$(u_3,w_6)$ must cross $(u_2,w_2)$, since they are allowed to cross neither $(u_2,w_1)$ nor 
$(u_1,w_2)$, which again yields a fan-crossing.
Finally, $w_6$ cannot be in one of the regions $R_4$ or $R_5$, as otherwise the edge from $u_3$ to 
$w_6$ cannot be drawn without introducing a fan-crossing. Thus, 
$w_6$ can be only in either $R_6$ or $R_1$. We consider each of these two cases separately. 

First, consider the case in which $w_6$ is in $R_6$. Since $(u_3,w_6)$ is allowed to cross neither 
$(u_1,w_5)$ nor $(u_3,w_4)$, and since this edge cannot cross $(u_1,w_4)$ without introducing a 
fan-crossing, vertex $w_6$ must be in $R'_6$. To avoid introducing any fan-crossing, edge $(u_3,w_6)$ 
must cross $(u_2,w_3)$. Similarly, the edge $(u_1,w_6)$ must cross $(u_2,w_4)$; see dashed dotted blue 
edges in Fig.~\ref{fig:fanCrossingFree:general:drawingBC}.
	
Next, we argue for vertex $w_7$, which cannot lie in $R'_6$, since by the arguments above $(u_3,w_7)$ 
would have to cross $(u_2,w_5)$. So, $w_7$ must be in $R_1$. It is not difficult to see that 
$(u_1,w_7)$ cannot cross an edge incident to $u_1$, or one of the edges $(u_3,w_3)$, $(u_3,w_4)$ and 
$(u_2,w_4)$, or the pair $(u_2,w_3)$ and $(u_2,w_1)$. As a result, $w_7$ must be in $R'_1$. But then it 
is easy to see that $(u_2,w_7)$ yields inevitably a fan-crossing. In fact, this edge is not allowed to 
cross any edge incident to $u_2$, or one of the edges $(u_3,w_3)$, $(u_3,w_6)$ and $(u_1,w_6)$, or both 
edges $(u_1,w_3)$, $(u_1,w_2)$. 

It remains to consider the case in which $w_6$ in $R_1$. We claim that in this case $w_6$ is 
inevitably in $R'_1$. To see this, observe that $(u_1,w_6)$ can cross neither any edge incident 
to $u_1$, nor one of the edges $(u_3,w_3)$ and $(u_3,w_4)$, nor both edges $(u_2,w_4)$ and $(u_2,w_5)$, 
nor both edges $(u_2,w_3)$ and $(u_2,w_1)$. As in the previous case, we next argue for $w_7$, which 
cannot be in $R_6$ (by the arguments above). So, we can assume that $w_7$ is in $R_1$, too. By the same 
arguments as for $w_6$, we obtain that $w_7$ is in $R'_1$. For both edges $(u_2,w_6)$ and $(u_2,w_7)$ 
the following holds. They can cross neither $(u_3,w_3)$, nor the pair $(u_1,w_3)$ and $(u_1,w_2)$, 
nor the pair $(u_1,w_4)$ and $(u_1,w_5)$. Thus, both have to cross $(u_3,w_4)$, yielding a 
fan-crossing.

From the case analysis above, we conclude that in the case in which $u_3$ is in $R_1$, it is not 
possible to augment $\Gamma_2$ to a fan-crossing free drawing of $K_{3,7}$.

We now consider the second case of our analysis for $u_3$, in which $u_3$ is in $R_6$. By 
Lemma~\ref{obs:3}, the edges $(u_3,w_1)$ and $(u_3,w_2)$ cross the edges $(u_1,w_5)$ and $(u_2,w_5)$, 
respectively; see
Fig~\ref{fig:fanCrossingFree:general:drawingBD}. It follows that $(u_3,w_3)$ crosses either $(u_1,w_4)$ 
or $(u_2,w_4)$. W.l.o.g.~assume that $(u_3,w_3)$ crosses $(u_1,w_4)$. Then, $(u_3,w_4)$ must be planar, as 
otherwise regardless of its drawing a fan-crossing is introduced.

\begin{figure}[tb]
	\centering
	 \subcaptionbox{\label{fig:fanCrossingFree:general:drawingBD}}{
	 \includegraphics[page=3, width=0.49\textwidth]{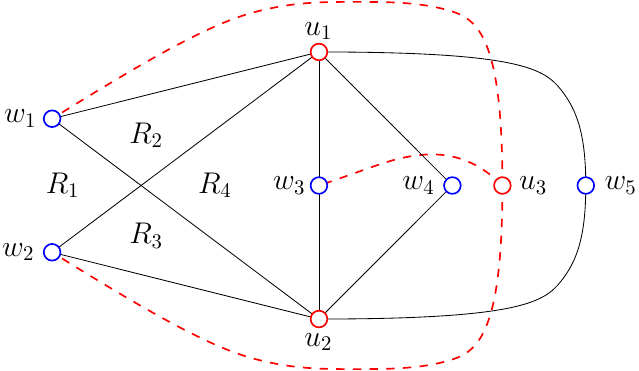}}
	 \hfil
	 \subcaptionbox{\label{fig:fanCrossingFree:general:drawingBF}}{
	 \includegraphics[page=4, width=0.49\textwidth]{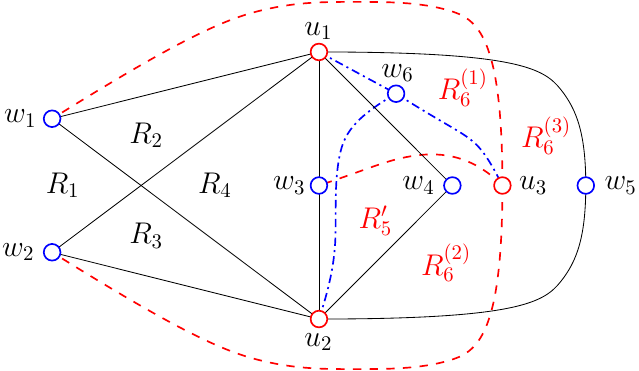}}
	\caption{
	(a)~$u_3$ is in $R_6$; 
	$(u_3,w_1)$ crosses $(u_1,w_5)$;
	$(u_3,w_2)$ crosses $(u_2,w_5)$; 
	$(u_3,w_3)$ crosses $(u_1,w_4)$, and 
	(b)~$w_6$ is in $R_6^{(1)}$; 
	$(u_2,w_6)$ crosses $(u_3,w_3)$.}
\end{figure}

We next argue for $w_6$. This vertex can be neither in $R_2$ nor in $R_3$, as otherwise either the edge $(u_2,w_6)$ or the $(u_1,w_6)$ would yield a fan-crossing, respectively. Vertex $w_6$ can be neither in $R_1$ nor in $R_4$ as well, because in both cases edge $(u_3,w_6)$ yields a fan-crossing. In follows that $w_6$ is either in $R_5$ or in $R_6$. We consider each of these two cases separately in the following. 
First suppose that $w_6$ is in $R_6$. The edges incident to $u_3$ partition $R_6$ into three subregions, which we denote by $R_6^{(1)}$, $R_6^{(2)}$ and $R_6^{(3)}$; refer to the regions annotated with red labels in Fig~\ref{fig:fanCrossingFree:general:drawingBF}.
If $w_6$ is in $R_6^{(2)}$, then both edges $(u_2,w_6)$ and $(u_3,w_6)$ can be even crossing-free, but 
nevertheless $(u_1,w_6)$ yields a fan-crossing. If $w_6$ is in $R_6^{(3)}$, then the edges $(u_1,w_6)$ 
and $(u_2,w_6)$ yield a fan-crossing. It follows that $w_6$ is in $R_6^{(1)}$.
	
We next observe that $(u_1,w_6)$ can cross none of the edges $(u_1,w_4)$, $(u_1,w_5)$, $(u_3,w_1)$ and $(u_3,w_3)$, which implies that $(u_1,w_6)$ is crossing-free; see dashed-dotted blue edged in Fig~\ref{fig:fanCrossingFree:general:drawingBF}). Analogously, $(u_3,w_6)$ must be also crossing-free.
Now consider the edge $(u_2,w_6)$. Observe that this edge cannot cross an edge incident to $u_2$, or two edges 	incident to $u_1$, or two edges incident to $u_3$. We conclude that the only way to draw $(u_2,w_6)$ is by crossing $(u_1,w_4)$ and $(u_3,w_3)$.
	
On the other hand, vertex $w_7$ cannot lie in $R_6$, since the same arguments as for $w_6$ above imply that $(u_2,w_7)$ must cross $(u_1,w_4)$ and $(u_3,w_3)$, thus forming a fan-crossing. So, vertex $w_7$ is in $R_5$. In particular, in the presence of the edge $(u_3,w_7)$, vertex $w_7$ must be in the region denoted by $R'_5$ (see Fig~\ref{fig:fanCrossingFree:general:drawingBF}). However, this implies that the edge $(u_1,w_7)$ yields a fan-crossing; indeed, this edge cannot cross $(u_3,w_3)$, or any edge incident to $u_1$, or the pair of edges $(u_2,w_6)$ and $(u_2,w_3)$, or the pair $(u_2,w_4)$ and $(u_2,w_5)$.
	
\begin{figure}[tb]
	\centering
	\includegraphics[page=5, width=0.49\textwidth]{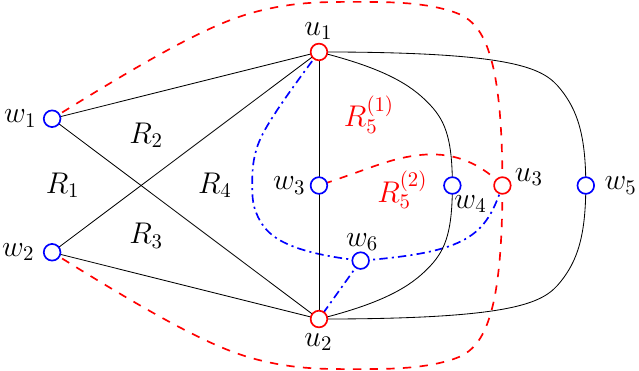}
	\caption{Vertex $w_6$ is in $R_5$; 
	$(u_3,w_6)$ crosses $(u_2,w_4)$;
	$(u_1,w_6)$ crosses $(u_2,w_3)$.}
	\label{fig:fanCrossingFree:general:drawingBE}
\end{figure}

It remains to consider the case in which $w_6$ is in $R_5$. In the presence of the edge $(u_3,w_3)$, this region is partitioned into two subregions $R_5^{(1)}$ and $R_5^{(2)}$; refer to the regions annotated with red labels in Fig~\ref{fig:fanCrossingFree:general:drawingBE}.
Suppose first that $w_6$ is in $R_5^{(1)}$. Then, it is not difficult to see that $(u_3,w_6)$ inevitably yields a fan-crossing, as this edge cannot cross $(u_1,w_4)$ or $(u_3,w_3)$.
It follows that $w_6$ must be in $R_5^{(2)}$. In this case, the edge $(u_1,w_6)$ has to cross $(u_2,w_3)$. Also, the edge $(u_3,w_6)$ has to cross $(u_2,w_4)$, and the edge $(u_2,w_6)$ must be crossing-freee; see dashed-dotted blue edges in Fig~\ref{fig:fanCrossingFree:general:drawingBE}).
Now, observe that $R_5$ cannot contain any other vertex among those that are not in the drawing constructed so far, since the edges incident to this vertex would cross exactly the same edges as the edges incident to $w_6$, which inevitably yield fan-crossings. Since these vertices cannot also lie in $R_6$ (as we consider this case earlier), we conclude that $\Gamma_2$ cannot be a subdrawing of a fan-crossing free drawing of $K_{3,7}$, when $u_3$ is in $R_6$.

Our analysis for vertex $u_3$ suggests that $\Gamma_2$ cannot be a subdrawing of a fan-crossing free drawing of $K_{3,7}$.

\medskip

\noindent\textbf{The drawing $\Gamma_3$.}
In drawing $\Gamma_3$ we have seven regions $R_1,\ldots,R_7$ as in Fig.~\ref{fig:fanCrossingFree:drawingC}.
As proved in Section~\ref{se:k55}, vertex $u_3$ can lie neither in regions $R_3,R_5$ (Lemma~\ref{obs:1}), nor in regions $R_6,R_7$ (Lemma~\ref{obs:3}).
Thus, vertex $u_3$ must lie in one of the following regions: $R_2$, $R_4$ or $R_1$. We consider each of these case separately.  

We first consider the case in which $u_3$ is in $R_4$. By Lemma~\ref{obs:3}, $w_6$ must be either in $R_4$ or in $R_5$. Both cases are not possible, because edge $(u_1,w_6)$ inevitably yields a fan-crossing (in fact, this edge cannot cross any of the edges $(u_1,w_2)$, $(u_2,w_1)$, or $(u_2,w_2)$). This rules out the case in which $u_3$ is in $R_4$.

Next, we consider the case in which $u_3$ is in $R_2$. By Lemma~\ref{obs:3}, $w_6$ must be either in $R_2$ or in $R_3$. Both cases are not possible, because edge $(u_2,w_6)$ inevitably yields a fan-crossing (in fact, this edge cannot cross any of the edges $(u_2,w_2)$, $(u_1,w_3)$, or $(u_1,w_2)$). So, the case in which $u_3$ is in $R_2$ is also ruled out.

Finally, we consider the case in which  $u_3$ is in $R_1$. By Lemma~\ref{obs:3} (for $R=R_1$ [$R=R_7$],
$w_i=w_2$, $w_j=w_3$ and $w_k=w_5$ [$w_k=w_4$]), it follows that the edge $(u_3,w_2)$ has to cross the pair of edges $(u_1,w_5)$ and $(u_1,w_4)$, which yields a fan-crossing. As a consequence, $u_3$ cannot be in $R_1$.

Our case analysis on $u_3$ implies that $\Gamma_3$ cannot be a subdrawing of a fan-crossing free drawing of $K_{3,7}$.

\medskip

\noindent\textbf{The drawing $\Gamma_4$.}
Next, we consider drawing $\Gamma_4$, which has seven regions $R_1,\ldots,R_7$ as in Fig.~\ref{fig:fanCrossingFree:drawingD}.
As proved in Section~\ref{se:k55}, vertex $u_3$ can lie neither in regions $R_2,R_3,R_5,R_6$ (Lemma~\ref{obs:1}), nor in region $R_4$ (Lemma~\ref{obs:3}).
This implies that $u_3$ is either in $R_1$ or in $R_7$. 
Suppose that $u_3$ is in $R_7$; the case where $u_3$ is in $R_1$ is symmetric. By Lemma~\ref{obs:3},
the edges $(u_3,w_1)$ and $(u_3,w_2)$ cross $(u_1,w_5)$ and $(u_2,w_5)$, respectively; see Fig.~\ref{fig:fanCrossingFree:general:drawingDB}. It follows that the edges $(u_3,w_3)$, $(u_3,w_4)$ and $(u_3,w_5)$ must be crossing-free.

\begin{figure}[tb]
	\centering
	 \subcaptionbox{\label{fig:fanCrossingFree:general:drawingDB}}{
	 \includegraphics[page=6, width=0.48\textwidth]{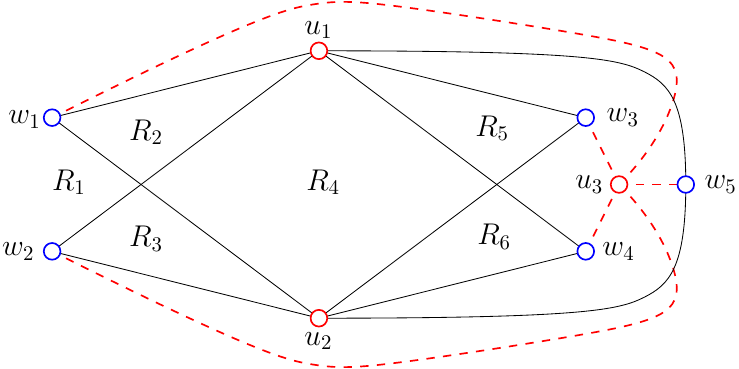}}
	 \hfil
	 \subcaptionbox{\label{fig:fanCrossingFree:general:drawingDC}}{
	 \includegraphics[page=7, width=0.48\textwidth]{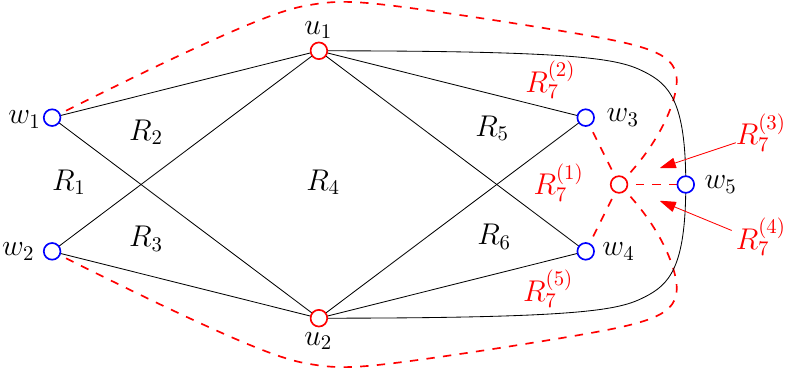}}
	\caption{
	(a)~$u_3$ is in $R_7$ of $\Gamma_4$ and 
	(b)~$R_7$ is partitioned into 5 subregions $R_7^{(1)}, \ldots R_7^{(5)}$.}
\end{figure}

Next, we argue for vertex $w_6\in V_2$. It is immediate to see that this vertex cannot lie in one of the regions $R_2,R_3,R_5$ and $R_6$, as otherwise either the edge $(u_1,w_6)$ or the edge $(u_2,w_6)$ yields a fan-crossing. We next claim  that $w_6$ is not in $R_1$, as well. Indeed, if $w_6$ were in $R_1$, then the edge $(u_3,w_6)$ would create a fan-crossing. It follows that $w_6$ is either in $R_4$ or in $R_7$.

\begin{figure}[b]
	\centering
	 \subcaptionbox{\label{fig:fanCrossingFree:general:drawingDD}}{
	 \includegraphics[page=8, width=0.48\textwidth]{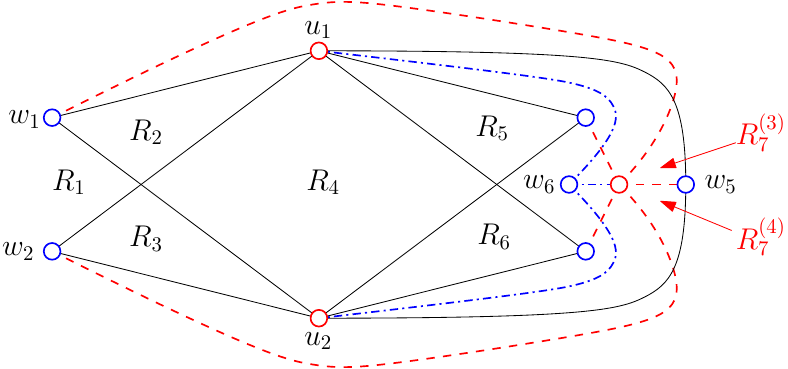}}
	 \hfil
	 \subcaptionbox{\label{fig:fanCrossingFree:general:drawingDE}}{
	 \includegraphics[page=9, width=0.48\textwidth]{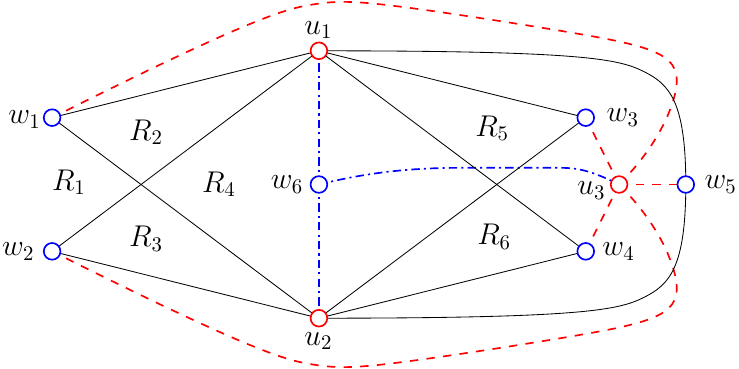}}
	\caption{Illustration of the cases in which
	(a)~$w_6$ is in $R_7^{(1)}$ and 
	(b)~$w_6$ is in $R_4$.}
\end{figure}

First, we consider the case in which $w_6$ is in $R_7$. Observe that the edges incident to $u_3$ partition the region $R_7$ into $5$ new regions, denoted by $R_7^{(1)}, \ldots R_7^{(5)}$; see Fig.~\ref{fig:fanCrossingFree:general:drawingDC}.
If $w_6$ is in $R_7^{(2)}$, then it is easy to both two edges $(u_1,w_6)$ and $(u_3,w_6)$ can be present, this is not true for the edge $(u_2,w_6)$, which inevitably yields a fan-crossing; in fact, the edge $(u_2,w_6)$ cannot cross an edge incident to $u_2$, or $(u_3,w_2)$, or $(u_1,w_4)$, or two edges incident to $u_1$. This already implies that $(u_2,w_6)$ must cross $(u_3,w_4)$ and $(u_3,w_3)$, which is not allowed. So, $w_6$ cannot be in $R_7^{(2)}$. By symmetry, $w_6$ cannot be in $R_7^{(5)}$ either.
If $w_6$ is in $R_7^{(3)}$, then the edge $(u_1,w_6)$ inevitably yields a fan-crossing; indeed, $(u_1,w_6)$ cannot cross
an edge incident to $u_1$ (especially the edge $(u_1,w_5)$), or the edge $(u_3,w_1)$, nor two edges incident to $w_5$ (namely, the edges $(u_3,w_5)$ and $(u_2,w_5)$), or two edges incident to $u_3$ (namely, the edges $(u_3,w_5)$ and $(u_2,w_2)$. Hence, $w_6$ can be neither in $R_7^{(3)}$ nor in $R_7^{(4)}$ (by symmetry). 
Fig.~\ref{fig:fanCrossingFree:general:drawingDD} illustrates the case in which $w_6$ is in $R_7^{(1)}$. In this case, we first observe that $(u_1,w_6)$ must cross $(u_3,w_3)$. Also, $(u_2,w_6)$ must cross $(u_3,w_4)$. Hence, $(u_3,w_6)$ must be crossing-free. Since it is easy to see that $w_7$ can be neither in $R_7$ nor in $R_4$, the case in which $w_6$ is in $R_7$ is ruled out.

To complete the analysis of drawing $\Gamma_4$, it remains to consider the case, in which $w_6$ is in $R_4$. In this case, the edge $(u_3,w_6)$ has to cross both edges $(u_1,w_4)$ and $(u_2,w_3)$. In addition, both edges $(u_1,w_6)$ and $(u_2,w_6)$ must be crossing-free; see dashed-dotted blue edges in Fig.~\ref{fig:fanCrossingFree:general:drawingDE}.
We next argue for $w_7$, which by symmetry can only be in $R_4$. However, in this case, $(u_3,w_7)$ would yield a fan-crossing. 
Since $w_6$ can be neither in $R_7$ nor in $R_4$, we conclude that the drawing $\Gamma_4$ cannot be a subdrawing of a fan-crossing free drawing of $K_{3,7}$.

\medskip

\noindent\textbf{The drawing $\Gamma_5$.}
In drawing $\Gamma_5$, we have eight regions $R_1,\ldots,R_8$ as in Fig.~\ref{fig:fanCrossingFree:drawingE}.
As proved in Section~\ref{se:k55}, vertex $u_3$ can lie neither in regions $R_3,R_5,R_7,R_8$ (Lemma~\ref{obs:1}), nor in regions $R_1,R_6$ (Lemma~\ref{obs:3}).
It follows that vertex $u_3$ lies either in $R_2$ or in $R_4$. 

Suppose first that $u_3$ lies in $R_4$. We next argue for vertex $w_6\in V_2$. By Lemma~\ref{obs:3}, vertex $w_6$ must be in $R_4$ or $R_5$. Both cases, however, are not possible, as the edge $(u_1,w_6)$ inevitably yields a fan-crossing; in fact, this edge cannot cross any of the edges $(u_1,w_2)$, $(u_2,w_1)$, and $(u_2,w_2)$, and as a results when $u_3$ is in $R_4$, a fan-crossing is always yielded.

Suppose now that $u_3$ is in $R_2$. Again we continue by arguing for $w_6$. By Lemma~\ref{obs:3}, this vertex must be either in $R_2$ or $R_3$. However, again both cases are not possible, as the edge $(u_2,w_6)$ inevitably yields a fan-crossing; in fact, this edge is cannot cross any of the edges $(u_2,w_2)$, $(u_1,w_3)$, and $(u_1,w_2)$, and as a result when $u_3$ is in $R_2$, a fan-crossing is always yielded.
We conclude that the drawing $\Gamma_5$ cannot be a subdrawing of a fan-crossing free drawing of $K_{3,7}$.

\medskip

\noindent\textbf{The drawing $\Gamma_6$.}
In drawing $\Gamma_6$, we have seven regions $R_1,\ldots,R_7$ as in Fig.~\ref{fig:fanCrossingFree:drawingF}.
As proved in Section~\ref{se:k55}, vertex $u_3$ can lie neither in regions $R_2,R_3,R_5,R_6$ (Lemma~\ref{obs:1}), nor in regions $R_1,R_4,R_7$ (Lemma~\ref{obs:3}).
Hence, in this case we can directly conclude that the drawing $\Gamma_5$ cannot be a subdrawing of a fan-crossing free drawing of $K_{3,7}$.

\medskip

\noindent\textbf{The drawing $\Gamma_7$.}
In drawing $\Gamma_7$, we have seven regions $R_1,\ldots,R_7$ as in Fig.~\ref{fig:fanCrossingFree:drawingG}.
As proved in Section~\ref{se:k55}, vertex $u_3$ can lie neither in regions $R_2,R_3,R_6,R_7$ (Lemma~\ref{obs:1}), nor in regions $R_4,R_5$ (Lemma~\ref{obs:3}).

Hence, vertex $u_3$ can only be in $R_1$. The edge $(u_3,w_3)$ can be drawn without introducing fan-crossings only if it crosses either both edges $(u_1,w_2)$ and $(u_2,w_1)$, or both edges $(u_1,w_5)$ and $(u_1,w_4)$. Assume w.l.o.g.~the former. Then, each of the edges $(u_3,w_1)$, $(u_3,w_2)$, $(u_3,w_4)$ and $(u_3,w_5)$ must be crossing-free; see
Fig.~\ref{fig:fanCrossingFree:general:drawingGB}.
We next argue for $w_6\in V_2$. We first observe that $w_6$ cannot lie in one of the regions $R_2,R_3,R_6$ or $R_7$, since
then the edge $(u_2,w_6)$ (if $w_6$ is in $R_2$ or $R_3$) or the edge $(u_1,w_6)$ (if $w_6$ is in $R_6$ or $R_7$) would yield fan-crossings. Also, $w_6$ is not in $R_4$, as otherwise $(u_3,w_6)$ would yield fan-crossings.
Thus, $w_6$ must be either in $R_1$ or in $R_5$ (by symmetry, $w_7$ must be either in $R_1$ or in $R_5$).

\begin{figure}[tb]
	\centering
	 \subcaptionbox{\label{fig:fanCrossingFree:general:drawingGB}}{
	 \includegraphics[page=10, width=0.48\textwidth]{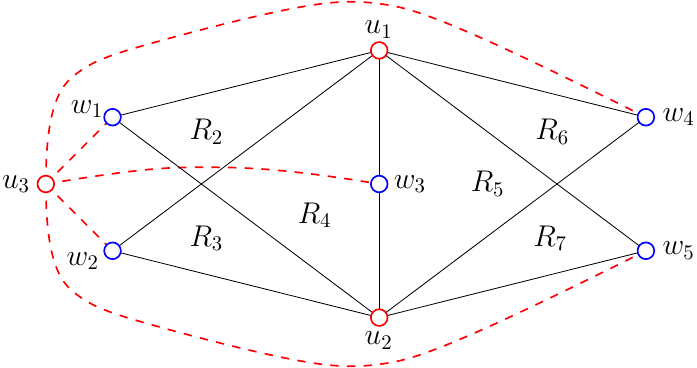}}
	 \hfil
	 \subcaptionbox{\label{fig:fanCrossingFree:general:drawingGD}}{
	 \includegraphics[page=11, width=0.48\textwidth]{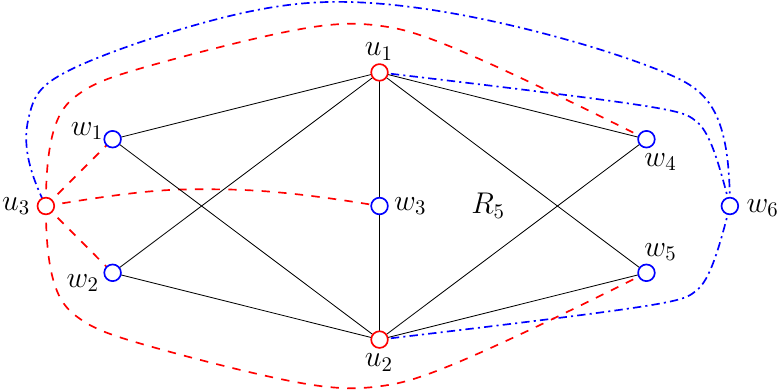}}
	\caption{Illustration of the cases in which
	(a)~$u_3$ is in $R_1$ and 
	(b)~$w_6$ is in $R_1$.}
\end{figure}

If $w_6\in V_2$ is in $R_1$, then it must be placed in the outer face of the drawing containing $u_3$, as otherwise one of the edges $(u_1,w_6)$ or $(u_2,w_6)$ would yield a fan-crossing. Further, the edge $(u_1,w_6)$ must cross $(u_3,w_4)$ and the edge $(u_2,w_6)$ must cross $(u_3,w_5)$. This implies that there $w_6$ is the only vertex in $R_1$, and the edge $(u_3,w_6)$ must be crossing-free; see Fig.~\ref{fig:fanCrossingFree:general:drawingGD}.
As a result, $w_7$ is in $R_5$ (recall that $w_7$ can either be in $R_1$ or in $R_5$). In this case, $(u_3,w_7)$ has to cross both edges $(u_1,w_5)$ and $(u_2,w_4)$, and additionally one of the edges $(u_1,w_6)$ or $(u_2,w_6)$, which yields a a fan-crossing. This rules out the case, in which $w_6\in V_2$ is in $R_1$.
	
To complete the case analysis, it remains to consider the case in which $w_6$ is in $R_5$. In this case, $(u_3,v)$ must cross both edges $(u_1,w_5)$ and $(u_2,w_4)$. This already implies that there $w_6$ is the only vertex in $R_5$ and that each of the edges $(u_1,w_6)$ and $(u_2,w_6)$ must be drawn crossing-free; see Fig.~\ref{fig:fanCrossingFree:general:drawingGE}. However, in this case $w_7$ must be in $R_1$, which yields a case that we have already ruled out.
We conclude that the drawing $\Gamma_7$ cannot be a subdrawing of a fan-crossing free drawing of $K_{3,7}$, which concludes the proof of Characterization~\ref{th:bipartite:fcf}.

\begin{figure}[h]
	\centering
	\includegraphics[page=12, width=0.49\textwidth]{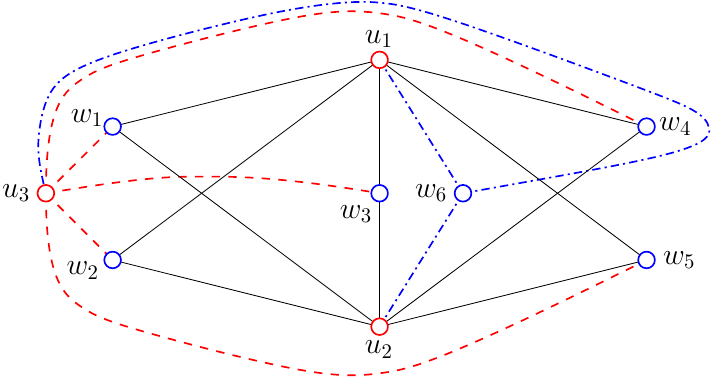}
	\caption{$w_6$ is in $R_5$.}
	\label{fig:fanCrossingFree:general:drawingGE}
\end{figure}

}
\end{document}